\documentclass{article}
\usepackage{color}
 \usepackage{url}
\usepackage{bm}
\usepackage{amsmath}
\usepackage{amssymb}\usepackage{amsthm}
\usepackage{amsfonts}
\usepackage{amscd}
\begin{document}
\theoremstyle{plain}
\newtheorem*{ithm}{Theorem}
\newtheorem*{idefn}{Definition}
\newtheorem{thm}{Theorem}[section]
\newtheorem{lem}[thm]{Lemma}
\newtheorem{dlem}[thm]{Lemma/Definition}
\newtheorem{prop}[thm]{Proposition}
\newtheorem{set}[thm]{Setting}
\newtheorem{cor}[thm]{Corollary}
\newtheorem*{icor}{Corollary}
\theoremstyle{definition}
\newtheorem{assum}[thm]{Assumption}
\newtheorem{notation}[thm]{Notation}
\newtheorem{defn}[thm]{Definition}
\newtheorem{clm}[thm]{Claim}
\newtheorem{ex}[thm]{Example}
\theoremstyle{remark}
\newtheorem{rem}[thm]{Remark}
\numberwithin{equation}{section}

\newcommand{\unit}{\mathbb I}
\newcommand{\ali}[1]{{\mathfrak A}_{[ #1 ,\infty)}}
\newcommand{\alm}[1]{{\mathfrak A}_{(-\infty, #1 ]}}
\newcommand{\nn}[1]{\lV #1 \rV}
\newcommand{\br}{{\mathbb R}}
\newcommand{\dm}{{\rm dom}\mu}
\newcommand{\inn}{({\rm {inner}})}
\newcommand{\Ad}{\mathop{\mathrm{Ad}}\nolimits}
\newcommand{\Proj}{\mathop{\mathrm{Proj}}\nolimits}
\newcommand{\RRe}{\mathop{\mathrm{Re}}\nolimits}
\newcommand{\RIm}{\mathop{\mathrm{Im}}\nolimits}
\newcommand{\Wo}{\mathop{\mathrm{Wo}}\nolimits}
\newcommand{\Prim}{\mathop{\mathrm{Prim}_1}\nolimits}
\newcommand{\Primz}{\mathop{\mathrm{Prim}}\nolimits}
\newcommand{\ClassA}{\mathop{\mathrm{ClassA}}\nolimits}
\newcommand{\Class}{\mathop{\mathrm{Class}}\nolimits}
\newcommand{\diam}{\mathop{\mathrm{diam}}\nolimits}
\def\qed{{\unskip\nobreak\hfil\penalty50
\hskip2em\hbox{}\nobreak\hfil$\square$
\parfillskip=0pt \finalhyphendemerits=0\par}\medskip}
\def\proof{\trivlist \item[\hskip \labelsep{\bf Proof.\ }]}
\def\endproof{\null\hfill\qed\endtrivlist\noindent}
\def\proofof[#1]{\trivlist \item[\hskip \labelsep{\bf Proof of #1.\ }]}
\def\endproofof{\null\hfill\qed\endtrivlist\noindent}

\newcommand{\mkA}{{\mathfrak A}}
\newcommand{\mkB}{{\mathfrak B}}
\newcommand{\mkC}{{\mathfrak C}}
\newcommand{\mkD}{{\mathfrak D}}
\newcommand{\varphii}{\varphi}

\newcommand{\pg}{{\mathfrak S}(\bbZ^2)}
\newcommand{\oo}{{\boldsymbol\varphii}}
\newcommand{\caA}{{\mathcal A}}
\newcommand{\caB}{{\mathcal B}}
\newcommand{\caC}{{\mathcal C}}
\newcommand{\caD}{{\mathcal D}}
\newcommand{\caE}{{\mathcal E}}
\newcommand{\caF}{{\mathcal F}}
\newcommand{\caG}{{\mathcal G}}
\newcommand{\caH}{{\mathcal H}}
\newcommand{\caI}{{\mathcal I}}
\newcommand{\caJ}{{\mathcal J}}
\newcommand{\caK}{{\mathcal K}}
\newcommand{\caL}{{\mathcal L}}
\newcommand{\caM}{{\mathcal M}}
\newcommand{\caN}{{\mathcal N}}
\newcommand{\caO}{{\mathcal O}}
\newcommand{\caP}{{\mathcal P}}
\newcommand{\caQ}{{\mathcal Q}}
\newcommand{\caR}{{\mathcal R}}
\newcommand{\caS}{{\mathcal S}}
\newcommand{\caT}{{\mathcal T}}
\newcommand{\caU}{{\mathcal U}}
\newcommand{\caV}{{\mathcal V}}
\newcommand{\caW}{{\mathcal W}}
\newcommand{\caX}{{\mathcal X}}
\newcommand{\caY}{{\mathcal Y}}
\newcommand{\caZ}{{\mathcal Z}}
\newcommand{\bbA}{{\mathbb A}}
\newcommand{\bbB}{{\mathbb B}}
\newcommand{\bbC}{{\mathbb C}}
\newcommand{\bbD}{{\mathbb D}}
\newcommand{\bbE}{{\mathbb E}}
\newcommand{\bbF}{{\mathbb F}}
\newcommand{\bbG}{{\mathbb G}}
\newcommand{\bbH}{{\mathbb H}}
\newcommand{\bbI}{{\mathbb I}}
\newcommand{\bbJ}{{\mathbb J}}
\newcommand{\bbK}{{\mathbb K}}
\newcommand{\bbL}{{\mathbb L}}
\newcommand{\bbM}{{\mathbb M}}
\newcommand{\bbN}{{\mathbb N}}
\newcommand{\bbO}{{\mathbb O}}
\newcommand{\bbP}{{\mathbb P}}
\newcommand{\bbQ}{{\mathbb Q}}
\newcommand{\bbR}{{\mathbb R}}
\newcommand{\bbS}{{\mathbb S}}
\newcommand{\bbT}{{\mathbb T}}
\newcommand{\bbU}{{\mathbb U}}
\newcommand{\bbV}{{\mathbb V}}
\newcommand{\bbW}{{\mathbb W}}
\newcommand{\bbX}{{\mathbb X}}
\newcommand{\bbY}{{\mathbb Y}}
\newcommand{\bbZ}{{\mathbb Z}}
\newcommand{\str}{^*}
\newcommand{\lv}{\left \vert}
\newcommand{\rv}{\right \vert}
\newcommand{\lV}{\left \Vert}
\newcommand{\rV}{\right \Vert}
\newcommand{\la}{\left \langle}
\newcommand{\ra}{\right \rangle}
\newcommand{\ltm}{\left \{}
\newcommand{\rtm}{\right \}}
\newcommand{\lcm}{\left [}
\newcommand{\rcm}{\right ]}
\newcommand{\ket}[1]{\lv #1 \ra}
\newcommand{\bra}[1]{\la #1 \rv}
\newcommand{\lmk}{\left (}
\newcommand{\rmk}{\right )}
\newcommand{\al}{{\mathcal A}}
\newcommand{\md}{M_d({\mathbb C})}
\newcommand{\eaut}{\mathop{\mathrm{EAut}}\nolimits}
\newcommand{\qaut}{\mathop{\mathrm{QAut}}\nolimits}
\newcommand{\sqaut}{\mathop{\mathrm{SQAut}}\nolimits}
\newcommand{\gsqaut}{\mathop{\mathrm{GSQAut}}\nolimits}
\newcommand{\QLS}{\mathop{\mathcal{SL}}\nolimits}
\newcommand{\haut}{\mathop{\mathrm{HAut}}\nolimits}
\newcommand{\guaut}{\mathop{\mathrm{GUQAut}}\nolimits}
\newcommand{\IG}{\mathop{\mathrm{IG}}\nolimits}
\newcommand{\IP}{\mathop{\mathrm{IP}}\nolimits}
\newcommand{\ainn}{\mathop{\mathrm{AInn}}\nolimits}
\newcommand{\id}{\mathop{\mathrm{id}}\nolimits}
\newcommand{\Tr}{\mathop{\mathrm{Tr}}\nolimits}
\newcommand{\co}{\mathop{\mathrm{co}}\nolimits}
\newcommand{\sym}{\mathop{\mathrm{Sym}}\nolimits}
\newcommand{\sgn}{\mathop{\mathrm{sgn}}\nolimits}
\newcommand{\Ran}{\mathop{\mathrm{Ran}}\nolimits}
\newcommand{\Ker}{\mathop{\mathrm{Ker}}\nolimits}
\newcommand{\Aut}{\mathop{\mathrm{Aut}}\nolimits}
\newcommand{\Inn}{\mathop{\mathrm{Inn}}\nolimits}
\newcommand{\spn}{\mathop{\mathrm{span}}\nolimits}
\newcommand{\Mat}{\mathop{\mathrm{M}}\nolimits}
\newcommand{\Dia}{\mathop{\mathrm{D}}\nolimits}
\newcommand{\UT}{\mathop{\mathrm{UT}}\nolimits}
\newcommand{\DT}{\mathop{\mathrm{DT}}\nolimits}
\newcommand{\GL}{\mathop{\mathrm{GL}}\nolimits}
\newcommand{\spa}{\mathop{\mathrm{span}}\nolimits}
\newcommand{\supp}{\mathop{\mathrm{supp}}\nolimits}
\newcommand{\rank}{\mathop{\mathrm{rank}}\nolimits}
\newcommand{\idd}{\mathop{\mathrm{id}}\nolimits}
\newcommand{\ran}{\mathop{\mathrm{Ran}}\nolimits}
\newcommand{\dr}{ \mathop{\mathrm{d}_{{\mathbb R}^k}}\nolimits} 
\newcommand{\dc}{ \mathop{\mathrm{d}_{\cc}}\nolimits} \newcommand{\drr}{ \mathop{\mathrm{d}_{\rr}}\nolimits} 
\newcommand{\zin}{\mathbb{Z}}
\newcommand{\rr}{\mathbb{R}}
\newcommand{\cc}{\mathbb{C}}
\newcommand{\ww}{\mathbb{W}}
\newcommand{\nan}{\mathbb{N}}\newcommand{\bb}{\mathbb{B}}
\newcommand{\aaa}{\mathbb{A}}\newcommand{\ee}{\mathbb{E}}
\newcommand{\pp}{\mathbb{P}}
\newcommand{\wks}{\mathop{\mathrm{wk^*-}}\nolimits}
\newcommand{\Hom}{\mathop{\mathrm{Hom}}\nolimits}
\newcommand{\mk}{{\Mat_k}}
\newcommand{\mnz}{\Mat_{n_0}}
\newcommand{\mn}{\Mat_{n}}
\newcommand{\dist}{\mathrm{d}}
\newcommand{\braket}[2]{\left\langle#1,#2\right\rangle}
\newcommand{\ketbra}[2]{\left\vert #1\right \rangle \left\langle #2\right\vert}
\newcommand{\abs}[1]{\left\vert#1\right\vert}
\newtheorem{nota}{Notation}[section]
\def\qed{{\unskip\nobreak\hfil\penalty50
\hskip2em\hbox{}\nobreak\hfil$\square$
\parfillskip=0pt \finalhyphendemerits=0\par}\medskip}
\def\proof{\trivlist \item[\hskip \labelsep{\bf Proof.\ }]}
\def\endproof{\null\hfill\qed\endtrivlist\noindent}
\def\proofof[#1]{\trivlist \item[\hskip \labelsep{\bf Proof of #1.\ }]}
\def\endproofof{\null\hfill\qed\endtrivlist\noindent}
%%%%%%%%%%%%%%%%%%%%%%%%%%%%%%
\newcommand{\ZZ}{\bbZ_2\times\bbZ_2}
\newcommand{\SSS}{\mathcal{S}}
\newcommand{\cs}{S}
\newcommand{\ct}{t}
\newcommand{\hS}{S}
\newcommand{\vv}{{\boldsymbol v}}
\newcommand{\ala}{a}
\newcommand{\bet}{b}
\newcommand{\gam}{c}
\newcommand{\alphas}{\alpha}
\newcommand{\alphai}{\alpha^{(\sigma_{1})}}
\newcommand{\alphan}{\alpha^{(\sigma_{2})}}
\newcommand{\betas}{\beta}
\newcommand{\betai}{\beta^{(\sigma_{1})}}
\newcommand{\betan}{\beta^{(\sigma_{2})}}
\newcommand{\alphass}{\alpha^{{(\sigma)}}}
\newcommand{\uu}{V}
\newcommand{\vp}{\varsigma}
\newcommand{\vpr}{R}
\newcommand{\tg}{\tau_{\Gamma}}
\newcommand{\sgg}{\Sigma_{\Gamma}^{(\sigma)}}
\newcommand{\nh}{1}
\newcommand{\rk}{2,a}
\newcommand{\nii}{1,a}
\newcommand{\nhh}{3,a}
\newcommand{\sjt}{2}
\newcommand{\sjtg}{2}
\newcommand{\bcg}{\caB(\caH_{\alpha})\otimes  C^{*}(\Sigma_{\Gamma}^{(\sigma)})}
\title{A $H^{3}(G,\bbT)$-valued index of symmetry protected topological phases with on-site finite group symmetry for two-dimensional quantum spin systems}

\author{Yoshiko Ogata \thanks{ Graduate School of Mathematical Sciences
The University of Tokyo, Komaba, Tokyo, 153-8914, Japan
Supported in part by
the Grants-in-Aid for
Scientific Research, JSPS.}}
\maketitle

\begin{abstract}
We consider SPT-phases with on-site finite group $G$ symmetry $\beta$ for two-dimensional quantum spin systems.
We show that  they
have $H^{3}(G,\bbT)$-valued invariant.
\end{abstract}

\section{Introduction }
The notion of symmetry protected topological (SPT) phases was introduced by Gu and Wen [GW].
It is defined as follows:
we consider the set of all Hamiltonians with some symmetry, 
which have a unique gapped ground state in the bulk, and can be smoothly deformed into
a common trivial gapped Hamiltonian without closing the gap.
We say 
two such Hamiltonians are equivalent, if they can be smoothly deformed into
each other, without breaking the symmetry.
We call an equivalence class of this classification, a
symmetry protected topological (SPT) phase.
Based on tensor network or TQFT analysis, \cite{cglw}, \cite{molnar}
 it is conjectured that 
SPT phases with on-site finite group $G$ symmetry for $\nu$-dimensional quantum spin systems
 have a $H^{\nu+1}(G,\bbT)$-valued invariant.
 We proved their conjecture affirmatively in \cite{ogata} for $\nu=1$.
 In this paper, We 
 show that the conjecture is also true for $\nu=2$.
 
 %\subsection{Two-dimensional quantum spin systems}\label{setting}

We start by summarizing standard setup of $2$-dimensional quantum spin systems on the two dimensional lattice $\bbZ^{2}$ \cite{BR1,BR2}. We will use freely the basic notation in section \ref{notasec}.
Throughout this paper, we fix some $2\le d\in\nan$.
We denote the algebra of $d\times d$ matrices by $\Mat_{d}$.

For each subset $\Gamma$ of $\bbZ^2$,
we denote the set of all finite subsets in $\Gamma$ by ${\mathfrak S}_{\Gamma}$.
We introduce the Euclidean metric on $\bbZ^2$, inherited from $\bbR^2$.
We denote by $\dist(S_1,S_2)$ the distance between $S_1,S_2\subset \bbZ^2$. 
For a subset $\Gamma$ of $\bbZ^2$ and $r\in\bbR_{\ge 0}$, 
$\hat \Gamma{(r)}$ denotes the all points in $\bbR^{2}$ whose distance from $\Gamma$ is less than or equal to $r$.
We also set $\Gamma(r):=\hat \Gamma{(r)}\cap\bbZ^{2}$.
When we take a complement of $\Gamma\subset\bbZ^{2}$,
it means $\Gamma^{c}:=\bbZ^{2}\setminus \Gamma$.
For each $n\in\bbN$, we denote $[-n,n]^2\cap \bbZ^2$ by $\Lambda_n$.

For each $z\in\bbZ^2$,  let $\caA_{\{z\}}$ be an isomorphic copy of $\Mat_{d}$, and for any finite subset $\Lambda\subset\bbZ^2$, we set $\caA_{\Lambda} = \bigotimes_{z\in\Lambda}\caA_{\{z\}}$.
%, which is the local algebra of observables in $\Lambda$. 
For finite $\Lambda$, the algebra $\caA_{\Lambda} $ can be regarded as the set of all bounded operators acting on
the Hilbert space $\bigotimes_{z\in\Lambda}{\bbC}^{d}$.
We use this identification freely.
If $\Lambda_1\subset\Lambda_2$, the algebra $\caA_{\Lambda_1}$ is naturally embedded in $\caA_{\Lambda_2}$ by tensoring its elements with the identity. 
For an infinite subset $\Gamma\subset \bbZ^{2}$,
$\caA_{\Gamma}$
is given as the inductive limit of the algebras $\caA_{\Lambda}$ with $\Lambda\in{\mathfrak S}_{\Gamma}$.
We call $\caA_{\Gamma}$ the quantum spin system on $\Gamma$.
For simplicity we denote the two dimensional quantum spin system $\caA_{\bbZ^{2}}$
by $\caA$.
We also set $\caA_{\rm loc}:=\bigcup_{\Lambda\in{\mathfrak S}_{\bbZ^{2}}}\caA_{\Lambda}
$.
%We denote by $\beta_x$ the automorphisms on $\caA$ representing the space translation by  $x\in\bbZ$.
%By $Q^{(j)}$, $j\in\bbZ$, we denote the element of $\caA$ with $Q\in\Mat_d$ in the $j$-th component of the tensor product of $\caA$ and the unit in any other component.
For a subset $\Gamma_1$ of $\Gamma\subset\bbZ^{2}$,
the algebra $\caA_{\Gamma_1}$ can be regarded as a subalgebra of $\caA_{\Gamma}$.
With this identification, for $A\in\caA_{\Gamma_1}$, we occasionally use the same symbol 
$A$ to denote $A\otimes \unit_{\caA_{\Gamma\setminus\Gamma_1}}\in \caA_\Gamma$.
Similarly, an automorphism $\gamma$ on $\caA_{\Gamma_1}$ can be naturally regarded as an automorphism
$\gamma\otimes \id_{\caA_{\Gamma\setminus \Gamma_1}}$ on $\caA_{\Gamma}$.
We use this identification freely and with a slight abuse of notation
we occasionally denote $\gamma\otimes \id_{\caA_{\Gamma\setminus \Gamma_1}}$ by $\gamma$.
Similarly, for disjoint $\Gamma_-,\Gamma_+\subset\bbZ^{2}$ and $\alpha_\pm\in\Aut\caA_{\Gamma\pm}$, 
we occasionally write $\alpha_-\otimes\alpha_+$ to denote 
$\lmk \alpha_-\otimes\id_{\Gamma_-^c}\rmk
\lmk \alpha_+\otimes\id_{\Gamma_+^c}\rmk$,
under the above identification.

Throughout this paper we fix a finite group $G$ and its  unitary representation $U$ on $\bbC^{d}$.
Let $\Gamma\subset \bbZ^{2}$ be a non-empty subset.
For each $g\in G$, there exists a unique automorphism $\beta^{\Gamma}$ on $\caA_{\Gamma}$
such that 
\begin{align}\label{tgg}
\beta_{g}^{\Gamma}\lmk A\rmk=\Ad\lmk\bigotimes_{I} U(g)\rmk\lmk A\rmk,\quad A\in\caA_{I},\quad g\in G,
\end{align}
for any finite subset $I$ of $\Gamma$.
We call the group homomorphism $\beta^{\Gamma}: G\to \Aut \caA_{\Gamma}$, 
the on-site action of $G$ 
on $\caA_{\Gamma}$ given by $U$.
For simplicity, we denote $\beta^{\bbZ^{2}}_{g}$ by $\beta_{g}$.

%\subsection{SPT-phases}
A mathematical model of a quantum spin system is fully specified by its interaction $\Phi$.
A uniformly bounded interaction on $\caA$ 
is a map $\Phi: {\mathfrak S}_{\bbZ^{2}}\to \caA_{\rm loc}$ such that
\begin{align}
\Phi(X)=\Phi(X)^*\in \caA_{X},\quad X\in {\mathfrak S}_{\bbZ^{2}},
\end{align}
and 
\begin{align}
\sup_{X\in {\mathfrak S}_{\bbZ^{2}}}\lV \Phi(X)\rV<\infty.
\end{align}
It is of finite range with interaction length less than or equal to $R\in\nan$ if 
$\Phi(X)=0$ for any $X\in {\mathfrak S}_{\bbZ^{2}}$
whose diameter is larger than $R$.
An on-site interaction, i.e., an interaction with
$\Phi(X)=0$ unless $X$ consists of a single point, is said to be trivial.
An interaction $\Phi$ is $\beta$-invariant if $\beta_g(\Phi(X))=\Phi(X)$
for any $X\in {\mathfrak S}_{\bbZ^{2}}$.
%We denote by $\Phi_n$ for each $n\in\nan$
%the interaction given by
%\begin{align}
%\Phi_n(X):=\left\{
%\begin{gathered}
%\Phi(X),\quad \text{if}\quad X\subset \Lambda_n,\\
%0,\quad \text{otherwise}.
%\end{gathered}
%\right.
%\end{align}
%
For a uniformly bounded and finite range interaction $\Phi$ and $\Lambda\in {\mathfrak S}_{\bbZ^{2}}$
define the local Hamiltonian
\begin{align}
\lmk H_\Phi\rmk_\Lambda
:=\sum_{X\subset\Lambda} \Phi(X),
\end{align}
and denote the dynamics
\begin{align}
\tau^{(\Lambda)\Phi}_t (A):=e^{it\lmk H_\Phi\rmk_\Lambda}Ae^{-it\lmk H_\Phi\rmk_\Lambda},
\quad t\in \bbR,\quad A\in\caA.
\end{align}
By the uniform boundedness and finite rangeness of $\Phi$, 
 for each $A\in\caA$, the following limit exists
\begin{align}
\lim_{\Lambda\to\bbZ^{\nu}} \tau^{(\Lambda),\Phi}_t\lmk
A\rmk=:
\tau^{\Phi}_t\lmk A\rmk,\quad t\in\bbR,
\end{align}
and defines the dynamics $\tau^{\Phi}$ on $\caA$.
(See \cite{BR2}.)
%Note that $\tau_{\Phi_n}=\tau_{\Phi, \Lambda_n}$.
%We denote by $\delta_\Phi$ the generator of $\tau_{\Phi}$.
%
For a uniformly bounded and finite range interaction $\Phi$,
a state $\varphi$ on $\caA$ is called a \mbox{$\tau^{\Phi}$-ground} state
if the inequality
$
-i\,\varphi(A^*{\delta_{\Phi}}(A))\ge 0
$
holds
for any element $A$ in the domain $\caD({\delta_{\Phi}})$ of the generator ${\delta_\Phi}$.
Let $\varphi$ be a $\tau^\Phi$-ground state, with a GNS triple $(\caH_\varphi,\pi_\varphi,\Omega_\varphi)$.
Then there exists a unique positive operator $H_{\varphi,\Phi}$ on $\caH_\varphi$ such that
$e^{itH_{\varphi,\Phi}}\pi_\varphi(A)\Omega_\varphi=\pi_\varphi(\tau^t_\Phi(A))\Omega_\varphi$,
for all $A\in\caA$ and $t\in\mathbb R$.
We call this $H_{\varphi,\Phi}$ the bulk Hamiltonian associated with $\varphi$.
\begin{defn}
We say that an interaction $\Phi$ has a unique gapped ground state if 
(i)~the $\tau^\Phi$-ground state, which we denote as $\omega_{\Phi}$, is unique, and 
(ii)~there exists a $\gamma>0$ such that
$\sigma(H_{\omega_{\Phi},\Phi})\setminus\{0\}\subset [\gamma,\infty)$, where  $\sigma(H_{\omega_{\Phi},\Phi})$ is the spectrum of $H_{\omega_{\Phi},\Phi}$.
We denote by $\caP_{UG}
$ the set of all uniformly bounded finite range 
interactions, with unique gapped ground state.
We denote by $\caP_{UG\beta}$ the set of all uniformly bounded finite range 
{\it $\beta$-invariant} interactions, with unique gapped ground state.
\end{defn}
In this paper we consider a classification problem of a subset of $\caP_{UG\beta}$, models with short range entanglement.
To describe the models with short range entanglement, we need 
to explain the classification problem of unique gapped ground state phases, without symmetry.
For $\Gamma\subset\bbZ^{2}$, we denote by  $\Pi_{\Gamma}:\caA\to \caA_{\Gamma}$ the conditional expectation with respect to the trace state.
Let $f:(0,\infty)\to (0,\infty)$ be a continuous decreasing function 
with $\lim_{t\to\infty}f(t)=0$.
For each $A\in\caA$, let
\begin{align}\label{dzeta}
\lV A\rV_f:=\lV A\rV
+ \sup_{N\in \nan}\lmk\frac{\lV
A-\Pi_{\Lambda_N}
(A)
\rV}
{f(N)}
\rmk.
\end{align}
We denote by $\caD_f$ the set of all $A\in\caA$ such that
$\lV A\rV_f<\infty$.

The classification of unique gapped ground state phases $\caP_{UG}
$ without symmetry is the following.
\begin{defn}\label{classsym}
Two interactions $\Phi_0,\Phi_1\in\caP_{UG}
$ are equivalent if there is
a path of interactions
 $\Phi : [0,1]\to \caP_{UG}
$
satisfying the following:
\begin{enumerate}
\item $\Phi(0)=\Phi_0$ and $\Phi(1)=\Phi_1$.
\item
For each $X\in{\mathfrak S}_{\bbZ^2}$, the map
$[0,1]\ni s\to \Phi(X;s)\in\caA_{X}$ is  $C^1$.
We denote by $\dot{\Phi}(X;s)$ 
%or $\dot{\Phi}(X,s)$ 
the corresponding derivatives.
The interaction obtained by differentiation is denoted by $\dot\Phi(s)$, for each $s\in[0,1]$.
\item
There is a number $R\in\nan$
such that $X \in {\mathfrak S}_{\bbZ^2}$ and $\diam{X}\ge R$ imply $\Phi(X;s)=0$, for all $s\in[0,1]$.
\item
Interactions are bounded as follows
\begin{align}
C_b^{\Phi}:=\sup_{s\in[0,1]}\sup_{X\in {\mathfrak S}_{\bbZ^2}}
\lmk
\lV
\Phi\lmk X;s\rmk
\rV+|X|\lV
\dot{\Phi} \lmk X;s\rmk
\rV
\rmk<\infty.
\end{align}
\item
Setting
\begin{align}
b(\varepsilon):=\sup_{Z\in{\mathfrak S}_{\bbZ^2}}
\sup_{s,s_0 \in[0,1],0<| s-s_0|<\varepsilon}
\lV
\frac{\Phi(Z;s)-\Phi(Z;s_0)}{s-s_0}-\dot{\Phi}(Z;s_0)
\rV
\end{align}
for each $\varepsilon>0$, we have
$\lim_{\varepsilon\to 0} b(\varepsilon)=0$. 
\item
There exists a $\gamma>0$ such that
$\sigma(H_{\omega_{\Phi(s)},\Phi(s)})\setminus\{0\}\subset [\gamma,\infty)$ for
all $s\in[0,1]$, where  $\sigma(H_{\omega_{\Phi(s)},\Phi(s)})$ is the spectrum of $H_{\omega_{\Phi(s)},\Phi(s)}$.
\item
There exists an $0<\eta<1$ satisfying the following:
Set $\zeta(t):=e^{-t^{ \eta}}$.
Then for each $A\in D_\zeta$, 
$\omega_{\Phi(s)}(A)$ is differentiable with respect to $s$, and there is a constant
$C_\zeta$ such that:
\begin{align}\label{dcon}
\lv
\frac{d}{ds}\omega_{\Phi(s)}(A)
\rv
\le C_\zeta\lV A\rV_\zeta,
\end{align}
for any $A\in D_\zeta$.(Recall (\ref{dzeta})).
\end{enumerate}
We write $\Phi_0\sim\Phi_1$ if
$\Phi_0$ and $\Phi_1$ are equivalent. 
If $\Phi_0,\Phi_1\in\caP_{UG\beta}$ and
if we can take the path in $\caP_{UG\beta}$, i.e.,
 so that $\beta_g\lmk \Phi(X;s)\rmk=\Phi(X;s)$, $g\in G$ for all $s\in[0,1]$,
then we say $\Phi_0$ and $\Phi_1$ are $\beta$-equivalent
and write  $\Phi_0\sim_\beta\Phi_1$.
\end{defn}
The object we classify in this paper is the following:
\begin{defn}
Fix a trivial interaction
$\Phi_0\in \caP_{UG}
$.
We denote by $\caP_{SL\beta}$ the set of all 
$\Phi\in\caP_{UG\beta}$
such that $\Phi\sim\Phi_0$.
Connected components of 
 $\caP_{SL\beta}$ with respect to
$\sim_\beta$ are the
symmetry protected topological (SPT)-phases.
\end{defn}
Because we have $\Phi_0\sim\tilde\Phi_0$ for any trivial $\Phi_0,\tilde\Phi_{0}\in \caP_{UG}
$, the set
$\caP_{SL\beta}$ does not depend on the choice of $\Phi_0$.
%Note that here, we are requiring $\Phi\sim\Phi_{0}$, not $\Phi\sim_{\beta}\Phi_{0}$.

%Now we are ready to define symmetry protected topological phases.
%\begin{defn}
%An interaction $\Phi\in \caP_{UF\beta}$ is in a symmetry protected topological (SPT) phase
%if 
%\begin{description}
%\item[(i)]
%$\Phi\in \caP_{SL\beta}$
%and
%\item[(ii)]
%there is no trivial $\Phi_1\in  \caP_{UF\beta}$ such that
%$\Phi\sim_\beta\Phi_1$.
%\end{description}
The main result of this paper is as follows.
\begin{thm}\label{main}
There is a $H^3(G,\bbT)$-valued index on $\caP_{SL\beta}$, which is an invariant of the classification
 $\sim_\beta$ of $\caP_{SL\beta}$.
\end{thm}

This paper is organized as follows.
In section
 \ref{h3index}, we define the $H^{3}(G,\bbT)$-valued index
for a class of states which are created from a fixed product state
via ``factorizable'' automorphisms, and  satisfying some additional condition.
This additional condition is the existence of the set of 
automorphisms 1.which do not move the state, and 2. are almost like 
$\beta$-action restricted to the upper half-plane except for some one-dimensional
perturbation.
In section
 \ref{tildebetasec}, we show that the existence of such set of
automorphisms are guaranteed in suitable situation.
Furthermore, in section
 \ref{stabilitysec}, we  show the stability of the index, i.e.,
suitably $\beta$-invariant automorphism does not
change this index.
Finally in section
 \ref{mainthproofsec}, we show our main Theorem
Theorem \ref{main}, showing that in our setting of Theorem \ref{main},
all the conditions required in section
  \ref{h3index}, \ref{tildebetasec}, \ref{stabilitysec}
are satisfied.
Although the index is defined in terms of GNS representations,
in some good situation,
we can calculate it without going through
GNS representation. This is shown in section
 \ref{autocasesec}.

The present result and the main idea of the proof were announced publicly on 15 December 2020 at IAMP One World Mathematical Physics Seminar (see you-tube video)\cite{IAMP}.
It was also presented in the international meeting {\it Theoretical studies of topological phases of matter}
on 17 December 2020, and
 Current Developments in Mathematics 4th January 2021 via zoom with a lecture note.
 \cite{IAMP}.

\section{The $H^{3}(G,\bbT)$-valued index in $2$-dimensional systems}\label{h3index}
In this section, we associate an $H^{3}(G,\bbT)$-index for some class of states.
It will turn out later that this class includes SPT phases.
\subsection{Definitions and the setting}\label{settingsubsec}
For $0<\theta<\frac\pi 2$, a cone $C_\theta$ is defined by
\begin{align}
C_\theta:=
\left\{
(x,y)\in\bbZ^2\mid
|y|\le \tan \theta\cdot |x|
\right\}.
\end{align}
%We also set $C_{\frac\pi 2}:=\bbZ^2$.
For $0<\theta_1<\theta_2\le \frac \pi 2$, we use a notation
$\caC_{(\theta_1,\theta_2]}:=C_{\theta_2}\setminus C_{\theta_1}$ and
$\caC_{[0,\theta_1]}:=C_{\theta_1}$.
%We define $\caC_{(\theta_{1},\theta_{2})}$ etc. analogously.
Left, right, upper, lower half planes are denoted by $H_L$, $H_R$,
$H_U$, $H_D$, i.e.,
\begin{align}
H_L:=\left\{ (x,y)\in\bbZ^2\mid x\le -1\right\},\quad
H_R:=\left\{ (x,y)\in\bbZ^2\mid 0\le x\right\},\\
H_U:=\left\{ (x,y)\in\bbZ^2\mid 0\le y\right\},\quad
H_D:=\left\{ (x,y)\in\bbZ^2\mid y\le -1\right\}.
\end{align}
We use a notation $\beta_g:=\beta_g^{\bbZ^2}$, $\beta_g^U:=\beta_g^{H_U}$,
$\beta_g^{RU}:=\beta_g^{H_R\cap H_U}$, $\beta_g^{LU}:=\beta_g^{H_L\cap H_U}$.

For each subset $S$ of $\bbZ^2$, we set
\begin{align}
S_\sigma:=S\cap H_\sigma,\quad S_\zeta:=S\cap H_\zeta,\quad
S_{\sigma\zeta}:=S\cap H_\sigma\cap H_\zeta
\quad \sigma=L,R,\quad\zeta=U,D.
\end{align}
We ocationally write $\caA_{S,\sigma}$, $\caA_{S,\zeta}$, $\caA_{S,\sigma,\zeta}$
to denote $\caA_{S_{\sigma}}$, $\caA_{S_{\zeta}}$, $\caA_{S_{\sigma\zeta}}$.
For an automorphism $\alpha$ on $\caA$ and $0<\theta<\frac{\pi}2$,
we denote by ${\mathfrak D}^{\theta}_\alpha$ a set of all triples $(\alpha_L,\alpha_R,\Theta)$
with 
\begin{align}
\alpha_L\in \Aut\lmk \caA_{H_L}\rmk,\quad \alpha_R\in \Aut\lmk \caA_{H_R}\rmk,\quad
\Theta\in\Aut\lmk\caA_{\lmk C_\theta\rmk^c}\rmk
\end{align}
decomposing $\alpha$ as
\begin{align}
\alpha=\inn\circ\lmk\alpha_L\otimes\alpha_R\rmk\circ\Theta.
\end{align}
For $(\alpha_L,\alpha_R,\Theta)\in {\mathfrak D}^{(\theta)}_\alpha$,
we set
\begin{align}\label{azdef}
\alpha_0:=\alpha_L\otimes\alpha_R.
\end{align}
The class of automorphisms
which allow such decompositions for any directions are denoted by
\begin{align}
\qaut\lmk \caA\rmk
:=\left\{
\alpha\in \Aut(\caA)\mid\;
{\mathfrak D}^{\theta}_\alpha\neq\emptyset \;\;\text{for all}\;
0<\theta<\frac\pi 2
\right\}.
\end{align}
Furthermore, for each\begin{align}\label{thetas1}
0<\theta_{0.8}<\theta_1<\theta_{1.2}<\theta_{1.8}<\theta_2<\theta_{2.2}<
\theta_{2.8}<\theta_3<\theta_{3.2}<\frac\pi 2,
\end{align}
we consider decompositions of $\alpha\in \Aut(\caA)$ such that
\begin{align}\label{sqaut}
&\alpha=\inn\circ\lmk
\alpha_{[0,\theta_1]}\otimes\alpha_{(\theta_1,\theta_2]}
\otimes \alpha_{(\theta_2,\theta_3]}\otimes
\alpha_{(\theta_3,\frac\pi 2]}
\rmk
\circ
\lmk
\alpha_{(\theta_{0.8}, \theta_{1.2}]}\otimes
\alpha_{(\theta_{1.8},\theta_{2.2}]}
\otimes \alpha_{(\theta_{2.8},\theta_{3.2}]}
\rmk
\end{align}
with  \begin{align}\label{sqaut2}
  \begin{split}
&  \alpha_X:=\bigotimes_{\sigma=L,R,\zeta=D,U} \alpha_{X,\sigma,\zeta},\quad
 \alpha_{[0,\theta_1]}:=\bigotimes_{\sigma=L,R}\alpha_{[0,\theta_{1}],\sigma},\quad
 \alpha_{(\theta_3,\frac\pi 2]}:=\bigotimes_{\zeta=D,U}  \alpha_{(\theta_3,\frac\pi 2],\zeta}\\
 &\alpha_{X,\sigma,\zeta}\in \Aut\lmk\caA_{C_{X,\sigma,\zeta}}\rmk,\quad
 \alpha_{X,\sigma}:=\bigotimes_{\zeta=U,D}\alpha_{X,\sigma,\zeta},\quad
\alpha_{X,\zeta}:=\bigotimes_{\sigma=L,R}\alpha_{X,\sigma,\zeta}\\
&\alpha_{[0,\theta_{1}],\sigma}\in \Aut\lmk\caA_{C_{[0,\theta_{0}],\sigma}}\rmk,\quad
 \alpha_{(\theta_3,\frac\pi 2],\zeta}\in \Aut\lmk\caA_{C_{(\theta_3,\frac\pi 2],\zeta}}\rmk, 
  \end{split} 
  \end{align}
 for
 \begin{align}\label{sqaut3}
 X=(\theta_1,\theta_2], (\theta_2,\theta_3],
% (\theta_{-0.2},\theta_{0.2}],
 (\theta_{0.8},\theta_{1.2}],
 (\theta_{1.8},\theta_{2.2}], 
(\theta_{2.8},\theta_{3.2}],\quad \sigma=L,R,\quad \zeta=D,U.
 \end{align}

%
%with
%\begin{align}\label{sqaut2}
%&\alpha_{I}\in \Aut\lmk\caA_{\caC_I\cap \Gamma}\rmk,\notag\\
%&\alpha_{Y,D}\in \Aut\lmk\caA_{(\caC_Y)_D\cap\Gamma}\rmk,\quad \alpha_{Y,U}\in \Aut\lmk\caA_{(\caC_Y)_U\cap\Gamma}\rmk,\notag\\
%&\alpha_{X,L}\in \Aut\lmk\caA_{(\caC_X)_L\cap\Gamma}\rmk,\quad \alpha_{X,R}\in \Aut\lmk\caA_{(\caC_X)_R\cap\Gamma}\rmk,
%\end{align}
%for 
%\begin{align}\label{sqaut3}
%\begin{split}
%&I=[0,\theta_1],(\theta_1,\theta_2], (\theta_2,\theta_3], \left(\theta_3,\frac\pi 2\right],
%(\theta_{0.8}, \theta_{1.2}], (\theta_{1.8},\theta_{2.2}], (\theta_{2.8},\theta_{3.2}],\\
%&Y=(\theta_1,\theta_2], (\theta_2,\theta_3], \left(\theta_3,\frac\pi 2\right],
%(\theta_{0.8}, \theta_{1.2}], (\theta_{1.8},\theta_{2.2}], (\theta_{2.8},\theta_{3.2}],\\
%&X=[0,\theta_1], (\theta_1,\theta_2], \theta_2,\theta_3], (\theta_{0.8}, \theta_{1.2}], \theta_{1.8},\theta_{2.2}], (\theta_{2.8},\theta_{3.2}].
%\end{split}
%\end{align}
The class of automorphisms on $\caA$ 
which allow such decompositions for any directions
$\theta_{0.8}, \theta_1, \theta_{1.2}$, $\theta_{1.8},\theta_2,\theta_{2.2},
\theta_{2.8},\theta_3,\theta_{3.2}$
(satisfying (\ref{thetas1})) is denoted by $\sqaut(\caA)$.
Note that $\sqaut(\caA)\subset\qaut(\caA)$.
The set of all $\alpha\in \sqaut(\caA)$ with
each of 
$\alpha_{I}$ in the decompositions required to commute with $\beta_g^{U}$, $g\in G$,
is denoted by $\gsqaut(\caA)$
\begin{align}
\gsqaut(\caA):=\left\{
\alpha\in \sqaut(\caA)\middle|\begin{gathered}
\text{for any}\; \theta_{0.8}, \theta_1, \theta_{1.2},\theta_{1.8},\theta_2,\theta_{2.2},
\theta_{2.8},\theta_3,\theta_{3.2}\quad\text{satisfying (\ref{thetas1})}\\\text{there is a decomposition }(\ref{sqaut}), (\ref{sqaut2}), (\ref{sqaut3})\;
\text{satisfying}\\
\alpha_{I}\circ\beta_g^{U}=\beta_g^{U}\circ\alpha_{I},\quad g\in G,\\
\text{for all}\;\;I=[0,\theta_1],(\theta_1,\theta_2], (\theta_2,\theta_3], \left(\theta_3,\frac\pi 2\right],
(\theta_{0.8}, \theta_{1.2}], (\theta_{1.8},\theta_{2.2}], (\theta_{2.8},\theta_{3.2}]
\end{gathered}
\right\}.
\end{align}
We also define
\begin{align}
\haut\lmk \caA\rmk:=
\left\{\alpha\in\Aut(\caA)\middle|
\begin{gathered}
\text{for any}\;\; 0<\theta<\frac\pi 2,\;\;
\text{there exist}\;\;\alpha_{\sigma}\in \Aut\lmk \caA_{{\lmk C_{\theta}\rmk_\sigma}}\rmk, \sigma=L,R\\
\text{such that}\;\; \alpha=\inn\circ\lmk \alpha_{L}\otimes \alpha_{R}\rmk
\end{gathered}
\right\}.
\end{align}
In section \ref{mainthproofsec}, we will see that
quasi-local automorphisms corresponding to paths in symmetric gapped phases
belong to the following set:
\begin{align}
\guaut\lmk\caA\rmk
:=\left\{
\gamma\in\Aut\lmk\caA\rmk\middle|\begin{gathered}
\text{there are} \;\; \gamma_{H}\in \haut(\caA), 
%\gamma_{U}\in \gsqaut(\caA_{H_{U}}),
\gamma_{C}\in \gsqaut(\caA)\\
\text{such that}\;\;
\gamma=\gamma_{C}\circ\gamma_{H}
\end{gathered}
\right\}.
\end{align}

We fix a reference state $\omega_0$ as follows.
We fix a unit vector  ${\xi_x}\in\bbC^d$ and let $\rho_{\xi_x}$ be the vector state on
$\Mat_d$ given by ${\xi_x}$, for each $x\in\bbZ^2$.
Then our reference state $\omega_0$ is given by
 \begin{align}\label{ozs}
 \omega_0:=\bigotimes_{\bbZ^2} \rho_{\xi_x}.
 \end{align}
Throughout this section this $\omega_0$ is fixed.
Let $(\caH_0,\pi_0,\Omega_0)$ be a GNS triple of $\omega_0$.
Because of the product structure of $\omega_0$,
it is decomposed as
\begin{align}
\caH_0=\caH_L\otimes\caH_R,\quad \pi_0=\pi_L\otimes\pi_R,\quad
\Omega_0=\Omega_L\otimes\Omega_R,
\end{align}
where $(\caH_\sigma,\pi_\sigma,\Omega_\sigma)$ is a GNS triple of
$\omega_\sigma:=\omega_0\vert_{\caA_{H_\sigma}}$ for $\sigma=L,R$.
As $\omega_0\vert_{\caA_{H_\sigma}}$ is pure, $\pi_\sigma$ is irreducible.
What we consider in this section is the set of states
created  via elements in $\qaut(\caA)$ from our reference state $\omega_0$:
\begin{align}
\QLS:=\left\{\omega_0\circ\alpha\mid \alpha\in \qaut(\caA)
\right\}.
\end{align}
Because any pure product states can be transformed to each other via 
an automorphism of product form $\tilde\alpha=\bigotimes_{x\in \bbZ^{2}}\tilde\alpha_{x}$
and $\tilde\alpha\alpha$ belongs to $\qaut(\caA)$ for any $\alpha\in\qaut(\caA)$,
 $\QLS$ does not depend on the choice of $\omega_0$.
For each $\omega\in \QLS$, we set
\begin{align}
\eaut(\omega):=
\left\{
\alpha\in \qaut(\caA)\mid \omega=\omega_0\circ\alpha
\right\}.
\end{align}
By the definition of $\QLS$, $\eaut(\omega)$ is not empty.

For $0<\theta<\frac\pi 2$
and a set of automorphisms
$(\tilde\beta_g)_{g\in G}\subset \Aut({\caA})$,
we introduce a set
\begin{align}
\caT(\theta, (\tilde\beta_g))
:=\left\{
(\eta_{g}^\sigma)_{g\in G, \sigma=L,R
}
\middle|
\begin{gathered}
\eta_g^\sigma\in\Aut\lmk \caA_{\lmk C_\theta\rmk_\sigma}\rmk,
%\quadg\in G,\sigma=L,R
\\
\tilde\beta_g=\inn\circ\lmk \eta_g^L\otimes\eta_g^R\rmk\circ\beta_g^U,\\
\text{for all}\;g\in G,\; \sigma=L,R
\end{gathered}
\right\}.
\end{align}
In a word, it is a set of decompositions of $\tilde\beta_g\circ\lmk \beta_g^{U}\rmk^{-1}$
into tensor of $\Aut\lmk \caA_{\lmk C_\theta\rmk_L}\rmk$,
$\Aut\lmk \caA_{\lmk C_\theta\rmk_R}\rmk$ modulo inner automorphisms.
For $(\eta_{g}^\sigma)_{g\in G, \sigma=L,R
}\in \caT(\theta, (\tilde\beta_g))$,
we set
\begin{align}\label{etadef}
\eta_g:=\eta_g^L\otimes\eta_g^R,\quad g\in G.
\end{align}
The following set of automorphisms is the key ingredient for the definition of our index.
For $\omega\in \QLS$ and $0<\theta<\frac\pi 2$, we set
\begin{align}
\IG\lmk\omega,\theta\rmk
:=\left\{
(\tilde\beta_g)_{g\in G}\in\Aut\lmk \caA\rmk^{\times G}
\middle|
\begin{gathered}
\omega\circ\tilde\beta_g=\omega \quad\text{for all}\quad g\in G,\\
\text{and}\quad\caT(\theta, (\tilde\beta_g))\neq\emptyset\end{gathered}
\right\}.
\end{align}
We also set
\begin{align}
\IG\lmk\omega\rmk:=\cup_{0<\theta<\frac \pi 2}\IG\lmk\omega,\theta\rmk.
\end{align}
In this section we associate some third cohomology  $h(\omega)$
for each $\omega\in\QLS$ with $\IG(\omega)\neq\emptyset$.
\subsection{Derivation of elements in $Z^3(G,\bbT)$}
In this subsection, we derive $3$-cocycles out of $\omega$, $\alpha$, $\theta$,
$(\tilde\beta_g)$, $(\eta_{g}^\sigma)$
$(\alpha_L,\alpha_R,\Theta)$.
 
\begin{lem}\label{ichi}
Let $\omega\in\QLS$, $\alpha\in \eaut(\omega)$, $0<\theta<\frac\pi 2$,
$(\tilde\beta_g)\in \IG\lmk\omega,\theta\rmk$, $(\eta_{g}^\sigma)\in \caT(\theta, (\tilde\beta_g))$,
$(\alpha_L,\alpha_R,\Theta)\in{\mathfrak D}^{\theta}_\alpha$.
Then 
\begin{description}
\item[(i)]
There are unitaries $W_g$, $g\in G$ on $\caH_0$
such that
\begin{align}\label{wimp}
\Ad\lmk W_g\rmk\circ\pi_0
=\pi_0\circ\alpha_0\circ\Theta\circ\eta_g\beta_g^U\circ\Theta^{-1}\circ\alpha_0^{-1},\quad
g\in G
\end{align}
with notation (\ref{azdef}), (\ref{etadef}).
\item[(ii)]
There exist a unitary $u_\sigma(g,h)$ on $\caH_{\sigma}$, for each $\sigma=L,R$ ,$g,h\in G$, 
such that
\begin{align}\label{usig}
\Ad\lmk u_\sigma(g,h)\rmk\circ\pi_\sigma
=\pi_\sigma\circ\alpha_\sigma\circ\eta_g^\sigma\beta_g^{\sigma U}
\eta_h^\sigma\lmk\beta_g^{\sigma U}\rmk^{-1}\lmk \eta_{gh}^\sigma\rmk^{-1}
\circ\alpha_\sigma^{-1},
\end{align}
and
\begin{align}\label{usigt}
\Ad\lmk u_L(g,h)\otimes u_R(g,h)\rmk\pi_0
=\pi_0\circ\alpha_0\circ\eta_g\beta_g^U\eta_h\lmk \beta_g^U\rmk^{-1}\lmk\eta_{gh}\rmk^{-1}
\circ\alpha_0^{-1}
.
\end{align}
Furthermore, $u_{\sigma}(g,h)$
commutes with any element of 
$ \pi_\sigma\circ\alpha_\sigma\lmk \caA_{\lmk \lmk C_\theta\rmk^c\rmk_\sigma}\rmk$.
\end{description}
\end{lem}

\begin{defn}
For $\omega\in\QLS$, $\alpha\in \eaut(\omega)$, $0<\theta<\frac\pi 2$,
$(\tilde\beta_g)\in \IG\lmk\omega,\theta\rmk$, $(\eta_{g}^\sigma)_{g\in G, \sigma=L,R
}\in \caT(\theta, (\tilde\beta_g))$,
$(\alpha_L,\alpha_R,\Theta)\in{\mathfrak D}^{\theta}_\alpha$,
we denote by
\begin{align}
\IP\lmk
\omega, \alpha, \theta,
(\tilde\beta_g), (\eta_{g}^\sigma),
(\alpha_L,\alpha_R,\Theta)
\rmk
\end{align}
the set of $\lmk (W_g)_{g\in G}, (u_\sigma(g,h))_{g,h\in G,\sigma=L,R}\rmk$ with
$W_{g}\in \caU\lmk\caH_{0}\rmk$ and $u_\sigma(g,h)\in\caU\lmk\caH_{\sigma}\rmk$
satisfying
\begin{align}
&\Ad\lmk W_g\rmk\circ\pi_0
=\pi_0\circ\alpha_0\circ\Theta\circ\eta_g\beta_g^U\circ\Theta^{-1}\circ\alpha_0^{-1},\; 
g\in G,\quad \text{and}\label{ipdef}\\
&\Ad\lmk u_{\sigma}(g,h)\rmk\circ\pi_{\sigma}
=\pi_{\sigma}\circ\alpha_{\sigma}\circ\eta_g^{\sigma}\beta_g^{{\sigma} U}
\eta_h^{\sigma}\lmk\beta_g^{{\sigma} U}\rmk^{-1}\lmk \eta_{gh}^{\sigma}\rmk^{-1}
\circ\alpha_{\sigma}^{-1},\quad g,h\in G, \sigma=L,R.\label{ipdef2}
\end{align}
(Here we used notation (\ref{azdef}) and (\ref{etadef}).)
By Lemma \ref{ichi}, it is non-empty.
\end{defn}
\begin{proof}
For a GNS triple $(\caH_0,\pi_0\circ\alpha,\Omega_0)$ of $\omega=\omega_0\circ\alpha$
there are unitaries $\tilde W_g$ on $\caH_0$ such that
\begin{align}
\Ad\lmk\tilde W_g\rmk\circ\pi_0\circ\alpha
=\pi_0\circ\alpha\circ\tilde\beta_g,\quad g\in G
\end{align}
because $\omega\circ\tilde\beta_g=\omega$.

Because $(\eta_{g}^\sigma)_{g\in G, \sigma=L,R
}\in \caT(\theta, (\tilde\beta_g))$,
and $(\alpha_L,\alpha_R,\Theta)\in{\mathfrak D}^{\theta}_\alpha$,
there are unitaries $v_g, V\in \caU\lmk\caA\rmk$ such that
\begin{align}
\tilde\beta_g=\Ad\lmk v_g\rmk\circ\lmk \eta_g^L\otimes\eta_g^R\rmk\circ\beta_g^U,\quad
\alpha=\Ad V\circ\alpha_0\circ\Theta.
\end{align}
Substituting these, we have
\begin{align}
\Ad\lmk \tilde W_g\pi_0(V)\rmk\pi_0\circ\alpha_0\circ\Theta
=\pi_0\circ\alpha\tilde\beta_g
=\pi_0\circ\alpha\circ\Ad(v_g)\circ \eta_g\beta_g^U
=\Ad\lmk
\lmk \pi_0\circ\alpha(v_g)\rmk\pi_0(V)
\rmk\pi_0\circ\alpha_0\circ\Theta\circ\eta_g\beta_g^U.
\end{align}
Therefore, setting $W_g:=\pi_0(V)^*\lmk \pi_0\circ\alpha(v_g^*)\rmk\tilde W_g\pi_0(V)\in\caU(\caH_0)$,
we obtain (\ref{wimp}).

Using this (\ref{wimp}), we have
\begin{align}\label{wwwgh}
\Ad\lmk W_gW_h W_{gh}^*\rmk \pi_0
=\pi_0\circ\alpha_0\circ\Theta \circ\eta_g\beta_g^U\eta_h\lmk \beta_g^{U}\rmk^{-1}
\eta_{gh}^{-1}\Theta^{-1}\alpha_0^{-1}.
\end{align}
Note that because conjugation by $\beta_g^U$ does not change the support of 
automorphisms, $\eta_g\beta_g^U\eta_h\lmk \beta_g^{U}\rmk^{-1}
\eta_{gh}^{-1}$ belongs to $\Aut\lmk\caA_{C_\theta}\rmk$.
On the other hand, $\Theta$ belongs to  $\Aut\lmk\caA_{\lmk C_\theta\rmk^{c}}\rmk$.
Therefore, they commute and we obtain
\begin{align}
\Ad\lmk W_gW_h W_{gh}^*\rmk \pi_0
=(\ref{wwwgh})=\pi_0\circ\alpha_0\circ\eta_g\beta_g^U\eta_h\lmk \beta_g^{U}\rmk^{-1}
\eta_{gh}^{-1}\alpha_0^{-1}
=\bigotimes_{\sigma=L,R}
\pi_\sigma\circ\alpha_\sigma\circ\eta_g^\sigma\beta_g^{\sigma U}
\eta_h^\sigma\lmk\beta_g^{\sigma U}\rmk^{-1}\lmk \eta_{gh}^\sigma\rmk^{-1}
\circ\alpha_\sigma^{-1}
\end{align}
From this and the irreducibility of $\pi_R$, we see that $\Ad\lmk W_gW_h W_{gh}^*\rmk$
gives rise to a $*$-isomorphism $\tau$ on $\caB(\caH_R)$.It is implemented by
some unitary $u_R(g,h)$ on $\caH_R$ by the Wigner theorem and
we obtain
\begin{align}
\begin{split}
&\unit_{\caH_L}\otimes \lmk \Ad\lmk u_R(g,h) \rmk\circ\pi_R(A)\rmk
=\unit_{\caH_L}\otimes \tau\lmk \pi_R(A)\rmk
=\Ad\lmk W_gW_h W_{gh}^*\rmk\lmk \unit_{\caH_L}\otimes \pi_R(A)\rmk\\
&=\unit_{\caH_L}\otimes \pi_R\circ\alpha_R\circ\eta_g^R\beta_g^{R U}
\eta_h^R\lmk\beta_g^{R U}\rmk^{-1}\lmk \eta_{gh}^R\rmk^{-1}
\circ\alpha_R^{-1}(A),
\end{split}
\end{align}
for any $A\in\caA_{H_{R}}$.
Hence we obtain (\ref{usig}) for $\sigma=R$.

To see that $u_R(g,h)$ belongs to $\lmk \pi_R\circ\alpha_R\lmk \caA_{\lmk \lmk C_\theta\rmk^c\rmk_R}\rmk\rmk'$,
let $A\in \caA_{\lmk \lmk C_\theta\rmk^c\rmk_R}$.
Then because $\eta_g^R\beta_g^{R U}
\eta_h^R\lmk\beta_g^{R U}\rmk^{-1}\lmk \eta_{gh}^R\rmk^{-1}$
belongs to $\Aut\lmk \caA_{\lmk C_\theta\rmk_R}\rmk$, we have
\begin{align}
\Ad\lmk u_R(g,h)\rmk\pi_R\lmk \alpha_R(A)\rmk
=\pi_R\alpha_R \eta_g^R\beta_g^{R U}
\eta_h^R\lmk\beta_g^{R U}\rmk^{-1}\lmk \eta_{gh}^R\rmk^{-1}\alpha_R^{-1}\alpha_R(A)
=\pi_R\alpha_R(A).
\end{align}
This proves that $u_R(g,h)$ belongs to $\lmk \pi_R\circ\alpha_R\lmk \caA_{\lmk \lmk C_\theta\rmk^c\rmk_R}\rmk\rmk'$.
Analogous statement for $u_L(g,h)$ can be shown exactly the same way.
The last statement (\ref{usigt}) of (ii) is trivial from (\ref{usig}).
\end{proof}

\begin{lem}\label{lemni}
Let $\omega\in\QLS$, $\alpha\in \eaut(\omega)$, $0<\theta<\frac\pi 2$,
$(\tilde\beta_g)\in \IG\lmk\omega,\theta\rmk$, $(\eta_{g}^\sigma)\in \caT(\theta, (\tilde\beta_g))$,
$(\alpha_L,\alpha_R,\Theta)\in{\mathfrak D}^{\theta}_\alpha$.
Let $\lmk (W_g), (u_R(g,h))\rmk $ be an element of $\IP\lmk
\omega, \alpha, \theta,
(\tilde\beta_g), (\eta_{g}^\sigma),
(\alpha_L,\alpha_R,\Theta)
\rmk$.

Then the followings hold.
\begin{description}
\item[(i)] For any $g,h,k\in G$,
\begin{align}
\Ad\lmk W_g\lmk \unit_{\caH_L}\otimes u_R(h,k)\rmk W_g^*\rmk\circ\pi_0
=\pi_0\circ
\lmk
\id_{\caA_{H_{L}}}\otimes \alpha_R\eta_g^R\beta_g^{RU}
\lmk
\eta_h^R\beta_h^{R U}
\eta_k^R\lmk\beta_h^{R U}\rmk^{-1}\lmk \eta_{hk}^R\rmk^{-1}
\rmk
\lmk \eta_g^R\beta_g^{RU}\rmk^{-1}
\alpha_R^{-1}
\rmk.
\end{align}
\item[(ii)]
For any $g,h\in G$,
\begin{align}\label{hato}
\Ad\lmk \lmk u_L(g,h)\otimes u_R(g,h)\rmk W_{gh}\rmk
=\Ad \lmk W_gW_h\rmk,
\end{align}
on $\caB(\caH_{0})$.
\item[(iii)]
For any $g,h,k\in G$,
\begin{align}
\Ad\lmk W_g\rmk \lmk \unit_{\caH_L}\otimes u_R(h,k)\rmk
\in\bbC\unit_{\caH_L}\otimes\caB(\caH_R).
\end{align}
\item[(iv)]
For any $g,h,k,f\in G$,
\begin{align}
\Ad\lmk W_g W_h\rmk\lmk  \unit_{\caH_L}\otimes u_R(k,f)\rmk
=\lmk
\Ad\lmk  \lmk \unit_{\caH_L}\otimes u_R(g,h)\rmk W_{gh}\rmk
\rmk
\lmk \unit_{\caH_L}\otimes u_R(k,f)\rmk.
\end{align}
\end{description}
\end{lem}
\begin{proof}
We use the notation  (\ref{azdef}), (\ref{etadef}).\\
(i)
Substituting (\ref{ipdef}) (\ref{ipdef2}), we have
\begin{align}\label{2ichi}
&\Ad\lmk W_g\lmk \unit_{\caH_L}\otimes u_R(h,k)\rmk W_g^*\rmk\circ\pi_0\notag\\
&=\pi_0\circ\alpha_0\circ\Theta\circ\eta_g\beta_g^U\circ\Theta^{-1}\circ\alpha_0^{-1}\circ
\lmk \id_{\caA_{H_L}}\otimes \alpha_R\circ\eta_h^R\beta_h^{R U}
\eta_k^R\lmk\beta_h^{R U}\rmk^{-1}\lmk \eta_{hk}^R\rmk^{-1}
\circ\alpha_R^{-1}\rmk
\circ\alpha_0\circ\Theta\circ\lmk \eta_g\beta_g^U\rmk^{-1}\circ\Theta^{-1}\circ\alpha_0^{-1}\notag\\
&=\pi_0\circ\alpha_0\circ\Theta\circ\eta_g\beta_g^U\circ\Theta^{-1}\circ
\lmk \id_{\caA_{H_L}}\otimes \eta_h^R\beta_h^{R U}
\eta_k^R\lmk\beta_h^{R U}\rmk^{-1}\lmk \eta_{hk}^R\rmk^{-1}
\rmk
\circ\Theta\circ\lmk \eta_g\beta_g^U\rmk^{-1}\circ\Theta^{-1}\circ\alpha_0^{-1}.
\end{align}
Because $\eta_h^R\beta_h^{R U}
\eta_k^R\lmk\beta_h^{R U}\rmk^{-1}\lmk \eta_{hk}^R\rmk^{-1}$ belongs to
$\Aut\lmk\caA_{\lmk C_\theta\rmk_R}\rmk$,
it commutes with $\Theta\in \Aut\lmk\caA_{\lmk C_\theta\rmk^{c}}\rmk$.
Hence we obtain
\begin{align}
&(\ref{2ichi})=
\pi_0\circ\alpha_0\circ\Theta\circ\eta_g\beta_g^U\circ
\lmk \unit_{\caA_{H_L}}\otimes \eta_h^R\beta_h^{R U}
\eta_k^R\lmk\beta_h^{R U}\rmk^{-1}\lmk \eta_{hk}^R\rmk^{-1}
\rmk
\circ\lmk \eta_g\beta_g^U\rmk^{-1}\circ\Theta^{-1}\circ\alpha_0^{-1}\notag\\
&=\pi_0\circ\alpha_0\circ\Theta\circ
\lmk \unit_{\caA_{H_L}}\otimes \eta_g^R\beta_g^{RU}\circ\eta_h^R\beta_h^{R U}
\eta_k^R\lmk\beta_h^{R U}\rmk^{-1}\lmk \eta_{hk}^R\rmk^{-1}\circ
\lmk \eta_g^R\beta_g^{RU}\rmk^{-1}
\rmk
\circ\Theta^{-1}\circ\alpha_0^{-1}.
\end{align}
Again, the term in the round braket in the last line is localized at $\lmk C_\theta\rmk_R$, and it commutes with $\Theta$.
Therefore, we have
\begin{align}
\Ad\lmk W_g\lmk \unit_{\caH_L}\otimes u_R(h,k)\rmk W_g^*\rmk\circ\pi_0
=\pi_0\circ
\lmk \id_{\caA_{H_L}}\otimes \alpha_R\circ\eta_g^R\beta_g^{RU}\circ\eta_h^R\beta_h^{R U}
\eta_k^R\lmk\beta_h^{R U}\rmk^{-1}\lmk \eta_{hk}^R\rmk^{-1}\circ
\lmk \eta_g^R\beta_g^{RU}\rmk^{-1}
\circ\alpha_R^{-1}\rmk
\end{align}
(ii)
Again by (\ref{ipdef}) and (\ref{ipdef2}), we have
\begin{align}
&\Ad\lmk \lmk u_L(g,h)\otimes u_R(g,h)\rmk W_{gh}\rmk\circ\pi_0
=\pi_0\circ\alpha_0\circ\eta_g\beta_g^U\eta_h\lmk \beta_g^U\rmk^{-1}\lmk\eta_{gh}\rmk^{-1}
\circ\Theta\circ\eta_{gh}\beta_{gh}^U\circ\Theta^{-1}\circ\alpha_0^{-1}\notag\\
&=\pi_0\circ\alpha_0\circ\Theta\circ\eta_g\beta_g^U\eta_h\lmk \beta_g^U\rmk^{-1}\lmk\eta_{gh}\rmk^{-1}
\circ\eta_{gh}\beta_{gh}^U\circ\Theta^{-1}\circ\alpha_0^{-1}
=\pi_0\circ\alpha_0\circ\Theta\circ\eta_g\beta_g^U\eta_h\beta_{h}^U\circ\Theta^{-1}\circ\alpha_0^{-1}
=\Ad \lmk W_gW_h\rmk\circ\pi_0.
\end{align}
Here, for the second equality, we again used the commutativity of $\eta$s and $\Theta$, due to their
disjoint support.
Because $\pi_0$ is irreducible, we obtain (\ref{hato}).\\
(iii)
For any $A\in \caA_{H_L}$, we have  
\begin{align}
\Theta^{-1}\circ\alpha_0^{-1}\lmk A\otimes \unit_{\caA_{H_R}}\rmk
=\Theta^{-1}\circ\lmk \alpha_L^{-1}(A)\otimes \unit_{\caA_{H_R}}\rmk
\in \Theta^{-1}\lmk\caA_{H_L}\otimes \bbC\unit_{\caA_{H_R}}\rmk
\subset \caA_{H_L\cup \lmk C_\theta^c\rmk_R },
\end{align}
because $\Theta\in \Aut\lmk\caA_{\lmk C_\theta\rmk^{c}}\rmk$.
Therefore,  $\eta_g^R\in \Aut(\caA_{\lmk C_\theta\rmk_R})$ acts trivially on it and
we have
\begin{align}
\lmk \beta_g^U\rmk^{-1}\lmk \eta_g\rmk^{-1}\circ\Theta^{-1}\circ\alpha_0^{-1}\lmk A\otimes \unit_{\caA_{H_R}}\rmk
\in  \caA_{H_L\cup \lmk C_\theta^c\rmk_R }.
\end{align}
As $\Theta$ preserves $\caA_{H_L\cup \lmk C_\theta^c\rmk_R }$, 
\begin{align}
\Theta\circ\lmk \beta_g^U\rmk^{-1}\lmk \eta_g\rmk^{-1}\circ\Theta^{-1}\circ\alpha_0^{-1}\lmk A\otimes \unit_{\caA_{H_R}}\rmk
\end{align}
also belongs to  $\caA_{H_L\cup \lmk C_\theta^c\rmk_R }$.
As a result,
\begin{align}
\Ad\lmk W_g^*\rmk\lmk \pi_L(A)\otimes \unit_{\caH_R}\rmk
=\pi_0\circ\alpha_0\circ\Theta\circ\lmk \beta_g^U\rmk^{-1}\lmk \eta_g\rmk^{-1}\circ\Theta^{-1}\circ\alpha_0^{-1}\lmk A\otimes \unit_{\caA_{H_R}}\rmk
\end{align}
belongs to 
$\pi_L(\caA_{H_L})\otimes \pi_R\circ\alpha_R\lmk \caA_{\lmk C_\theta^c\rmk_R } \rmk$,
hence commutes with $\unit_{\caH_L}\otimes u_R(h,k)$.
Hence $\Ad(W_g)\lmk \unit_{\caH_L}\otimes u_R(h,k) \rmk$
commutes with any elements in $\pi_L(\caA_L)\otimes \bbC\unit_{\caH_R}$.
Because $\pi_L$ is irreducible,
$\Ad(W_g)\lmk \unit_{\caH_L}\otimes u_R(h,k) \rmk$
belongs to $\bbC\unit_{\caH_L}\otimes \caB(\caH_R)$.\\
(iv)
By (iii), $\Ad\lmk  W_{gh}\rmk\lmk  \unit_{\caH_L}\otimes u_R(k,f)\rmk$
belongs to $\bbC\unit_{\caH_L}\otimes \caB(\caH_R)$.
Therefore, from (ii), we have
\begin{align}
&\Ad\lmk W_g W_h\rmk\lmk  \unit_{\caH_L}\otimes u_R(k,f)\rmk
=\Ad\lmk \lmk u_L(g,h)\otimes u_R(g,h)\rmk W_{gh}\rmk\lmk  \unit_{\caH_L}\otimes u_R(k,f)\rmk\notag\\
&=\Ad\lmk \lmk \unit_{\caH_L}\otimes u_R(g,h)\rmk W_{gh}\rmk\lmk  \unit_{\caH_L}\otimes u_R(k,f)\rmk,
\end{align}
obtaining (iv).
\end{proof}
With this preparation we may obtain some element of $Z^3(G,\bbT)$
from $((W_g), (u_\sigma(g,h)))$.

\begin{lem}\label{lem3}
Let $\omega\in\QLS$, $\alpha\in \eaut(\omega)$, $0<\theta<\frac\pi 2$,
$(\tilde\beta_g)\in \IG\lmk\omega,\theta\rmk$, 
$(\eta_{g}^\sigma)\in \caT(\theta, (\tilde\beta_g))$,
$(\alpha_L,\alpha_R,\Theta)\in{\mathfrak D}^{\theta}_\alpha$.
Let $\lmk (W_g), (u_\sigma(g,h))\rmk$ be an element of $ \IP\lmk
\omega, \alpha, \theta,
(\tilde\beta_g), (\eta_{g}^\sigma),
(\alpha_L,\alpha_R,\Theta)
\rmk$.
Then there is a $c_R\in Z^3(G,\bbT)$
such that
\begin{align}\label{uwc}
\unit_{\caH_L}\otimes u_R(g,h) u_R(gh,k)
=c_R(g,h,k)
\lmk W_g\lmk \unit_{\caH_L}\otimes u_R(h,k)\rmk W_g^*\rmk
\lmk \unit_{\caH_L}\otimes u_R(g,hk)\rmk,
\end{align}
for all $g,h,k\in G$.
\end{lem}
\begin{defn}\label{nagasaki}
We denote this $3$-cocycle $c_R$ in the Lemma by 
\begin{align}
c_R\lmk
\omega, \alpha, \theta,
(\tilde\beta_g), (\eta_{g}^\sigma),
(\alpha_L,\alpha_R,\Theta),\lmk (W_g), (u_\sigma(g,h))\rmk
\rmk
\end{align}
and its cohomology class by
\begin{align}
h^{(1)}
\lmk
\omega, \alpha, \theta,
(\tilde\beta_g), (\eta_{g}^\sigma),
(\alpha_L,\alpha_R,\Theta),\lmk (W_g), (u_\sigma(g,h))\rmk
\rmk
:=
\lcm
c_R\lmk
\omega, \alpha, \theta,
(\tilde\beta_g), (\eta_{g}^\sigma),
(\alpha_L,\alpha_R,\Theta),\lmk (W_g), (u_\sigma(g,h))\rmk
\rmk
\rcm_{H^3(G,\bbT)}.
\end{align}
\end{defn}
\begin{proof}
First we prove that there is a number $c_R(g,h,k)\in\bbT$
satisfying (\ref{uwc}).
From (\ref{ipdef2}),we have
\begin{align}\label{gomach}
\Ad\lmk \unit_{\caH_L}\otimes u_R(g,h) u_R(gh,k)
\rmk\pi_0
=\pi_L\otimes \pi_R \circ \alpha_R\circ \lmk \eta_g^R\beta_g^{RU}\rmk
\lmk \eta_h^R\beta_h^{RU}\rmk
\lmk \eta_k^R\beta_k^{RU}\rmk
\lmk \eta_{ghk}^R\beta_{ghk}^{RU}\rmk^{-1}
\alpha_R^{-1}.
\end{align}
On the other hand, using (i) of Lemma \ref{lemni}, we have 
\begin{align}
\Ad\lmk \lmk W_g\lmk \unit_{\caH_L}\otimes u_R(h,k)\rmk W_g^*\rmk
\lmk \unit_{\caH_L}\otimes u_R(g,hk)\rmk\rmk
\pi_0
\end{align}
is also equal to the right hand side of (\ref{gomach}).
Because $\pi_0$ is irreducible, this means that there is a number $c_R(g,h,k)\in\bbT$
satisfying (\ref{uwc}).

Now let us check that this $c_R$ is a $3$-cocycle.
For any $g,h,k,f\in G$, by repeated use of (\ref{uwc}), we get
\begin{align}
& \unit_{\caH_L}\otimes u_R(g,h) u_R(gh,k)u_R(ghk,f)
=\lcm \unit_{\caH_L}\otimes u_R(g,h) u_R(gh,k)\rcm\cdot\lmk \unit_{\caH_L}\otimes u_R(ghk,f)\rmk\\
& =\lmk
c_R(g,h,k)
 \lmk W_g\lmk \unit_{\caH_L}\otimes u_R(h,k)\rmk W_g^*\rmk
\lmk \unit_{\caH_L}\otimes u_R(g,hk)\rmk\rmk\cdot
\lmk \unit_{\caH_L}\otimes u_R(ghk,f)\rmk
\notag\\
& =\lmk c_R(g,h,k)
 \lmk W_g\lmk \unit_{\caH_L}\otimes u_R(h,k) \rmk W_g^*\rmk\rmk\cdot
\lcm \unit_{\caH_L}\otimes u_R(g,hk)u_R(ghk,f)\rcm
\notag\\
& =\lmk c_R(g,h,k)
 \lmk W_g\lmk \unit_{\caH_L}\otimes u_R(h,k) \rmk W_g^*\rmk\rmk\cdot
\lmk  c_R(g,hk,f) \lmk W_g\lmk \unit_{\caH_L}\otimes u_R(hk,f)\rmk W_g^*\rmk
 \lmk \unit_{\caH_L}\otimes u_R(g,hkf)\rmk \rmk\notag\\
 &=
  c_R(g,h,k) c_R(g,hk,f)
 \lmk 
 W_g\lcm \unit_{\caH_L}\otimes u_R(h,k)  u_R(hk,f)\rcm W_g^*\rmk\cdot
 \lmk \unit_{\caH_L}\otimes u_R(g,hkf)\rmk \notag\\
 &=c_R(g,h,k)
 c_R(g,hk,f) c_R(h,k,f)
 W_g\lmk
 W_h \lmk \unit_{\caH_L}\otimes u_R(k,f)\rmk W_h^*\lmk \unit_{\caH_L}\otimes u_R(h,kf)\rmk\rmk
 W_g^*\cdot
 \lmk \unit_{\caH_L}\otimes u_R(g,hkf)\rmk\notag\\
 &=c_R(g,h,k)
 c_R(g,hk,f) c_R(h,k,f)
 \cdot W_gW_h \lmk \unit_{\caH_L}\otimes u_R(k,f)\rmk W_h^*W_g^*\cdot 
 \lcm\lmk W_g\lmk \unit_{\caH_L}\otimes u_R(h,kf)\rmk W_g^*\rmk
 \unit_{\caH_L}\otimes u_R(g,hkf)\rcm\notag\\
& =c_R(g,h,k)
 c_R(g,hk,f) c_R(h,k,f) \overline{c(g,h,kf)}
 \cdot \left\{W_gW_h \lmk \unit_{\caH_L}\otimes u_R(k,f)\rmk W_h^*W_g^*\rmk\}\notag\\
 &\quad \cdot 
\lmk
\unit_{\caH_L}\otimes u_{R}(g,h)u_{R}(gh,kf)
\rmk.
 \label{temp}
\end{align}
Here, (and below) we apply  (\ref{uwc}) for terms in $\lcm\cdot\rcm$ to get the succeeding equality.
Applying (iv) of Lemma \ref{lemni} to the $\{\cdot\}$ part of (\ref{temp}),
we obtain
\begin{align}
(\ref{temp})
&=c_R(g,h,k)
 c_R(g,hk,f) c_R(h,k,f)\overline{c(g,h,kf)}
 \lmk
\Ad\lmk  \lmk \unit_{\caH_L}\otimes u_R(g,h)\rmk W_{gh}\rmk
\rmk
\lmk \unit_{\caH_L}\otimes u_R(k,f)\rmk 
\lmk
\unit_{\caH_L}\otimes u_{R}(g,h)u_{R}(gh,kf)
\rmk\notag\\
&=c_R(g,h,k)
 c_R(g,hk,f) c_R(h,k,f)\overline{c(g,h,kf)}
 \lmk \unit_{\caH_L}\otimes u_R(g,h)\rmk 
 \lcm W_{gh}\lmk \unit_{\caH_L}\otimes u_R(k,f)\rmk W_{gh}^{*}
 \lmk
\unit_{\caH_L}\otimes u_{R}(gh,kf)
\rmk\rcm\notag\\
&=c_R(g,h,k)
 c_R(g,hk,f) c_R(h,k,f)\overline{c(g,h,kf)}\overline{c_{R}(gh,k,f)}
  \lmk \unit_{\caH_L}\otimes u_R(g,h)u_R(gh, k)u_{R}(ghk,f)\rmk.
  \end{align}
Hence, we obtain
\begin{align}
c_R(g,h,k)
 c_R(g,hk,f) c_R(h,k,f)\overline{c(g,h,kf)}\overline{c_{R}(gh,k,f)}=1,\quad\text{for all}\quad
 g,h,k,f\in G.
\end{align}
This means $c_{R}\in Z^{3}(G,\bbT)$.
\end{proof}

\subsection{The $H^{3}(G,\bbT)$-valued index}
From the previous subsection, we remark the following fact.
\begin{lem}
For any $\omega\in\QLS$ with $\IG(\omega)\neq\emptyset$,
there are
\begin{align}\label{oneset}
\begin{split}
\alpha\in \eaut(\omega), \; 0<\theta<\frac\pi 2,\;
(\tilde\beta_g)\in \IG\lmk\omega,\theta\rmk,\; (\eta_{g}^\sigma)\in \caT(\theta, (\tilde\beta_g)),\;
(\alpha_L,\alpha_R,\Theta)\in{\mathfrak D}^{\theta}_\alpha,\\
\lmk (W_g), (u_R(g,h))\rmk\in\IP\lmk
\omega, \alpha, \theta,
(\tilde\beta_g), (\eta_{g}^\sigma),
(\alpha_L,\alpha_R,\Theta)
\rmk.
\end{split}
\end{align}
\end{lem}
\begin{proof}
Because $\IG(\omega)\neq\emptyset$, there is some $0<\theta<\frac\pi 2$
such that $\IG(\omega,\theta)\neq\emptyset$,
and hence $(\tilde\beta_g)\in \IG(\omega,\theta)$ and 
$(\eta_g^\sigma)\in  \caT(\theta, (\tilde\beta_g))$ exist.
Because
$\omega\in\QLS$, by definition there exists some $\alpha\in \eaut(\omega)$
and by the definition of $\eaut(\omega)$,
there is some $(\alpha_L,\alpha_R,\Theta)\in {\mathfrak D}^{\theta}_\alpha$.
The existence of $\lmk (W_g), (u_R(g,h))\rmk\in\IP\lmk
\omega, \alpha, \theta,
(\tilde\beta_g), (\eta_{g}^\sigma),
(\alpha_L,\alpha_R,\Theta)
\rmk$ is given by Lemma \ref{ichi}.
\end{proof}
By Lemma \ref{lem3}, for
$\omega\in\QLS$ with $\IG(\omega)\neq\emptyset$,
for each choice of (\ref{oneset}),
we can associate some element of
$H^3(G,\bbT)$:
\begin{align}\label{hq}
h^{(1)}
\lmk
\omega, \alpha , \theta,
(\tilde\beta_g), (\eta_{g}^\sigma),
(\alpha_L,\alpha_R,\Theta),\lmk (W_g), (u_\sigma(g,h))\rmk
\rmk.
\end{align}
In this subsection, we show that the third cohomology class does not depend on
the choice of (\ref{oneset}):
\begin{thm}\label{welldefined}
For any  $\omega\in\QLS$ with
$\IG(\omega)\neq\emptyset$,
\[
h^{(1)}
\lmk
\omega, \alpha , \theta,
(\tilde\beta_g), (\eta_{g}^\sigma),
(\alpha_L,\alpha_R,\Theta),\lmk (W_g), (u_\sigma(g,h))\rmk
\rmk
\]
 is independent of the choice of
\[
\alpha, \theta,
(\tilde\beta_g), (\eta_{g}^\sigma), 
(\alpha_L,\alpha_R,\Theta),
\lmk (W_g), (u_\sigma(g,h))\rmk.
\]
\end{thm}
\begin{defn}
Let $\omega\in\QLS$ with $\IG(\omega)\neq\emptyset$.
We denote the third cohomology given in Theorem \ref{welldefined} by
\[
h(\omega):=h^{(1)}
\lmk
\omega, \alpha , \theta,
(\tilde\beta_g), (\eta_{g}^\sigma),
(\alpha_L,\alpha_R,\Theta),\lmk (W_g), (u_\sigma(g,h))\rmk
\rmk.
\]
\end{defn}

First we show the independence from $\lmk (W_g), (u_\sigma(g,h))\rmk$.
\begin{lem}
Let
\begin{align}
&\omega\in\QLS, \; \alpha\in \eaut(\omega), \; 0<\theta<\frac\pi 2,\;
(\tilde\beta_g)\in \IG(\omega,\theta), \; (\eta_{g}^\sigma)\in \caT\lmk\theta, (\tilde\beta_g)\rmk,\;
(\alpha_L,\alpha_R,\Theta)\in \caD_{\alpha}^{\theta},\\
&\lmk (W_g), (u_\sigma(g,h))\rmk, \lmk (\tilde W_g), (\tilde u_\sigma(g,h))\rmk\in \IP\lmk
\omega, \alpha, \theta,
(\tilde\beta_g), (\eta_{g}^\sigma),
(\alpha_L,\alpha_R,\Theta)
\rmk.
\end{align}
Then
we have
\begin{align}\label{hq1}
h^{(1)}
\lmk
\omega, \alpha , \theta,
(\tilde\beta_g), (\eta_{g}^\sigma),
(\alpha_L,\alpha_R,\Theta),\lmk (W_g), (u_\sigma(g,h))\rmk
\rmk
=
h^{(1)}
\lmk
\omega, \alpha , \theta,
(\tilde\beta_g), (\eta_{g}^\sigma),
(\alpha_L,\alpha_R,\Theta),\lmk (\tilde W_g), (\tilde u_\sigma(g,h))\rmk
\rmk.
\end{align}
\begin{defn}
From this lemma and because there is always
$\lmk (W_g), (u_R(g,h))\rmk$
in $\IP\lmk
\omega, \alpha, \theta,
(\tilde\beta_g), (\eta_{g}^\sigma),
(\alpha_L,\alpha_R,\Theta)
\rmk$ by
 Lemma \ref{ichi},  we may
 define
 \begin{align}
 h^{(2)}
\lmk
\omega, \alpha , \theta,
(\tilde\beta_g), (\eta_{g}^\sigma),
(\alpha_L,\alpha_R,\Theta)
\rmk
:=
 h^{(1)}
\lmk
\omega, \alpha , \theta,
(\tilde\beta_g), (\eta_{g}^\sigma),
(\alpha_L,\alpha_R,\Theta),\lmk (W_g), (u_\sigma(g,h))\rmk
\rmk
 \end{align}
 for any 
 \begin{align}
 \omega\in\QLS, \; \alpha\in \eaut(\omega), \; 0<\theta<\frac\pi 2,\;
(\tilde\beta_g)\in \IG(\omega,\theta), \; (\eta_{g}^\sigma)\in \caT\lmk\theta, (\tilde\beta_g)\rmk,\;
(\alpha_L,\alpha_R,\Theta)\in \caD_{\alpha}^{\theta},
 \end{align}
 independent of the choice of $\lmk (W_g), (u_\sigma(g,h))\rmk$.
\end{defn}
\end{lem}
\begin{proof}
Because 
\begin{align}
&\Ad\lmk W_g\rmk\circ\pi_0
=\pi_0\circ\alpha_0\circ\Theta\circ\eta_g\beta_g^U\circ\Theta^{-1}\circ\alpha_0^{-1}
=\Ad\lmk \tilde W_g\rmk\circ\pi_0,\\
&\Ad\lmk u_R(g,h)\rmk\circ\pi_R
=\pi_R\circ\alpha_R\circ\eta_g^R\beta_g^{R U}
\eta_h^R\lmk\beta_g^{R U}\rmk^{-1}\lmk \eta_{gh}^R\rmk^{-1}
\circ\alpha_R^{-1}=\Ad\lmk \tilde u_R(g,h)\rmk\circ\pi_R
\end{align}
and $\pi_{0}$, $\pi_{R}$ are irreducible, 
there are $b(g),a(g,h)\in\bbT$, $g,h\in G$ such that
\begin{align}\label{propwu}
W_g=b(g)\tilde W_g,\quad
\tilde u_R(g,h)=a(g,h)u_R(g,h).
\end{align}
Set 
\begin{align}
&c_{R}:=c_R\lmk
\omega, \alpha, \theta,
(\tilde\beta_g), (\eta_{g}^\sigma),
(\alpha_L,\alpha_R,\Theta),\lmk (W_g), (u_\sigma(g,h))\rmk
\rmk,\notag\\
&\tilde c_R:=c_{R}\lmk
\omega, \alpha, \theta,
(\tilde\beta_g), (\eta_{g}^\sigma),
(\alpha_L,\alpha_R,\Theta),\lmk (\tilde W_g), (\tilde u_\sigma(g,h))\rmk
\rmk.
\end{align}
Then from the definition of these values and (\ref{propwu}), we have
\begin{align}
&a(g,h)a(gh,k)\lmk \unit_{\caH_L}\otimes u_R(g,h) u_R(gh,k)\rmk
=
\unit_{\caH_L}\otimes \tilde u_R(g,h) \tilde u_R(gh,k)\notag\\
&=\tilde c_R(g,h,k)
\lmk \tilde W_g\lmk \unit_{\caH_L}\otimes \tilde u_R(h,k)\rmk \tilde W_g^*\rmk
\lmk \unit_{\caH_L}\otimes \tilde u_R(g,hk)\rmk\notag\\
&=\tilde c_R(g,h,k)a(h,k)a(g,hk)
\lmk W_g\lmk \unit_{\caH_L}\otimes u_R(h,k)\rmk  W_{g}^{*}\rmk
\lmk \unit_{\caH_L}\otimes u_R(g,hk)\rmk\notag\\
&=\tilde c_R(g,h,k)a(h,k)a(g,hk)\overline{c_{R}(g,h,k)}
\lmk \unit_{\caH_L}\otimes u_R(g,h) u_R(gh,k)\rmk.
\end{align}
Hence we have $\tilde c_{R}(g,h,k)=c_R(g,h,k)\overline{a(h,k)a(g,hk)}a(g,h)a(gh,k)$, and
we get $[c_{R}]_{H^{3}(G,\bbT)}=[\tilde c_{R}]_{H^{3}(G,\bbT)}$, proving the claim.
\end{proof}
Next  we show the independence from $\alpha,
(\alpha_L,\alpha_R,\Theta)$.
\begin{lem}
Let
\begin{align}
&\omega\in\QLS, \; \alpha_{1},\alpha_{2}\in \eaut(\omega), \; 0<\theta<\frac\pi 2,\;
(\tilde\beta_g)\in \IG(\omega,\theta), \; (\eta_{g}^\sigma)\in \caT\lmk\theta, (\tilde\beta_g)\rmk,\\
&
(\alpha_{L,1},\alpha_{R,1},\Theta_{1})\in \caD_{\alpha_{1}}^{\theta},\quad
(\alpha_{L,2},\alpha_{R,2},\Theta_{2})\in \caD_{\alpha_{2}}^{\theta}.
\end{align}
Then
we have
\begin{align}\label{hq2}
h^{(2)}
\lmk
\omega, \alpha_{1} , \theta,
(\tilde\beta_g), (\eta_{g}^\sigma),
(\alpha_{L,1},\alpha_{R,1},\Theta_{1})
\rmk
=
h^{(2)}
\lmk
\omega, \alpha_{2} , \theta,
(\tilde\beta_g), (\eta_{g}^\sigma),
(\alpha_{L,2},\alpha_{R,2},\Theta_{2})
\rmk.
\end{align}
\end{lem}\begin{defn}
From this lemma
and because there are always 
$ \alpha\in \eaut(\omega)$
and $(\alpha_L,\alpha_R,\Theta)\in \caD_{\alpha}^{\theta}$ 
for $\omega\in \QLS$ and $0<\theta<\frac\pi 2$ by the definition,
we may
 define
 \begin{align}
 h^{(3)}
\lmk
\omega, \theta,
(\tilde\beta_g), (\eta_{g}^\sigma)
\rmk
:=
 h^{(2)}
\lmk
\omega, \alpha , \theta,
(\tilde\beta_g), (\eta_{g}^\sigma),
(\alpha_{L},\alpha_{R},\Theta)
\rmk
  \end{align}
   for any 
 \begin{align}
 \omega\in\QLS, \;, \; 0<\theta<\frac\pi 2,\;
(\tilde\beta_g)\in \IG(\omega,\theta), \; (\eta_{g}^\sigma)\in \caT\lmk\theta, (\tilde\beta_g)\rmk,
 \end{align}
 independent of the choice of $\alpha,
(\alpha_L,\alpha_R,\Theta)$.
\end{defn}
\begin{proof}
By Lemma \ref{ichi}, there are 
\begin{align}
\lmk (W_{g,1}), (u_{\sigma,1}(g,h))\rmk
\in \IP\lmk
\omega, \alpha_{1}, \theta,
(\tilde\beta_g), (\eta_{g}^\sigma),
(\alpha_{L,1},\alpha_{R,1},\Theta_1)
\rmk.
\end{align}
For each $i=1,2$, 
we have $\Theta_{i}\in\Aut\caA_{C_{\theta}^{c}}$
and
\begin{align}\label{aii}
\alpha_{i}=\inn\circ\alpha_{0,i}\circ\Theta_{i},
\end{align}
setting
\begin{align}
\alpha_{0,i}:=\alpha_{L,i}\otimes \alpha_{R,i}.
\end{align}
Because $\omega_{0}\circ\alpha_{1}=\omega=\omega_{0}\circ\alpha_{2}$,
we have $\omega_{0}\circ\alpha_{2}\circ\alpha_{1}^{-1}=\omega_{0}$.
Therefore, there is a unitary $\tilde V$ on $\caH_{0}$
such that
$
\pi_{0}\circ \alpha_{2}\circ\alpha_{1}^{-1}=\Ad\lmk \tilde V\rmk\circ\pi_{0}
$.
Substituting (\ref{aii}) to this,
we see that there is a unitary $V$ on $\caH_{0}$
satisfying
\begin{align}
\pi_{0}\circ \alpha_{0, 2}\circ\Theta_{2}=\Ad\lmk  V\rmk\circ\pi_{0}\circ \alpha_{0, 1}\circ\Theta_{1}.
\end{align}
From this, we obtain
\begin{align}\label{wno}
\begin{split}
&\Ad\lmk V W_{g,1} V^{*}\rmk\circ\pi_{0}
=\Ad\lmk V W_{g,1}\rmk \pi_{0}\circ \alpha_{0, 1}\circ\Theta_{1}
\circ \Theta_{2}^{-1}\circ  \alpha_{0, 2}^{-1}\\
&=\Ad\lmk V\rmk\circ \pi_{0}
\circ\alpha_{0,1}\circ\Theta_{1}\circ\eta_g\beta_g^U\circ\Theta_{1}^{-1}\circ\alpha_{0,1}^{-1}
\circ \alpha_{0, 1}\circ\Theta_{1}
\circ \Theta_{2}^{-1}\circ  \alpha_{0, 2}^{-1}\\
&=\pi_{0}\circ\alpha_{0, 2}\circ\Theta_{2}\circ \Theta_{1}^{-1}\circ \alpha_{0, 1}^{-1}
\circ\alpha_{0,1}\circ\Theta_{1}\circ\eta_g\beta_g^U\circ\Theta_{1}^{-1}\circ\alpha_{0,1}^{-1}
\circ \alpha_{0, 1}\circ\Theta_{1}
\circ \Theta_{2}^{-1}\circ  \alpha_{0, 2}^{-1}\\
&=\pi_{0}\circ\alpha_{0, 2}\circ\Theta_{2}
\circ\eta_g\beta_g^U
\circ \Theta_{2}^{-1}\circ  \alpha_{0, 2}^{-1},
\end{split}
\end{align}
for all $g\in G$.
Furthermore, we have
\begin{align}
&\Ad\lmk
V\lmk\unit_{\caH_{L}}\otimes u_{R,1}(g,h)\rmk V^{*}
\rmk\circ\pi_{0}
=
\Ad\lmk
V\lmk\unit_{\caH_{L}}\otimes u_{R,1}(g,h)\rmk 
\rmk\circ
\pi_{0}\circ \alpha_{0, 1}\circ\Theta_{1}
\circ \Theta_{2}^{-1}\circ  \alpha_{0, 2}^{-1}\notag\\
&=\Ad\lmk
V\rmk\circ
\pi_{0}\circ 
\lmk \id_{\caA_{H_{L}}}\otimes \alpha_{R,1} \eta_g^R\beta_g^{R U}
\eta_h^R\lmk\beta_g^{R U}\rmk^{-1}\lmk \eta_{gh}^R\rmk^{-1}\alpha_{R,1}^{-1}\rmk
\alpha_{0, 1}\circ\Theta_{1}
\circ \Theta_{2}^{-1}\circ  \alpha_{0, 2}^{-1}\notag\\
&=
\pi_{0}\circ \alpha_{0, 2}\circ\Theta_{2}\circ \Theta_{1}^{-1}
\circ \alpha_{0, 1}^{-1}\lmk \id_{\caA_{H_{L}}}\otimes \alpha_{R,1} \eta_g^R\beta_g^{R U}
\eta_h^R\lmk\beta_g^{R U}\rmk^{-1}\lmk \eta_{gh}^R\rmk^{-1}\alpha_{R,1}^{-1}\rmk
\circ\alpha_{0, 1}\circ\Theta_{1}
\circ \Theta_{2}^{-1}\circ  \alpha_{0, 2}^{-1}
\notag\\
&=
\pi_{0}\circ \alpha_{0, 2}\circ\Theta_{2}\circ \Theta_{1}^{-1}
\circ \lmk \id_{\caA_{H_{L}}}\otimes  \eta_g^R\beta_g^{R U}
\eta_h^R\lmk\beta_g^{R U}\rmk^{-1}\lmk \eta_{gh}^R\rmk^{-1}\rmk
\circ\Theta_{1}
\circ \Theta_{2}^{-1}\circ  \alpha_{0, 2}^{-1}
\label{vuv}
\end{align}
Now, because $\eta_g^R\beta_g^{R U}
\eta_h^R\lmk\beta_g^{R U}\rmk^{-1}\lmk \eta_{gh}^R\rmk^{-1}$
is an automorphism on $\caA_{C_{\theta}}$
and $\Theta_{2}\circ \Theta_{1}^{-1}$  is an automorphism
on  $\caA_{C_{\theta}^{c}}$, they commute.
Therefore, we have
\begin{align}
&\Ad\lmk
V\lmk\unit_{\caH_{L}}\otimes u_{R,1}(g,h)\rmk V^{*}
\rmk\circ\pi_{0}\notag\\
&=
(\ref{vuv})=
\pi_{0}\circ \alpha_{0, 2}
\circ \lmk \id_{\caA_{H_{L}}}\otimes  \eta_g^R\beta_g^{R U}
\eta_h^R\lmk\beta_g^{R U}\rmk^{-1}\lmk \eta_{gh}^R\rmk^{-1}\rmk
\circ  \alpha_{0, 2}^{-1}\notag\\
&=
\pi_{L}\otimes 
\lmk \pi_{R}\circ \alpha_{R,2}\eta_{g}^{R}\beta_{g}^{RU}
\eta_{h}^{R}\lmk \beta_{g}^{RU}\rmk^{-1}
\lmk \eta_{gh}^{R}\rmk^{-1}\lmk \alpha_{R,2}\rmk^{-1}\rmk.
\end{align}
From this equality and the fact that $\pi_{L}$ is irreducible,
we see that $V\lmk\unit_{\caH_{L}}\otimes u_{R,1}(g,h)\rmk V^{*}$ is of the form
$\unit_{\caH_{L}}\otimes u_{R,2}(g,h)$ with some unitary $u_{R,2}(g,h)$ on $\caH_{R}$.
This $u_{R,2}(g,h)$ satisfies
\begin{align}\label{urno}
\Ad\lmk u_{R,2}(g,h)\rmk\circ\pi_{R}
=\pi_{R}\circ \alpha_{R,2}\eta_{g}^{R}\beta_{g}^{RU}
\eta_{h}^{R}\lmk \beta_{g}^{RU}\rmk^{-1}
\lmk \eta_{gh}^{R}\rmk^{-1}\lmk \alpha_{R,2}\rmk^{-1}.
\end{align}
Analogously, we obtain a
unitary $u_{L,2}(g,h)$ on $\caH_{L}$
such that
\begin{align}\label{ulno}
&V\lmk
 u_{L,1}(g,h)\otimes \unit_{\caH_{R}}\rmk V^{*}= u_{L,2}(g,h)\otimes\unit_{\caH_{R}},\\
&\Ad\lmk u_{L,2}(g,h)\rmk\circ\pi_{L}
=\pi_{L}\circ \alpha_{L,2}\eta_{g}^{L}\beta_{g}^{LU}
\eta_{h}^{L}\lmk \beta_{g}^{LU}\rmk^{-1}
\lmk \eta_{gh}^{L}\rmk^{-1}\lmk \alpha_{L,2}\rmk^{-1}.
\end{align}
From (\ref{wno}), (\ref{urno}), (\ref{ulno}), we see that 
\begin{align}
\lmk  \lmk V W_{g,1} V^{*}\rmk , \lmk u_{\sigma,2}(g,h)\rmk\rmk\in 
 \IP\lmk
\omega, \alpha_{2}, \theta,
(\tilde\beta_g), (\eta_{g}^\sigma),
(\alpha_{L,2},\alpha_{R,2},\Theta)
\rmk.
\end{align}
Set 
\begin{align}
&c_{R,1}:=c_R\lmk
\omega, \alpha_{1}, \theta,
(\tilde\beta_g), (\eta_{g}^\sigma),
(\alpha_{L,1},\alpha_{R,1},\Theta_{1}),\lmk (W_{g,1}), (u_{\sigma,1}(g,h))\rmk
\rmk,\notag\\
&c_{R,2}:=c_{R}\lmk
\omega, \alpha_{2}, \theta,
(\tilde\beta_g), (\eta_{g}^\sigma),
(\alpha_{L,2},\alpha_{R,2},\Theta_{2}),
\lmk V W_{g,1} V^{*}\rmk , \lmk u_{\sigma,2}(g,h)\rmk
\rmk.
\end{align}
It suffices to show that $c_{R,1}=c_{R,2}$.
This can be checked directly as follows:
\begin{align}
&V
\lmk \unit_{\caH_L}\otimes u_{R,1}(g,h) u_{R,1}(gh,k)
\rmk
V^{*}
=
\unit_{\caH_L}\otimes u_{R,2}(g,h) u_{R,2}(gh,k)\notag\\
&=c_{R,2}(g,h,k)
\lmk V W_{g,1}V^{*}\lmk \unit_{\caH_L}\otimes u_{R,2}(h,k)\rmk VW_{g,1}^*V^{*}\rmk
\lmk \unit_{\caH_L}\otimes u_{R,2}(g,hk)\rmk\notag\\
&=c_{R,2}(g,h,k)V
\lmk W_{g,1}\lmk \unit_{\caH_L}\otimes u_{R,1}(h,k)\rmk W_{g,1}^{*}\rmk
\lmk \unit_{\caH_L}\otimes u_{R,1}(g,hk)\rmk V^{*}\notag\\
&=
c_{R,2}(g,h,k)\overline{c_{R,1}(g,h,k)}V
\lmk
\unit_{\caH_L}\otimes u_{R,1}(g,h) u_{R,1}(gh,k)
\rmk
V^{*}.
\end{align}
\end{proof}
\begin{lem}
Let
\begin{align}
&\omega\in\QLS,  \; 0<\theta<\frac\pi 2,\;
(\tilde\beta_g)\in \IG(\omega,\theta), \; (\eta_{g}^\sigma), (\tilde \eta_{g}^\sigma)\in \caT\lmk\theta, (\tilde\beta_g)\rmk.
\end{align}
Then
we have
\begin{align}\label{hq3}
h^{(3)}
\lmk
\omega , \theta,
(\tilde\beta_g), (\eta_{g}^\sigma)
\rmk
=
h^{(3)}
\lmk
\omega, \theta,
(\tilde\beta_g), (\tilde \eta_{g}^\sigma)
\rmk.
\end{align}
\end{lem}
\begin{defn}
From this lemma and the definition of $ \IG(\omega,\theta)$,
we may
 define
 \begin{align}
 h^{(4)}
\lmk
\omega, \theta,
(\tilde\beta_g)
\rmk
:=
 h^{(3)}
\lmk
\omega, \theta,
(\tilde\beta_g), (\eta_{g}^\sigma)
\rmk
  \end{align}
   for any 
 \begin{align}
 \omega\in\QLS,  \; 0<\theta<\frac\pi 2,\;
(\tilde\beta_g)\in \IG(\omega,\theta), \; (\eta_{g}^\sigma)\in \caT\lmk\theta, (\tilde\beta_g)\rmk,
 \end{align}
 independent of the choice of $(\eta_{g}^\sigma)$.
\end{defn}
\begin{proof}
There are
$ \alpha\in \eaut(\omega)$
and $(\alpha_L,\alpha_R,\Theta)\in \caD_{\alpha}^{\theta}$ 
for $\omega\in \QLS$ by the definition.
We set $\alpha_0:=\alpha_L\otimes\alpha_R$ and
$\eta_g:=\eta_g^L\otimes\eta_g^R$,
$\tilde \eta_g:=\tilde \eta_g^L\otimes\tilde \eta_g^R$.
By Lemma \ref{ichi}, there is some
\begin{align}
\lmk (W_{g}), (u_{\sigma}(g,h))\rmk
\in \IP\lmk
\omega, \alpha, \theta,
(\tilde\beta_g), (\eta_{g}^\sigma),
(\alpha_{L},\alpha_{R},\Theta)
\rmk.
\end{align}

Because $ (\eta_{g}^\sigma), (\tilde \eta_{g}^\sigma)\in \caT\lmk\theta, (\tilde\beta_g)\rmk$,
we have
\begin{align}
\tilde\beta_g=\inn\circ\lmk \eta_g^L\otimes\eta_g^R\rmk\circ\beta_g^U
=\inn \circ\lmk\tilde  \eta_g^L\otimes\tilde \eta_g^R\rmk\circ\beta_g^U.
\end{align}
From this, we obtain
\begin{align}
\tilde  \eta_g^L\circ\lmk\eta_g^L\rmk^{-1}\otimes\tilde \eta_g^R\circ\lmk\eta_g^R\rmk^{-1}
=\inn,
\end{align}
hence
there are unitaries $v_g^\sigma\in\caA_{H_\sigma}$, $\sigma=L,R$
such that
\begin{align}
\tilde  \eta_g^\sigma\circ\lmk\eta_g^\sigma\rmk^{-1}
=\Ad\lmk v_g^\sigma\rmk.
\end{align}
Because $\tilde\eta_g^\sigma, \eta_g^\sigma$ are
automorphisms on $\caA_{C_\theta,\sigma}$,
$v_g^\sigma$ belongs to $\caA_{C_\theta,\sigma}$.
(See Lemma \ref{bimp}.)
Setting $v_g:=v_g^L\otimes v_g^R$,
we obtain 
$\tilde\eta_g=\Ad\lmk v_g\rmk\circ\eta_g$.

Set
\begin{align}
&\tilde W_g:=\lmk \lmk
\pi_L\alpha_L\lmk v_g^L\rmk\rmk \otimes \lmk \pi_R\alpha_R\lmk v_g^R\rmk\rmk\rmk
W_g,\\
&\tilde u_\sigma(g,h)
:=
\pi_\sigma\lmk
\alpha_\sigma\lmk v_g^\sigma \cdot \lmk \eta_g^\sigma\beta_g^{\sigma U}\rmk\lmk v_h^\sigma\rmk\rmk
\rmk\cdot
u_\sigma\lmk g,h\rmk\cdot
\pi_\sigma\lmk\alpha_\sigma\lmk \lmk v_{gh}^{\sigma}\rmk^*\rmk\rmk
\end{align}
for each $g,h\in G$ and $\sigma=L,R$.
We claim that 
\begin{align}\label{tip}
\lmk\lmk \tilde W_g\rmk,  \lmk \tilde u_\sigma(g,h)\rmk
\rmk
\in
\IP\lmk
\omega, \alpha, \theta,
(\tilde\beta_g), (\tilde \eta_{g}^\sigma),
(\alpha_L,\alpha_R,\Theta)
\rmk.
\end{align}
First, we have
\begin{align}
\begin{split}
&\pi_0\circ\alpha_0\circ\Theta\circ\tilde \eta_g\beta_g^U\circ\Theta^{-1}\circ\alpha_0^{-1}
=\pi_0\circ\alpha_0\circ\Theta\circ
\Ad\lmk v_g\rmk\circ\eta_g
\beta_g^U\circ\Theta^{-1}\circ\alpha_0^{-1}\\
&=\pi_0\circ\alpha_0\circ\Ad\lmk v_g\rmk\circ\Theta
\circ\eta_g
\beta_g^U\circ\Theta^{-1}\circ\alpha_0^{-1}\\
&=
\Ad\lmk
 \lmk
\pi_L\alpha_L\lmk v_g^L\rmk\rmk \otimes \lmk \pi_R\alpha_R\lmk v_g^R\rmk\rmk
\rmk
\pi_0\circ\alpha_0\circ\Theta\circ\eta_g\beta_g^U\circ\Theta^{-1}\circ\alpha_0^{-1}
=
\Ad\lmk \tilde W_g\rmk\circ\pi_0.
\end{split}
%&\pi_0\circ\alpha_0\circ\Theta\circ\tilde \eta_g\beta_g^U\circ\Theta^{-1}\circ\alpha_0^{-1}
%=\pi_0\circ\alpha_0\circ\Theta\circ
%\Ad\lmk v_g\rmk\circ\eta_g
%\beta_g^U\circ\Theta^{-1}\circ\alpha_0^{-1}\\
%&=\pi_0\circ\alpha_0\circ\Ad\lmk v_g\rmk\circ\Theta
%\circ\eta_g
%\beta_g^U\circ\Theta^{-1}\circ\alpha_0^{-1}\\
%&=
%\Ad\lmk
% \lmk
%\pi_L\alpha_L\lmk v_g^L\rmk\rmk \otimes \lmk \pi_R\alpha_R\lmk v_g^R\rmk\rmk
%\rmk
%\pi_0\circ\alpha_0\circ\Theta\circ\eta_g\beta_g^U\circ\Theta^{-1}\circ\alpha_0^{-1}
%=
%\Ad\lmk \tilde W_g\rmk\circ\pi_0.
\end{align}
For the first equality, we substituted $\tilde\eta_g=\Ad\lmk v_g\rmk\circ\eta_g$, 
and for the second equality, we used the fact that $v_g^\sigma$ belongs to $\caA_{C_\theta,\sigma}$,
while $\Theta$ is an automorphism on $\caA_{\lmk C_\theta\rmk^c,\sigma}$.
The last equality follows from the definition of $W_g$.
On the other hand,
we have
\begin{align}
&\pi_{\sigma}\circ\alpha_{\sigma}\circ{\tilde\eta}_g^{\sigma}\beta_g^{{\sigma} U}
{\tilde\eta}_h^{\sigma}\lmk\beta_g^{{\sigma} U}\rmk^{-1}\lmk {\tilde\eta}_{gh}^{\sigma}\rmk^{-1}
\circ\alpha_{\sigma}^{-1}\notag\\
&=
\pi_{\sigma}\circ\alpha_{\sigma}\circ\Ad\lmk v_g^\sigma\rmk\circ\eta_g^{\sigma}\beta_g^{{\sigma} U}
\Ad\lmk v_h^\sigma\rmk\circ\eta_h^{\sigma}\lmk\beta_g^{{\sigma} U}\rmk^{-1}\lmk \eta_{gh}^{\sigma}\rmk^{-1}
\Ad\lmk {v_{gh}^\sigma}^*\rmk
\circ\alpha_{\sigma}^{-1}\notag\\
&=
\Ad\lmk
\pi_{\sigma}\circ\alpha_{\sigma}
\lmk
\lmk v_g^\sigma\rmk
  \eta_g^{\sigma}\beta_g^{{\sigma} U}\lmk v_h^\sigma\rmk\rmk
  \rmk
\pi_{\sigma}\circ\alpha_{\sigma}\eta_g^{\sigma}\beta_g^{{\sigma} U}
\eta_h^{\sigma}\lmk\beta_g^{{\sigma} U}\rmk^{-1}\lmk \eta_{gh}^{\sigma}\rmk^{-1}
\circ\alpha_{\sigma}^{-1}\circ
\Ad\lmk\alpha_\sigma\lmk  {v_{gh}^\sigma}^*\rmk\rmk\notag\\
&=
\Ad\lmk
\pi_{\sigma}\circ\alpha_{\sigma}
\lmk
\lmk v_g^\sigma\rmk
  \eta_g^{\sigma}\beta_g^{{\sigma} U}\lmk v_h^\sigma\rmk\rmk\rmk
\circ \Ad\lmk u_\sigma(g,h)\rmk
 \pi_\sigma\circ \Ad\lmk\alpha_\sigma\lmk  {v_{gh}^\sigma}^*\rmk\rmk
=
\Ad\lmk \tilde u_{\sigma}(g,h)\rmk\circ\pi_{\sigma},
\end{align}
for all $g,h\in G$.
For the first equality, we substituted $\tilde\eta_g=\Ad\lmk v_g\rmk\circ\eta_g$.
The third equality is the definition of $u(g,h)$.
Hence we have proven (\ref{tip}).

Set 
\begin{align}
&c_{R}:=c_R\lmk
\omega, \alpha, \theta,
(\tilde\beta_g), (\eta_{g}^\sigma),
(\alpha_L,\alpha_R,\Theta),\lmk (W_g), (u_\sigma(g,h))\rmk
\rmk,\notag\\
&\tilde c_R:=c_{R}\lmk
\omega, \alpha, \theta,
(\tilde\beta_g), (\tilde \eta_{g}^\sigma),
(\alpha_L,\alpha_R,\Theta),\lmk (\tilde W_g), (\tilde u_\sigma(g,h))\rmk
\rmk.
\end{align}
In order to show the statement of the Lemma, it suffices to show that $c_R=\tilde c_R$.
Substituting the definition of $\tilde u_R$, we obtain 
\begin{align}
&
%\unit_{\caH_L}\otimes
 \tilde u_R(g,h) \tilde u_R(gh,k)\notag\\
&=
%\unit_{\caH_L}\otimes
%\lmk
\pi_R\lmk
\alpha_R\lmk v_g^R \lmk \eta_g^R\beta_g^{R U}\rmk\lmk v_h^R\rmk\rmk
\rmk\cdot
u_R\lmk g,h\rmk\cdot
\pi_R\lmk\alpha_R\lmk \lmk v_{gh}^{R}\rmk^*\rmk\rmk
\pi_R\lmk
\alpha_R\lmk v_{gh}^R \lmk \eta_{gh}^R\beta_{gh}^{R U}\rmk\lmk v_k^R\rmk\rmk
\rmk\cdot
u_R\lmk gh,k\rmk\cdot
\pi_R\lmk\alpha_R\lmk \lmk v_{ghk}^{R}\rmk^*\rmk\rmk
%\rmk
\notag\\
&=
%\unit_{\caH_L}\otimes
%\lmk
\pi_R\lmk
\alpha_R\lmk v_g^R \lmk \eta_g^R\beta_g^{R U}\rmk\lmk v_h^R\rmk\rmk
\rmk\cdot
\lcm
u_R\lmk g,h\rmk\cdot
\pi_R\lmk
\alpha_R\lmk \lmk \eta_{gh}^R\beta_{gh}^{R U}\rmk\lmk v_k^R\rmk\rmk
\rmk
\rcm
u_R\lmk gh,k\rmk
\pi_R\lmk\alpha_R\lmk \lmk v_{ghk}^{R}\rmk^*\rmk\rmk
%\rmk
\notag\\
&=
%\unit_{\caH_L}\otimes
%\lmk
\pi_R\lmk
\alpha_R\lmk v_g^R \lmk \eta_g^R\beta_g^{R U}\rmk\lmk v_h^R\rmk\rmk
\rmk\cdot 
\lcm
\Ad\lmk u_R\lmk g,h\rmk\rmk\lmk
\pi_R\lmk
\alpha_R\lmk \lmk \eta_{gh}^R\beta_{gh}^{R U}\rmk\lmk v_k^R\rmk\rmk
\rmk\rmk
\cdot u_R\lmk g,h\rmk
\rcm
u_R\lmk gh,k\rmk\cdot
\pi_R\lmk\alpha_R\lmk \lmk v_{ghk}^{R}\rmk^*\rmk\rmk
%\rmk
\notag\\
&=
%\lmk
\pi_R\lmk
\alpha_R\lmk v_g^R \lmk \eta_g^R\beta_g^{R U}\rmk\lmk v_h^R\rmk\rmk
\rmk\cdot 
\lmk
\pi_R\lmk\alpha_R\circ\eta_g^R\beta_g^{R U}
\eta_h^R\lmk\beta_g^{R U}\rmk^{-1}\lmk \eta_{gh}^R\rmk^{-1}
\circ\alpha_R^{-1}
\alpha_R\lmk \lmk \eta_{gh}^R\beta_{gh}^{R U}\rmk\lmk v_k^R\rmk\rmk
\rmk\rmk\notag\\
&\cdot u_R\lmk g,h\rmk
u_R\lmk gh,k\rmk\cdot
\pi_R\lmk\alpha_R\lmk \lmk v_{ghk}^{R}\rmk^*\rmk\rmk\notag\\
%\rmk
&=\pi_R\lmk
\alpha_R\lmk v_g^R \cdot \lmk \eta_g^R\beta_g^{R U}\rmk\lmk v_h^R\rmk\cdot
\eta_g^R\beta_g^{R U}
\eta_h^R\beta_{h}^{R U}\lmk v_k^R\rmk\rmk
\rmk\cdot u_R\lmk g,h\rmk
u_R\lmk gh,k\rmk\cdot
\pi_R\lmk\alpha_R\lmk \lmk v_{ghk}^{R}\rmk^*\rmk\rmk
\end{align}
For the fourth equality, we used the definition of $u_R$.
From the above equation, applying (\ref{uwc}) to the $[\cdot]$ part below, we have
\begin{align}\label{uue}
&\unit_{\caH_L}\otimes
 \tilde u_R(g,h) \tilde u_R(gh,k)\notag\\
&=
\unit_{\caH_L}\otimes
\pi_R\lmk
\alpha_R\lmk v_g^R \cdot \lmk \eta_g^R\beta_g^{R U}\rmk\lmk v_h^R\rmk\cdot
\eta_g^R\beta_g^{R U}
\eta_h^R\beta_{h}^{R U}\lmk v_k^R\rmk\rmk
\rmk\cdot \lcm
u_R\lmk g,h\rmk
u_R\lmk gh,k\rmk\rcm\cdot
\pi_R\lmk\alpha_R\lmk \lmk v_{ghk}^{R}\rmk^*\rmk\rmk\notag\\
&=c_R(g,h,k)\lmk
\unit_{\caH_L}\otimes
\pi_R\lmk
\alpha_R\lmk v_g^R \cdot \lmk \eta_g^R\beta_g^{R U}\rmk\lmk v_h^R\rmk\cdot
\eta_g^R\beta_g^{R U}
\eta_h^R\beta_{h}^{R U}\lmk v_k^R\rmk\rmk
\rmk\rmk\notag\\
&
\left\{ W_g\lmk \unit_{\caH_L}\otimes u_R(h,k)\rmk W_g^*\right\}
\lmk \unit_{\caH_L}\otimes u_R(g,hk)\rmk
\cdot
\pi_R\lmk\alpha_R\lmk \lmk v_{ghk}^{R}\rmk^*\rmk\rmk.
\end{align}
Now from the definition of $\tilde u_R$, the $\{\cdot\}$ part above becomes
\begin{align}\label{shimokita}
& W_g\lmk \unit_{\caH_L}\otimes u_R(h,k)\rmk W_g^*\notag\\
 &=
 \Ad\lmk W_g\rmk\circ\pi_0\circ\lmk
 \id_L\otimes
\alpha_R\lmk \lmk v_h^R \cdot  \eta_h^R\beta_h^{R U}\lmk v_k^R\rmk\rmk^*\rmk
\rmk
\cdot
\Ad\lmk W_g\rmk\lmk\unit_{\caH_L}\otimes \tilde u_R\lmk h,k\rmk\rmk\cdot\lmk
\Ad\lmk W_g\rmk\pi_0\lmk \id_L\otimes\alpha_R\lmk  v_{hk}^{R}\rmk\rmk
 \rmk.
\end{align}
Because  $v_g^R$ belongs to $\caA_{C_\theta,R}$ and
$\eta_g^R$ is an automorphism on $\caA_{C_\theta,R}$
while $\Theta$ is an automorphism on $\caA_{\lmk C_\theta\rmk^c}$
and $\beta_g^U\lmk \caA_{C_\theta,R}\rmk=\caA_{C_\theta,R}$,
we have
\begin{align}
&\Ad\lmk W_g\rmk\circ\pi_0\circ\lmk
 \id_L\otimes
\alpha_R\lmk \lmk v_h^R \cdot  \eta_h^R\beta_h^{R U}\lmk v_k^R\rmk\rmk^*\rmk
\rmk
=
\pi_0\circ\alpha_0\circ\Theta\circ\eta_g\beta_g^U\circ\Theta^{-1}\circ\alpha_0^{-1}
\circ
\lmk
 \id_L\otimes
\alpha_R\lmk \lmk v_h^R \cdot  \eta_h^R\beta_h^{R U}\lmk v_k^R\rmk\rmk^*\rmk
\rmk\notag\\
&=
\pi_0
\lmk
 \id_L\otimes\alpha_R\circ \eta_g^R\beta_g^{R U}
\lmk \lmk v_h^R \cdot \eta_h^R\beta_h^{R U}\lmk v_k^R\rmk\rmk^*\rmk
\rmk,\quad \text{and}\notag\\
&\Ad\lmk W_g\rmk\circ\pi_0\lmk \id_L\otimes\alpha_R\lmk  v_{hk}^{R}\rmk\rmk
=\pi_0
\lmk
 \id_L\otimes\alpha_R\circ \eta_g^R\beta_g^{R U}
\lmk 
 v_{hk}^{R}
\rmk
\rmk.
\end{align}
Substituting this to (\ref{shimokita}), we obtain
\begin{align}
& W_g\lmk \unit_{\caH_L}\otimes u_R(h,k)\rmk W_g^*\notag\\
 &=
 \pi_0
\lmk
 \id_L\otimes\alpha_R\circ \eta_g^R\beta_g^{R U}
\lmk \lmk v_h^R \cdot \eta_h^R\beta_h^{R U}\lmk v_k^R\rmk\rmk^*\rmk
\rmk
 \cdot
\Ad\lmk W_g\rmk\lmk\unit_{\caH_L}\otimes \tilde u_R\lmk h,k\rmk\rmk\cdot
\pi_0
\lmk
 \id_L\otimes\alpha_R\circ \eta_g^R\beta_g^{R U}
\lmk 
 v_{hk}^{R}
\rmk
\rmk.
\end{align}
Substituting this, $\{\cdot\}$ part of (\ref{uue}), we obtain
\begin{align}
&\unit_{\caH_L}\otimes
 \tilde u_R(g,h) \tilde u_R(gh,k)\notag\\
 &=c_R(g,h,k)\lmk
\unit_{\caH_L}\otimes
\pi_R\lmk
\alpha_R\lmk v_g^R \cdot \lmk \eta_g^R\beta_g^{R U}\rmk\lmk v_h^R\rmk\cdot
\eta_g^R\beta_g^{R U}
\eta_h^R\beta_{h}^{R U}\lmk v_k^R\rmk\rmk
\rmk\rmk\notag\\
&
\pi_0
\lmk
 \id_L\otimes\alpha_R\circ \eta_g^R\beta_g^{R U}
\lmk \lmk v_h^R \cdot \eta_h^R\beta_h^{R U}\lmk v_k^R\rmk\rmk^*\rmk
\rmk
 \cdot
\Ad\lmk W_g\rmk\lmk\unit_{\caH_L}\otimes \tilde u_R\lmk h,k\rmk\rmk\cdot
\pi_0
\lmk
 \id_L\otimes\alpha_R\circ \eta_g^R\beta_g^{R U}
\lmk 
 v_{hk}^{R}
\rmk
\rmk\notag\\
&\lmk \unit_{\caH_L}\otimes u_R(g,hk)\cdot
\pi_R\lmk\alpha_R\lmk \lmk v_{ghk}^{R}\rmk^*\rmk\rmk\rmk
\notag\\
&
=c_R(g,h,k)\lmk
\unit_{\caH_L}\otimes
\pi_R\lmk
\alpha_R\lmk v_g^R \rmk
\rmk\rmk
\Ad\lmk W_g\rmk\lmk\unit_{\caH_L}\otimes \tilde u_R\lmk h,k\rmk\rmk\cdot
\pi_0
\lmk
 \id_L\otimes\alpha_R\circ \eta_g^R\beta_g^{R U}
\lmk 
 v_{hk}^{R}
\rmk
\rmk\notag\\
&
\pi_0\circ\lmk
 \id_L\otimes
\alpha_R\lmk \lmk v_g^R \cdot  \eta_g^R\beta_g^{R U}\lmk v_{hk}^R\rmk\rmk^*\rmk
\rmk
\cdot
\lmk\unit_{\caH_L}\otimes \tilde u_R\lmk g,hk\rmk\cdot\rmk
\pi_0\lmk \id_L\otimes\alpha_R\lmk  v_{ghk}^{R}\rmk
\alpha_R\lmk \lmk v_{ghk}^{R}\rmk^*\rmk
\rmk \notag\\
&=
c_R(g,h,k)
\Ad\lmk 
\lmk
\unit_{\caH_L}\otimes
\pi_R\lmk
\alpha_R\lmk v_g^R \rmk
\rmk\rmk W_g\rmk\lmk\unit_{\caH_L}\otimes \tilde u_R\lmk h,k\rmk\rmk\cdot
\lmk\unit_{\caH_L}\otimes \tilde u_R\lmk g,hk\rmk\rmk
\notag\\
 &=c_R(g,h,k)
 \left\{\Ad\lmk 
\pi_L\lmk
\alpha_L\lmk{ v_g^L}^* 
\rmk\otimes\unit_{\caH_R}
\rmk\rmk 
\Ad
\tilde W_g\lmk\unit_{\caH_L}\otimes \tilde u_R\lmk h,k\rmk\rmk\right\}\cdot
\lmk\unit_{\caH_L}\otimes \tilde u_R\lmk g,hk\rmk\rmk
\notag\\
\end{align}
Because of (iii) of Lemma \ref{lemni},
the $\{\cdot\}$ part of the last equation is equal to
$\Ad
\tilde W_g\lmk\unit_{\caH_L}\otimes \tilde u_R\lmk h,k\rmk\rmk$.
Hence we obtain
\begin{align}
\unit_{\caH_L}\otimes
 \tilde u_R(g,h) \tilde u_R(gh,k)
 =c_R(g,h,k)
 \Ad
\tilde W_g
\lmk\unit_{\caH_L}\otimes \tilde u_R\lmk h,k\rmk\rmk\cdot
\lmk\unit_{\caH_L}\otimes \tilde u_R\lmk g,hk\rmk\rmk.
\end{align}
This proves $c_R=\tilde c_R$, completing the proof.
\end{proof}

\begin{lem}
Let
\begin{align}
&\omega\in\QLS,  \; 0<\theta<\frac\pi 2,\;
(\tilde\beta_{g}^{(1)}), (\tilde\beta_g^{(2)})\in \IG(\omega,\theta).
\end{align}
Then
we have
\begin{align}\label{hq4}
h^{(4)}
\lmk
\omega , \theta,
(\tilde\beta_g^{(1)})
\rmk
=
h^{(4)}
\lmk
\omega, \theta,
(\tilde\beta_g^{(2)})
\rmk.
\end{align}
\end{lem}
\begin{defn}
From this lemma 
we may
 define
 \begin{align}
 h^{(5)}
\lmk
\omega, \theta
\rmk
:=
 h^{(4)}
\lmk
\omega, \theta,
(\tilde\beta_g)
\rmk
  \end{align}
   for any 
 \begin{align}
 \omega\in\QLS,  \; 0<\theta<\frac\pi 2,\;\text{with}\;
\IG(\omega,\theta)\neq \emptyset
 \end{align}
 independent of the choice of $(\tilde\beta_g)$.
\end{defn}
\begin{proof}
By the definition of  $\IG(\omega,\theta)$, there are
\begin{align}
(\eta_{g,i}^{\sigma})_{g\in G,\sigma=L,R} \in
\caT(\theta, (\tilde\beta_g^{(i)})),\quad
\text{for}\quad i=1,2.
\end{align}
We set $\eta_{g,i}:=\eta_{g,i}^L\otimes\eta_{g,i}^R$, for $i=1,2$.
There are
$ \alpha\in \eaut(\omega)$
and $(\alpha_L,\alpha_R,\Theta)\in \caD_{\alpha}^{\theta}$ 
for $\omega\in \QLS$ by the definition.
Setting $\alpha_0:=\alpha_L\otimes\alpha_R$, we have $\alpha=\inn\circ\alpha_0\circ\Theta$.
By Lemma \ref{ichi}, there is some
\begin{align}
\lmk (W_{g,1}), (u_{\sigma}^{(1)}(g,h))\rmk
\in \IP\lmk
\omega, \alpha, \theta,
(\tilde\beta_g^{(1)}), (\eta_{g,1}^\sigma),
(\alpha_{L},\alpha_{R},\Theta)
\rmk.
\end{align}
Set
\begin{align}
K^\sigma_g:=\eta_{g,2}^\sigma\circ \lmk \eta_{g,1}^\sigma\rmk^{-1}\in \Aut\lmk
\caA_{C_\theta,\sigma}\rmk,\quad\text{for}\quad \sigma=L,R,\; g\in G,\quad
K_g:=K_{g}^L\otimes K_g^R\in\Aut\lmk
\caA_{C_\theta}\rmk.
\end{align}

We claim that there are unitaries $V_g^\sigma$, $g\in G,\sigma=L,R$
on $\caH_\sigma$
such that
\begin{align}\label{vke}
\Ad \lmk V_g^\sigma\rmk\circ\pi_\sigma
=\pi_\sigma\circ\alpha_\sigma\circ K_g^\sigma\circ \lmk\alpha_\sigma\rmk^{-1}.
\end{align}
To see this, note that
\begin{align}
\omega=\omega\circ\tilde \beta_g^{(i)}
=\omega_0\circ\alpha\circ\tilde \beta_g^{(i)}
\sim_{q.e.} \omega_0\circ\alpha_0\circ\Theta\circ
\lmk
\eta_{g,i}^{L}\otimes \eta_{g,i}^{R}
\rmk\circ\beta_g^U,\quad i=1,2.
\end{align}
Therefore, we have
\begin{align}
%\omega_0\circ\alpha\circ\tilde \beta_g^{(1)}
%\sim_{q.e.} 
\omega_0\circ\alpha_0\circ\Theta\circ
\lmk
\eta_{g,1}^{L}\otimes \eta_{g,1}^{R}
\rmk
\sim_{q.e.}
\omega\circ\lmk \beta_g^U\rmk^{-1}
\sim_{q.e.}
%\omega_0\circ\alpha\circ\tilde \beta_g^{(2)}
%\sim_{q.e.}
 \omega_0\circ\alpha_0\circ\Theta\circ
\lmk
\eta_{g,2}^{L}\otimes \eta_{g,2}^{R}
\rmk,
\end{align}
and then using the fact that 
$\Theta\in \Aut(\caA_{C_\theta^c})$ and
$K_g\in \Aut(\caA_{C_\theta})$,
\begin{align}
\omega_0\sim_{q.e.}
\omega_0\circ\alpha_0\circ\Theta\circ K_g\circ\Theta^{-1}\circ \alpha_0^{-1}
=\omega_0\circ \alpha_0\circ K_g\circ \lmk \alpha_0\rmk^{-1}
=\bigotimes_{\sigma=L,R} \omega_\sigma \circ\alpha_\sigma K_g^\sigma \lmk \alpha_\sigma\rmk^{-1}.
\end{align}
This implies
that $\omega_\sigma$ and $ \omega_\sigma \circ\alpha_\sigma K_g^\sigma \lmk \alpha_\sigma\rmk^{-1}$
are quasi-equivalent.
Because $\pi_\sigma$ is irreducible, this implies the existence of a unitary $V_g^\sigma$
on $\caH_\sigma$
satisfying (\ref{vke}), proving the claim.

Next we claim that there are unitaries $v^\sigma_{g,h}$ on $\caH_\sigma$, for $g,h\in G$ and $\sigma=L,R$
such that 
\begin{align}
\Ad_{W_{g,1}}\lmk\unit_{\caH_L}\otimes V_h^R\rmk
=\unit_{\caH_L}\otimes v^{R}_{g,h},\quad
\Ad_{W_{g,1}}\lmk V_h^L\otimes \unit_{\caH_R}\rmk
=v^{L}_{g,h}\otimes\unit_{\caH_R}, \label{vew}
\end{align}
and \begin{align}\label{vwwv}
\Ad\lmk V_g^\sigma v^\sigma_{g,h}
u_\sigma^{(1)}(g,h)
\lmk
V_{gh}^\sigma
\rmk^*
\rmk\pi_\sigma
=\pi_\sigma\circ \alpha_\sigma 
 \eta_{g,2}^{\sigma}\beta_g^{{\sigma} U}
\eta_{h,2}^{\sigma}\lmk\beta_g^{{\sigma} U}\rmk^{-1}\lmk \eta_{gh,2}^{\sigma}\rmk^{-1}
\alpha_\sigma ^{-1},
\end{align}
for any $g,h\in G$ and $\sigma=L,R$.
To see this, 
first we calculate
\begin{align}\
&\Ad\lmk W_{g,1}\lmk \unit_{\caH_L}\otimes V_h^R\rmk \lmk W_{g,1}\rmk^*
\rmk
\circ\pi_0
=\Ad\lmk W_{g,1}\lmk \unit_{\caH_L}\otimes V_h^R\rmk 
\rmk\pi_0\circ
\alpha_0\circ\Theta\circ\lmk {\eta_{g,1}}\beta_g^U\rmk^{-1}\circ\Theta^{-1}\circ\alpha_0^{-1}\notag\\
&=\pi_0\circ
\alpha_0\circ\Theta\circ {\eta_{g,1}}\beta_g^U\circ\Theta^{-1}\circ\alpha_0^{-1}
\circ
\lmk
\id_{L}\otimes \alpha_R\circ K_h^R\circ \lmk\alpha_R\rmk^{-1}
\rmk
\circ
\alpha_0\circ\Theta\circ\lmk {\eta_{g,1}}\beta_g^U\rmk^{-1}\circ\Theta^{-1}\circ\alpha_0^{-1}\notag\\
&=
\pi_0\circ
\alpha_0\circ\Theta\circ {\eta_{g,1}}\beta_g^U\circ\Theta^{-1}
\circ
\lmk
\id_{L}\otimes  K_h^R
\rmk
\circ\Theta\circ\lmk {\eta_{g,1}}\beta_g^U\rmk^{-1}\circ\Theta^{-1}\circ\alpha_0^{-1}
\notag\\
&=
\pi_0\circ
\alpha_0\circ\Theta\circ {\eta_{g,1}}\beta_g^U
\circ
\lmk
\id_{L}\otimes  K_h^R
\rmk
\circ\lmk {\eta_{g,1}}\beta_g^U\rmk^{-1}\circ\Theta^{-1}\circ\alpha_0^{-1}\notag\\
&=\pi_0\circ\alpha_0\circ\Theta
\circ
\lmk
\id_{L}\otimes  {\eta_{g,1}^R}\beta_g^{RU}
K_h^R \lmk {\eta_{g,1}^R}\beta_g^{RU}\rmk^{-1}
\rmk
\circ\Theta^{-1}\circ\alpha_0^{-1}\notag\\
&=
\pi_0\circ
\lmk
\id_{L}\otimes  \alpha_R\circ {\eta_{g,1}^R}\beta_g^{RU}
K_h^R \lmk {\eta_{g,1}^R}\beta_g^{RU}\rmk^{-1}\alpha_R^{-1}
\rmk.\label{zousan}
\end{align}
Here, in the fourth and sixth equality,
we used the fact that $K_h^R, \eta_{g,1}^R\beta_g^{RU}
K_h^R \lmk \eta_{g,1}^R\beta_g^{RU}\rmk^{-1}\in\Aut\lmk\caA_{C_\theta}\rmk$
and $\Theta\in \Aut\lmk\caA_{C_\theta^c}\rmk$ commute, in order to remove $\Theta$.
Equation (\ref{zousan}) and the fact that $\pi_L$ is irreducible
imply that there is a unitary  $v^R_{g,h}$ satisfying (\ref{vew}).
The same argument implies the existence of  $v^L_{g,h}$ satisfying (\ref{vew}).

For this $v^R_{g,h}$, we would like to show (\ref{vwwv}).
Rewriting
\begin{align}
\eta_{g,2}^{\sigma}\beta_g^{{\sigma} U}
\eta_{h,2}^{\sigma}\lmk\beta_g^{{\sigma} U}\rmk^{-1}\lmk \eta_{gh,2}^{\sigma}\rmk^{-1}
=K_g^\sigma\cdot \lmk
\eta_{g,1}^{\sigma}\beta_g^{{\sigma} U}
K_h^\sigma
\lmk \eta_{g,1}^{\sigma}\beta_g^{{\sigma} U}\rmk^{-1}
\rmk\cdot
\eta_{g,1}^{\sigma}\beta_g^{{\sigma} U}
\eta_{h,1}^{\sigma}\lmk\beta_g^{{\sigma} U}\rmk^{-1}\lmk \eta_{gh,1}^{\sigma}\rmk^{-1}
\cdot
\lmk
K_{gh}^{\sigma}
\rmk^{-1},
\end{align}
we obtain 
\begin{align}
&\pi_L\otimes \pi_R\circ \alpha_R 
 \eta_{g,2}^{R}\beta_g^{{R} U}
\eta_{h,2}^{R}\lmk\beta_g^{{R} U}\rmk^{-1}\lmk \eta_{gh,2}^{R}\rmk^{-1}
\alpha_R ^{-1}\notag\\
&=\pi_0\circ\lmk
\id_L\otimes \alpha_R\circ K_g^R\cdot \lmk
\eta_{g,1}^{R}\beta_g^{{R} U}
K_h^R
\lmk \eta_{g,1}^{R}\beta_g^{{R} U}\rmk^{-1}
\rmk\cdot
\eta_{g,1}^{R}\beta_g^{{R} U}
\eta_{h,1}^{R}\lmk\beta_g^{{R} U}\rmk^{-1}\lmk \eta_{gh,1}^{R}\rmk^{-1}
\cdot
\lmk
K_{gh}^{R}
\rmk^{-1}
\alpha_R^{-1}
\rmk\notag\\
&=
\pi_L\otimes 
\Ad\lmk {V_g^R v_{g,h}^R u_R^{(1)}(g,h) \lmk V_{gh}^R\rmk^*}\rmk
\pi_R,
\end{align}
substituting
(\ref{vke}), (\ref{zousan}), (\ref{vew}).
This proves (\ref{vwwv}) for $\sigma=R$.
Analogous result for $\sigma=L$ can be proven by the same argument.
Hence we have proven the claim (\ref{zousan}) and (\ref{vwwv}).

Setting 
\begin{align}
V_g:=V_g^L\otimes V_g^R\in\caU(\caH_0),
\end{align}
we have 
\begin{align}\label{vwo}
&\Ad\lmk V_gW_{g,1}\rmk\circ\pi_0
=\pi_0\circ \alpha_0\circ K_g\circ \alpha_0^{-1}\circ \alpha_0\circ\Theta\circ \eta_{g,1}
\circ\beta_g^U\circ \Theta^{-1}\circ \alpha_0^{-1}\notag\\
&=\pi_0\circ \alpha_0\circ\Theta\circ \eta_{g,2}\circ\beta_g^U\circ\Theta^{-1}\circ \alpha_0^{-1}.
\end{align}
In the last equality, we used the definition of $K_g$ and the commutativity of $\Theta$ and $K_g$ again.
From (\ref{vwo}) and (\ref{vwwv}),
setting
\begin{align}
u_\sigma^{(2)}(g,h):=
V_g^\sigma v^\sigma_{g,h}
u_\sigma^{(1)}(g,h)
\lmk
V_{gh}^\sigma
\rmk^*,
\end{align}
we see that
\begin{align}\label{nagasaki3}
\lmk
\lmk
V_gW_{g,1}\rmk,  \lmk
{u_{R}^{(2)}}(g,h )\rmk\rmk
\in  \IP\lmk
\omega, \alpha, \theta,
(\tilde\beta_g^{(2)}), (\eta_{g,2}^\sigma),
(\alpha_{L},\alpha_{R},\Theta)
\rmk,
\end{align}
and 
\begin{align}\label{uw2}
\unit_{\caH_L}\otimes {u_{R}^{(2)}}(g,h )
=\lmk \unit_{\caH_L}\otimes V_g^R\rmk
W_{g,1} \lmk \unit_{\caH_L}\otimes V_h^R\rmk
\lmk W_{g,1}\rmk^*
\lmk
\unit_{\caH_L}\otimes {u_{R}^{(1)}}(g,h )\lmk V_{gh}^R\rmk^*
\rmk.
\end{align}

Now we set
\begin{align}
&c_{R,1}:=c_R\lmk
\omega, \alpha, \theta,
(\tilde\beta_g^{(1)}), (\eta_{g,1}^\sigma),
(\alpha_{L},\alpha_{R},\Theta),\lmk (W_{g,1}), (u_{\sigma}^{(1)}(g,h))\rmk
\rmk,\notag\\
&c_{R,2}:=c_{R}\lmk
\omega, \alpha, \theta,
(\tilde\beta_g^{(2)}), (\eta_{g,2}^\sigma),
(\alpha_{L},\alpha_{R},\Theta),
\lmk
\lmk
V_gW_{g,1}\rmk,  \lmk
{u_{R}^{(2)}}(g,h )\rmk\rmk
\rmk.
\end{align}
To prove the Lemma, it suffices to show 
$c_{R,1}=c_{R,2}$.
By (\ref{uw2}), we have
\begin{align}
&\unit_{\caH_L}\otimes u_R^{(2)}(g,h) u_R^{(2)}(gh,k)\notag\\
&=\lmk \unit_{\caH_L}\otimes V_g^R\rmk
W_{g,1} \lmk \unit_{\caH_L}\otimes V_h^R\rmk
\lmk W_{g,1}\rmk^*
\lmk
\unit_{\caH_L}\otimes {u_{R}^{(1)}}(g,h )\lmk V_{gh}^R\rmk^*
\rmk\notag\\\cdot
&\lmk \unit_{\caH_L}\otimes V_{gh}^R\rmk
W_{gh,1} \lmk \unit_{\caH_L}\otimes V_k^R\rmk
\lmk W_{gh,1}\rmk^*
\lmk
\unit_{\caH_L}\otimes {u_{R}^{(1)}}(gh,k )\lmk V_{ghk}^R\rmk^*
\rmk\notag\\
&=
\lmk \unit_{\caH_L}\otimes V_g^R\rmk
W_{g,1} \lmk \unit_{\caH_L}\otimes V_h^R\rmk
\lmk W_{g,1}\rmk^*
\lmk
\unit_{\caH_L}\otimes {u_{R}^{(1)}}(g,h )
\rmk\cdot
W_{gh,1} \lmk \unit_{\caH_L}\otimes V_k^R\rmk
\lmk W_{gh,1}\rmk^*
\lmk
\unit_{\caH_L}\otimes {u_{R}^{(1)}}(gh,k )\lmk V_{ghk}^R\rmk^*
\rmk\notag\\
&=
\lmk \unit_{\caH_L}\otimes V_g^R\rmk
W_{g,1} \lmk \unit_{\caH_L}\otimes V_h^R\rmk
\lmk W_{g,1}\rmk^*
\left\{\Ad\lmk
\lmk
\unit_{\caH_L}\otimes {u_{R}^{(1)}}(g,h )
\rmk\cdot
W_{gh,1} \rmk
\lmk \unit_{\caH_L}\otimes V_k^R\rmk
\right\}\cdot\\
&\lmk
\unit_{\caH_L}\otimes  \lcm {u_{R}^{(1)}}(g,h ){u_{R}^{(1)}}(gh,k )\rcm\lmk V_{ghk}^R\rmk^*
\rmk\notag\\
&
=c_{R,1}(g,h,k)
\lmk \unit_{\caH_L}\otimes V_g^R\rmk
W_{g,1} \lmk \unit_{\caH_L}\otimes V_h^R\rmk
\lmk W_{g,1}\rmk^*
\left\{\Ad\lmk
W_{g,1}W_{h,1}  \rmk
\lmk \unit_{\caH_L}\otimes V_k^R\rmk
\right\}\notag\\
&\cdot \lmk W_{g,1}\lmk \unit_{\caH_L}\otimes u_{R}^{(1)}(h,k)\rmk W_{g,1}^*\rmk
\lmk \unit_{\caH_L}\otimes u_{R}^{(1)}(g,hk)\lmk V_{ghk}^R\rmk^*\rmk.
\label{roe}
\end{align}
We used (\ref{uwc}) for $[\cdot]$ part and 
Lemma \ref{lemni} (ii) and (\ref{vew}) for $\{\cdot\}$ part 
for the fourth equality.
Again using (\ref{uw2}),
we have
\begin{align}
&\unit_{\caH_L}\otimes u_R^{(2)}(g,h) u_R^{(2)}(gh,k)=
(\ref{roe})\\
&
=c_{R,1}(g,h,k)
\lmk \unit_{\caH_L}\otimes V_g^R\rmk
W_{g,1} \lmk \unit_{\caH_L}\otimes V_h^R\rmk
\left\{\Ad\lmk
W_{h,1} \rmk
\lmk \unit_{\caH_L}\otimes V_k^R\rmk
\right\}\notag\\
&\cdot\lmk 
W_{h,1}
\lmk \unit_{\caH_L}\otimes \lmk V_k^R\rmk^*\rmk
\lmk W_{h,1}\rmk^* \lmk \unit_{\caH_L}\otimes V_h^R\rmk^*
\lmk \unit_{\caH_L}\otimes
u_{R}^{(2)}(h,k)
\rmk
\lmk
\unit_{\caH_L}\otimes \lmk V_{hk}^R\rmk
\rmk
\lmk W_{g,1}\rmk^*\rmk\notag\\
&W_{g,1}
\lmk \unit_{\caH_L}\otimes \lmk V_{hk}^R\rmk^*\rmk
\lmk W_{g,1}\rmk^* \lmk \unit_{\caH_L}\otimes V_g^R\rmk^*
\lmk \unit_{\caH_L}\otimes
u_{R}^{(2)}(g,hk)
\rmk
\lmk
\unit_{\caH_L}\otimes \lmk V_{ghk}^R\rmk
\rmk
\lmk \unit_{\caH_L}\otimes \lmk V_{ghk}^R\rmk^*\rmk\notag\\
&=
c_{R,1}(g,h,k)
\lmk \unit_{\caH_L}\otimes V_g^R\rmk
W_{g,1} \cdot\lmk 
\lmk \unit_{\caH_L}\otimes
u_{R}^{(2)}(h,k)
\rmk
\rmk\cdot
\lmk W_{g,1}\rmk^* \lmk \unit_{\caH_L}\otimes V_g^R\rmk^*
\lmk \unit_{\caH_L}\otimes
u_{R}^{(2)}(g,hk)
\rmk\notag\\
&=
c_{R,1}(g,h,k)\cdot
\Ad\lmk
\lmk \unit_{\caH_L}\otimes V_g^R\rmk
W_{g,1} 
\rmk
\lmk 
\lmk \unit_{\caH_L}\otimes
u_{R}^{(2)}(h,k)
\rmk
\rmk\cdot
\lmk \unit_{\caH_L}\otimes
u_{R}^{(2)}(g,hk)
\rmk\notag\\
&=
c_{R,1}(g,h,k)\cdot
\Ad\lmk
\lmk  {V_g^L}^*\otimes \unit_{\caH_R}\rmk
V_g
W_{g,1} 
\rmk
\lmk 
\lmk \unit_{\caH_L}\otimes
u_{R}^{(2)}(h,k)
\rmk
\rmk\cdot
\lmk \unit_{\caH_L}\otimes
u_{R}^{(2)}(g,hk)
\rmk\notag\\
&=
c_{R,1}(g,h,k)\cdot
\Ad\lmk
V_g
W_{g,1} 
\rmk
\lmk 
\lmk \unit_{\caH_L}\otimes
u_{R}^{(2)}(h,k)
\rmk
\rmk\cdot
\lmk \unit_{\caH_L}\otimes
u_{R}^{(2)}(g,hk)
\rmk
\end{align}
In the last line we used (\ref{nagasaki3}) and  Lemma \ref{lemni} (iii)
to remove ${V_g^L}^*$.
From this, we see that $c_{R,1}=c_{R,2}$, completing the proof.
\end{proof}
\begin{lem}
Let
\begin{align}
&\omega\in\QLS,  \; 0<\theta_1<\theta_2<\frac\pi 2,\;
\text{with}\; \IG(\omega,\theta_1), \IG(\omega,\theta_2)\neq\emptyset.
\end{align}
Then
we have
\begin{align}\label{hq5}
h^{(5)}
\lmk
\omega , \theta_1
\rmk
=
h^{(5)}
\lmk
\omega, \theta_2
\rmk.
\end{align}
\end{lem}\begin{defn}\label{theindexdef}
From this lemma, for any  $\omega\in\QLS$ with
$\IG(\omega)\neq\emptyset$,
we may
 define
 \begin{align}
 h
\lmk
\omega
\rmk
:=
 h^{(5)}
\lmk
\omega, \theta
\rmk
  \end{align}
  independent of the choice of $\theta$.
  This is the index we associate to $\omega\in\QLS$ with
$\IG(\omega)\neq\emptyset$.
\end{defn}
\begin{proof}
By the assumption, there are
some $(\tilde\beta_g)\in\IG(\omega,\theta_1)$
and $(\eta_g^\sigma)\in \caT\lmk (\theta_1, \tilde\beta_g)\rmk$.
Because $\omega\in\QLS$, there are
$ \alpha\in \eaut(\omega)$
and $(\alpha_L,\alpha_R,\Theta)\in \caD_{\alpha}^{\theta_2}$ 
by the definition.
Setting $\alpha_0:=\alpha_L\otimes\alpha_R$, we have $\alpha=\inn\circ\alpha_0\circ\Theta$.
Because $0<\theta_1<\theta_2<\frac\pi 2$,
we also have $(\eta_g^\sigma)\in \caT\lmk (\theta_2, \tilde\beta_g)\rmk$, and
$(\tilde\beta_g)\in\IG(\omega,\theta_2)$.
For the same reason, we also have 
$(\alpha_L,\alpha_R,\Theta)\in \caD_{\alpha}^{\theta_1}$.

By Lemma \ref{ichi}, there is some
\begin{align}
\lmk (W_g), (u_\sigma(g,h))\rmk
\in \IP\lmk
\omega, \alpha, \theta_1,
(\tilde\beta_g), (\eta_{g}^\sigma),
(\alpha_L,\alpha_R,\Theta)
\rmk.
\end{align}
However, 
we also have
\begin{align}
\lmk (W_g), (u_\sigma(g,h))\rmk
\in \IP\lmk
\omega, \alpha, \theta_2,
(\tilde\beta_g), (\eta_{g}^\sigma),
(\alpha_L,\alpha_R,\Theta)
\rmk.
\end{align}
Therefore, we obtain
$h^{(5)}
\lmk
\omega , \theta_1
\rmk
=
h^{(5)}
\lmk
\omega, \theta_2
\rmk
$.
\end{proof}
This completes the proof of Theorem \ref{welldefined}.

\section{The existence of $\tilde\beta$ for SPT phases}\label{tildebetasec}
In this section, we 
give a sufficient condition for $\IG(\omega)$ to be non-empty.
We consider the same setting as in subsection \ref{settingsubsec}.
\begin{thm}\label{defindexspt}
For any $0<\theta<\frac\pi 2$ and $\alpha\in \sqaut(\caA)$ satisfying
$
\omega_0\circ\alpha\circ\beta_g=\omega_0\circ\alpha$
for all $g\in G$,
$\IG(\omega_0\circ\alpha,\theta)$ is not empty.
\end{thm}

In order to prove this theorem, we first show several general lemmas.
\begin{lem}\label{splitlem5}
Let $\mkA,\mkB$ be UHF-algebras.
Let $\omega$ be a pure state on $\mkA\otimes\mkB$ and
$\varphi_{\mkA}$,  $\varphi_{\mkB}$ states on $\mkA$,
$\mkB$ respectively.
Assume that $\omega$ is quasi-equivalent to $\varphi_{\mkA}\otimes\varphi_{\mkB}$.
Then for any pure states $\psi_{\mkA}$, $\psi_{\mkB}$
on $\mkA$, $\mkB$, there are automorphisms 
$\gamma_{\mkA}\in\Aut\lmk \mkA\rmk$,
$\gamma_{\mkB}\in\Aut\lmk \mkB\rmk$ and
a unitary $u\in\caU\lmk\mkA\otimes\mkB\rmk$
such that 
\begin{align}
\omega=\lmk
\lmk
\psi_{\mkA}\circ\gamma_{\mkA}\rmk\otimes\lmk
\psi_{\mkB}\circ\gamma_{\mkB}
\rmk
\rmk\circ\Ad(u).
\end{align}
If $\psi_{\mkA}$ and $\varphi_{\mkA}$ are quasi-equivalent,
then we may set $\gamma_{\mkA}=\id_{\mkA}$.
\end{lem}
\begin{proof}
Let $(\caH_{\omega},\pi_{\omega},\Omega_{\omega})$
$(\caH_{\varphi_{\mkA}},\pi_{\varphi_{\mkA}},\Omega_{\varphi_{\mkA}})$
$(\caH_{\varphi_{\mkB}},\pi_{\varphi_{\mkB}},\Omega_{\varphi_{\mkB}})$
be
GNS triples of $\omega$, $\varphi_{\mkA}$, $\varphi_{\mkB}$ respectively.
Then $(\caH_{\varphi_{\mkA}}\otimes
\caH_{\varphi_{\mkB}},\pi_{\varphi_{\mkA}}\otimes 
\pi_{\varphi_{\mkB}},\Omega_{\varphi_{\mkA}}\otimes\Omega_{\varphi_{\mkB}})$
is a GNS triple of $\varphi_{\mkA}\otimes\varphi_{\mkB}$.
As $\omega$ is quasi-equivalent to $\varphi_{\mkA}\otimes\varphi_{\mkB}$,
there is a $*$-isomorphism 
$\tau:\pi_{\omega}\lmk\mkA\otimes\mkB\rmk''\to
\pi_{\varphi_{\mkA}}(\mkA)''\otimes\pi_{\varphi_{\mkB}}(\mkB)''$
such that $\tau\circ\pi_{\omega}=\pi_{\varphi_{\mkA}}\otimes 
\pi_{\varphi_{\mkB}}$.
Because $\omega$ is pure, we have $\pi_{\omega}\lmk\mkA\otimes\mkB\rmk''=\caB(\caH_{\omega})$
and from the isomorphism $\tau$, $\pi_{\varphi_{\mkA}}(\mkA)''\otimes\pi_{\varphi_{\mkB}}(\mkB)''$
is also a type $I$ factor.
Then from Theorem 2.30 V \cite{takesaki},
both of $\pi_{\varphi_{\mkA}}(\mkA)''$ and $\pi_{\varphi_{\mkB}}(\mkB)''$
are type $I$ factors.
The restriction of $\tau$ to $\pi_{\omega}\lmk\mkA\otimes\bbC\unit_{\mkB}\rmk''$
implies a $*$-isomorphism from $\pi_{\omega}\lmk\mkA\otimes\bbC\unit_{\mkB}\rmk''$
onto the type I factor $\pi_{\varphi_{\mkA}}(\mkA)''$.
Hence we see that $\pi_{\omega}\lmk\mkA\otimes\bbC\unit_{\mkB}\rmk''$ is a type $I$-factor.
%Because $\omega$ is pure, $\pi_{\omega}\lmk\mkA\otimes\mkB\rmk''=\caB(\caH_{\omega})$
%and from the isomorphism $\tau$,
%$\pi_{\varphi}(\mkA)''\otimes\pi_{\psi}(\mkB)''$ is a type $I$-factor.
%By Theorem 2.30 V \cite{takesaki},  
Therefore, from Theorem 1.31 V of \cite{takesaki},
there are Hilbert spaces
$\caK_{\mkA},\caK_{\mkB}$ and a unitary $W: \caH_\omega\to \caK_{\mkA}\otimes \caK_{\mkB}$
such that $\Ad\lmk W\rmk\lmk \pi_{\omega}\lmk\mkA\otimes\bbC\unit_{\mkB}\rmk''\rmk=
\caB\lmk \caK_{\mkA}\rmk\otimes\bbC\unit_{\caK_{\mkB}} $.
Because $\omega$ is pure, we also have 
$\Ad\lmk W\rmk\lmk \pi_{\omega}\lmk\bbC\unit_{\mkA}
\otimes \mkB\rmk''\rmk=\bbC\unit_{\caK_{\mkA}}\otimes \caB(\caK_{\mkB})$.
From this, we see that there are irreducible representations
$\rho_{\mkA}$, $\rho_{\mkB}$ of $\mkA$ and $\mkB$ on
$\caK_{\mkA}$, $\caK_{\mkB}$
such that $\Ad(W)\circ \pi_{\omega}=\rho_{\mkA}\otimes \rho_{\mkB}$.
Fix some unit vectors $\xi_{\mkA}\in \caK_{\mkA}$, $\xi_{\mkB}\in \caK_{\mkB}$.
Then because of the irreducibility of $\rho_{\mkB}$ and $\rho_{\mkB}$,
 $\omega_{\mkA}:=\braket{\xi_{\mkA}}{\rho_{\mkA}\lmk\cdot\rmk\xi_{\mkA}}$ and
$\omega_{\mkB}:=\braket{\xi_{\mkB}}{\rho_{\mkB}\lmk\cdot\rmk\xi_{\mkB}}$
are pure states on $\mkA$, $\mkB$.
By Theorem 1.1 of \cite{kos} (originally proved by Powers \cite{powers} for UHF-algebras)
for any pure states $\psi_{\mkA}$, $\psi_{\mkB}$
on $\mkA$, $\mkB$, there exist automorphisms 
$\gamma_{\mkA}\in\Aut(\mkA)$
$\gamma_{\mkB}\in\Aut(\mkB)$
such that 
$\omega_{\mkA}=\psi_{\mkA}\circ \gamma_{\mkA}$
$\omega_{\mkB}=\psi_{\mkB}\circ \gamma_{\mkB}$.
Now for unit vectors $W^{*}\lmk\xi_{\mkA}\otimes \xi_{\mkB}\rmk,\Omega_{\omega}\in\caH_{\omega}$,
by Kadison's transitivity theorem and the irreducibility of $\pi_\omega$, there exists a unitary $u\in\caU\lmk \mkA\otimes\mkB\rmk$
such that $\pi_{\omega}(u)\Omega_{\omega}=W^{*}\lmk\xi_{\mkA}\otimes \xi_{\mkB}\rmk$.
Substituting this, we obtain
\begin{align}\label{opro}
&\omega=\braket{\Omega_{\omega}}{\pi_{\omega}\lmk\cdot\rmk \Omega_{\omega}}
=\braket{\pi_{\omega}(u^{*})W^{*}\lmk\xi_{\mkA}\otimes \xi_{\mkB}\rmk}{\pi_{\omega}\lmk\cdot\rmk\pi_{\omega}(u^{*})W^{*}\lmk\xi_{\mkA}\otimes \xi_{\mkB}\rmk}
=\braket{W^{*}\lmk\xi_{\mkA}\otimes \xi_{\mkB}\rmk}{\pi_{\omega}\circ\Ad(u)\lmk\cdot\rmk
W^{*}\lmk\xi_{\mkA}\otimes \xi_{\mkB}\rmk}\notag\\
&=\braket{\lmk\xi_{\mkA}\otimes \xi_{\mkB}\rmk}{\lmk\rho_{\mkA}\otimes \rho_{\mkB}\rmk\circ\Ad(u)\lmk\cdot\rmk
\lmk\xi_{\mkA}\otimes \xi_{\mkB}\rmk}
=\lmk \omega_{\mkA}\otimes\omega_{\mkB}\rmk\circ \Ad(u)
=\lmk \psi_{\mkA}\circ \gamma_{\mkA}\otimes \psi_{\mkB}\circ \gamma_{\mkB}\rmk
\circ \Ad(u).
\end{align}
Now assume that $\psi_{\mkA}$ and $\varphi_{\mkA}$ are quasi-equivalent, i.e.,
the GNS representations of $\psi_{\mkA}$, $\varphi_{\mkA}$, denoted by
$\pi_{\psi_{\mkA}}$ and $\pi_{\varphi_{\mkA}}$ are quasi-equivalent.
From the above argument, $\pi_{\omega}\vert_{\mkA}$ and $\pi_{\varphi_{\mkA}}$
are quasi-equivalent. At the same time, $\pi_{\omega}\vert_{\mkA}$ and
$\rho_{\mkA}$ are quasi-equivalent. 
Therefore, $\pi_{\psi_{\mkA}}$ and $\rho_{\mkA}$ are quasi-equivalent.
Because both of them are irreducible, we
see that a pure state ${\psi_{\mkA}}$  can be represented by a unit vector $\zeta\in \caK_{\mkA}$
as ${\psi_{\mkA}}=\braket{\zeta}{\rho_{\mkA}\lmk\cdot\rmk \zeta}$.
Because $\rho_{\mkA}$ is irreducible, 
by Kadison's transitivity theorem, there exists a unitary $w\in\caU\lmk \mkA\rmk$ such that
$\rho_{\mkA}(w^{*})\zeta=\xi_{\mkA}$.
Hence we obtain ${\psi_{\mkA}}\circ \Ad(w)=\omega_{\mkA}$.
Substituting this instead of $\omega_{\mkA}=\psi_{\mkA}\circ \gamma_{\mkA}$
in (\ref{opro}),
we obtain
\begin{align}
\omega=
\lmk \psi_{\mkA}\otimes \psi_{\mkB}\circ \gamma_{\mkB}\rmk
\circ \Ad\lmk \lmk w\otimes \id_{\mkB}\rmk u\rmk,
\end{align}
proving the last claim.

\end{proof}

\begin{lem}\label{lem8s}
Let $\mkB,\mkA_{1,L},\mkA_{2,L},\mkA_{1,R},\mkA_{2,R}$ be UHF-algebras.
Set $\mkA_{1}:=\mkA_{1,L}\otimes\mkA_{1,R}$,
$\mkA_{2}:=\mkA_{2,L}\otimes\mkA_{2,R}$,
$\mkA_{L}:=\mkA_{1,L}\otimes\mkA_{2,L}$, and
$\mkA_{R}:=\mkA_{1,R}\otimes\mkA_{2,R}$. Let 
$\omega$, $\varphi_{L}^{(1,2)}$, $\varphi_{R}^{(1,2)}$, $\psi$
  be pure states on $\mkB\otimes \mkA_{1}$, $\mkA_{L}$,
$\mkA_{R}$, $\mkB$, respectively.
Suppose that
$\omega$ is quasi-equivalent to $\left.\lmk
\psi\otimes \varphi_{L}^{(1,2)}\otimes \varphi_{R}^{(1,2)}
\rmk\right\vert_{\mkB\otimes\mkA_{1}}$.
Then for any pure states $\varphi_{L}^{(1)}$, $\varphi_{R}^{(1)}$
on $\mkA_{1,L}$, $\mkA_{1,R}$ respectively,
there are automorphisms $\gamma_{L}^{(1)}\in\Aut\lmk \mkA_{1,L}\rmk$,
$\gamma_{R}^{(1)}\in\Aut\lmk \mkA_{1,R}\rmk$,
and a unitary $u\in\caU\lmk\mkB\otimes \mkA_{1}\rmk$
such that
\begin{align}\label{kurashiki}
\omega=\lmk
\psi\otimes \lmk \varphi_{L}^{(1)}\circ\gamma_{L}^{(1)}\rmk
\otimes \lmk
\varphi_{R}^{(1)}\circ\gamma_{R}^{(1)}\rmk
\rmk\circ\Ad u.
\end{align}
\end{lem}
\begin{proof}
Because the pure state $\omega$ is quasi-equivalent to $\left.\lmk
\psi\otimes \varphi_{L}^{(1,2)}\otimes \varphi_{R}^{(1,2)}
\rmk\right\vert_{\mkB\otimes\mkA_{1}}
=\psi\otimes \left.\lmk
\varphi_{L}^{(1,2)}\otimes \varphi_{R}^{(1,2)}
\rmk\right\vert_{\mkA_{1}}$, applying Lemma \ref{splitlem5},
 for any pure states $\varphi_{L}^{(1)}$, $\varphi_{R}^{(1)}$
on $\mkA_{1,L}$, $\mkA_{1,R}$,
there exist an automorphism $S\in\Aut\mkA_1$ and a unitary $v\in\caU\lmk \mkB\otimes\mkA_1\rmk$
such that
\begin{align}\label{matsushima}
\omega=\lmk
\psi\otimes\lmk
 \lmk \varphi_{L}^{(1)}
\otimes 
\varphi_{R}^{(1)}\rmk\circ S\rmk
\rmk\circ \Ad v.
\end{align}
From (\ref{matsushima}) and $\omega\sim_{q.e.}\left.\lmk
\psi\otimes \varphi_{L}^{(1,2)}\otimes \varphi_{R}^{(1,2)}
\rmk\right\vert_{\mkB\otimes\mkA_{1}}$,
we get
$
\lmk
\psi\otimes\lmk
 \lmk \varphi_{L}^{(1)}
\otimes 
\varphi_{R}^{(1)}\rmk\circ S\rmk
\rmk\sim_{q.e.}
\left.\lmk
\psi\otimes \varphi_{L}^{(1,2)}\otimes \varphi_{R}^{(1,2)}
\rmk\right\vert_{\mkB\otimes\mkA_{1}}$,
which implies
\begin{align}\label{hiraizumi}
\lmk \varphi_{L}^{(1)}
\otimes 
\varphi_{R}^{(1)}\rmk\circ S\sim_{q.e.} \left.\lmk
\varphi_{L}^{(1,2)}\otimes \varphi_{R}^{(1,2)}
\rmk\right\vert_{\mkA_{1}}.
\end{align}
Applying Lemma \ref{splitlem5} to (\ref{hiraizumi}),
there are automorphisms $\gamma_{L}^{(1)}\in\Aut\lmk \mkA_{1,L}\rmk$,
$\gamma_{R}^{(1)}\in\Aut\lmk \mkA_{1,R}\rmk$,
and a unitary $w\in\caU\lmk \mkA_{1}\rmk$
such that
\begin{align}
\lmk \varphi_{L}^{(1)}
\otimes 
\varphi_{R}^{(1)}\rmk\circ S=\lmk\lmk \varphi_{L}^{(1)}\circ\gamma_{L}^{(1)}\rmk
\otimes \lmk
\varphi_{R}^{(1)}\circ\gamma_{R}^{(1)}\rmk
\rmk\circ\Ad w.
\end{align}
Substituting this to (\ref{matsushima}), we obtain
(\ref{kurashiki}).
\end{proof}
\begin{lem}\label{lem11s}
Let $\mkA_L$, $\mkA_R$, $\mkB_{LU}$, $\mkB_{LD}$, $\mkB_{RU}$, $\mkB_{RD}$,
$\mkC_U$, $\mkC_D$ be UHF-algebras, and set
\begin{align}
&
\mkB_U:=\mkB_{LU}\otimes\mkB_{RU},\quad
\mkB_D:=\mkB_{LD}\otimes\mkB_{RD},\quad
\mkB_L:=\mkB_{LD}\otimes\mkB_{LU},\quad
\mkB_R:=\mkB_{RD}\otimes\mkB_{RU},\
\notag\\
&\mkA:=\mkA_L\otimes\mkA_R,\quad\mkB:=\mkB_D\otimes\mkB_U=\mkB_L\otimes\mkB_R,\quad
\mkC:=\mkC_D\otimes\mkC_U,\quad
\mkD:=\mkA\otimes\mkB\otimes\mkC.
\end{align}
Let $\omega_X$ be a pure state on
each $X=\mkA_L, \mkA_R, \mkB_{LU}, \mkB_{LD}, \mkB_{RU},
\mkB_{RD}, \mkC_U, \mkC_D$, and
set
\begin{align}
&\omega_{\mkB\mkC}^U:=\omega_{\mkB_{LU}}\otimes\omega_{\mkB_{RU}}\otimes\omega_{\mkC_U},\quad
\text{on}\quad \mkB_U\otimes\mkC_U\notag\\
&\omega_{\mkB\mkC}^D:=\omega_{\mkB_{LD}}\otimes\omega_{\mkB_{RD}}\otimes\omega_{\mkC_D},\quad
\text{on}\quad \mkB_D\otimes\mkC_D,\notag\\
&\omega_{\mkA}:=\omega_{\mkA_L}\otimes\omega_{\mkA_R}
\quad
\text{on}\quad\mkA\notag\\
&\omega_{\mkA\mkB}^L
:=\omega_{\mkA_L}\otimes\omega_{\mkB_{LU}}\otimes\omega_{\mkB_{LD}}
\quad
\text{on}\quad 
\mkA_L\otimes\mkB_L
\notag\\
&\omega_{\mkA\mkB}^R
:=\omega_{\mkA_R}\otimes\omega_{\mkB_{RU}}\otimes\omega_{\mkB_{RD}}
\quad
\text{on}\quad 
\mkA_R\otimes\mkB_R
\notag\\
&\omega_0:=\bigotimes_{\substack{X=\mkA_L, \mkA_R, \mkB_{LU}, \mkB_{LD},\\ \mkB_{RU},
\mkB_{RD}, \mkC_U, \mkC_D}}\omega_X,\quad \text{on} \quad\mkD.
\end{align}
Let $\alpha,\hat\alpha$
be automorphisms on $\mkD$ which allow the following decompositions
\begin{align}
&\hat\alpha=\lmk\rho_{\mkB\mkC}^U\otimes\id_{\mkA}\otimes \rho_{\mkB\mkC}^D\rmk
\circ \lmk
\id_{\mkC_U}\otimes\hat\gamma_{\mkA\mkB}^L\otimes \hat\gamma_{\mkA\mkB}^R
\otimes \id_{\mkC_D}
\rmk\circ\inn\label{hatal}\\
&\alpha=\lmk\rho_{\mkB\mkC}^U\otimes\id_{\mkA}\otimes \id_{\mkB_D\otimes\mkC_D}\rmk
\circ \lmk
\id_{\mkC_U}\otimes\gamma_{\mkA\mkB}^L\otimes \gamma_{\mkA\mkB}^R
\otimes \id_{\mkC_D}
\rmk\circ\inn\label{ald}.
\end{align}
Here,  $\rho_{\mkB\mkC}^U$/ $\rho_{\mkB\mkC}^D$
are automorphisms on $\mkB_U\otimes\mkC_U$/  $\mkB_D\otimes\mkC_D$ respectively.
For each $\sigma=L,R$,
$\gamma_{\mkA\mkB}^{\sigma},\hat\gamma_{\mkA\mkB}^{\sigma}$
are automorphisms on $\mkA_{\sigma}\otimes\mkB_{{\sigma}D}\otimes\mkB_{{\sigma}U}$.
Suppose that $\omega_0\circ\hat\alpha=\omega_0$.
Then there are automorphisms $\eta_L,\eta_R$ on $\mkA_L\otimes\mkB_{LD}\otimes\mkB_{LU}$,
$\mkA_R\otimes\mkB_{RD}\otimes\mkB_{RU}$
such that
$\omega_0\circ \alpha$ is quasi-equivalent to 
$\omega_0\circ\lmk\id_{\mkC_U}\otimes\eta_L\otimes\eta_R\otimes\id_{\mkC_D}\rmk$.
\end{lem}
\begin{proof}
First we claim that there are automorphisms $\theta_\mkB^{LU}\in\Aut{\mkB_{LU}}$, 
$\theta_\mkB^{RU}\in\Aut\mkB_{RU}$
and a unitary $u\in\caU\lmk \mkB^U\otimes\mkC^U\rmk$
such that
\begin{align}\label{lem10s}
\omega_{\mkB\mkC}^U\circ\rho_{\mkB\mkC}^U
=\omega_{\mkB\mkC}^U\circ
\lmk \theta_{\mkB}^{LU}\otimes \theta_{\mkB}^{RU}\otimes\id_{\mkC^U}\rmk
\circ\Ad\lmk u\rmk.
\end{align}
To prove this, we first note that from $\omega_0\circ\hat \alpha=\omega_0$
and the decomposition (\ref{hatal}), we have
\begin{align}
\omega_{\mkB\mkC}^U\circ\rho_{\mkB\mkC}^U\otimes \omega_{\mkA}
\otimes \omega_{\mkB\mkC}^D\circ\rho_{\mkB\mkC}^D
\sim_{q.e.} \omega_{\mkC_U}
\otimes\omega_{\mkA\mkB}^L\circ\lmk\widehat{ \gamma_{\mkA\mkB}^L}\rmk^{-1}
\otimes \omega_{\mkA\mkB}^R\circ\lmk \widehat{\gamma_{\mkA\mkB}^R}\rmk^{-1}
\otimes\omega_{\mkC_D}.
\end{align}
From this, because both of the states above are pure, (hence the restrictions of
their GNS representations onto $\mkC_U\otimes\mkB_U$ are factors)
we have 
\begin{align}\label{daisetsu}
\omega_{\mkB\mkC}^U\circ\rho_{\mkB\mkC}^U=
\left.\lmk 
\omega_{\mkB\mkC}^U\circ\rho_{\mkB\mkC}^U\otimes \omega_{\mkA}
\otimes \omega_{\mkB\mkC}^D\circ\rho_{\mkB\mkC}^D
\rmk
\right\vert_{\mkC_U\otimes
\mkB_U}
\sim_{q.e.}
\omega_{\mkC_U}
\otimes\left.\lmk \omega_{\mkA\mkB}^L\circ\lmk \widehat{\gamma_{\mkA\mkB}^L}\rmk^{-1}
\otimes \omega_{\mkA\mkB}^R\circ\lmk \widehat{\gamma_{\mkA\mkB}^R}\rmk^{-1}\rmk\right\vert_{\mkB_U}.
\end{align}
We apply Lemma \ref{lem8s} for
$\mkB$, $\mkA_{1L}$, $\mkA_{1R}$, $\mkA_{2L}$, $\mkA_{2R}$,
$\omega$, $\varphi_L^{(1,2)}$, $\varphi_R^{(1,2)}$, $\psi$
replaced by $\mkC_U$, $\mkB_{LU}$, $\mkB_{RU}$, $\mkA_{L}\otimes \mkB_{LD}$,
$\mkA_R\otimes \mkB_{RD}$, 
$\omega_{\mkB\mkC}^U\circ\rho_{\mkB\mkC}^U$, $\omega_{\mkA\mkB}^L\circ\lmk\widehat{\gamma_{\mkA\mkB}^L}\rmk^{-1}$, $\omega_{\mkA\mkB}^R\circ\lmk \widehat{\gamma_{\mkA\mkB}^R}\rmk^{-1}$,
$\omega_{\mkC_U}$
 respectively.
From (\ref{daisetsu}), they satisfy the conditions in Lemma \ref{lem8s}.
Applying Lemma \ref{lem8s}, 
(for pure states $\varphi_L^{(1)}=\omega_{\mkB_{LU}}$ and
$\varphi_R^{(1)}=\omega_{\mkB_{RU}}$)
we obtain automorphisms $\theta_{\mkB}^{LU}\in\Aut\lmk\mkB_{LU}\rmk$,
 $\theta_{\mkB}^{RU}\in\Aut\lmk\mkB_{RU}\rmk$,
and a unitary $u\in \caU\lmk\mkB_U\otimes\mkC_U\rmk$
satisfying (\ref{lem10s}).

We set
\begin{align}
\begin{split}
&\eta_L:=\lmk \theta^{LU}_{\mkB}\otimes \id_{\mkA_L}\otimes \id_{\mkB_{LD}}\rmk
\circ\gamma_{\mkA\mkB}^L
\in \Aut\lmk\mkB_{LU}\otimes  \mkA_L\otimes\mkB_{LD}\rmk\\
&\eta_R:=\lmk \theta^{RU}_{\mkB}\otimes \id_{\mkA_R}\otimes \id_{\mkB_{RD}}\rmk
\circ\gamma_{\mkA\mkB}^R
\in \Aut\lmk \mkB_{RU}\otimes \mkA_R\otimes\mkB_{RD}\rmk.
\end{split}
\end{align}
Then we have
\begin{align}
\begin{split}
&\omega_0\circ\alpha=
\lmk
\omega_{\mkA_L}\otimes\omega_{\mkA_R}
\otimes \omega_{\mkB\mkC}^U\otimes \omega_{\mkB\mkC}^D
\rmk\circ\alpha
\sim_{q.e.}
\lmk
\omega_{\mkA_L}\otimes\omega_{\mkA_R}
\otimes
 \omega_{\mkB\mkC}^U\circ\rho_{\mkB\mkC}^U
\otimes
\omega_{\mkB\mkC}^D
\rmk
\circ
\lmk
\id_{\mkC_U}\otimes \gamma_{\mkA\mkB}^L\otimes \gamma_{\mkA\mkB}^R
\otimes \id_{\mkC_D}
\rmk\\
&\sim_{q.e.}
\lmk
\omega_{\mkA_L}\otimes\omega_{\mkA_R}
\otimes
 \omega_{\mkB\mkC}^U
\otimes
\omega_{\mkB\mkC}^D
\rmk
\circ
\lmk
\id_{\mkC_U}\otimes
\lmk
\lmk
\theta_{\mkB}^{LU}\otimes \id_{\mkA_L}\otimes\id_{\mkB_{LD}}
\rmk\circ\gamma_{\mkA\mkB}^L\rmk
\otimes 
\lmk
\lmk
\theta_{\mkB}^{RU}\otimes \id_{\mkA_R}\otimes\id_{\mkB_{RD}}
\rmk\circ\gamma_{\mkA\mkB}^R\rmk
\otimes\id_{\mkC_D}
\rmk\\
&=\omega_0\circ\lmk\id_{\mkC_U}\otimes\eta_L\otimes\eta_R\otimes\id_{\mkC_D}\rmk.
\end{split}
\end{align}
This completes the proof.
\end{proof}
Now we are ready to prove Theorem \ref{defindexspt}.

\begin{proofof}[Theorem \ref{defindexspt}]
Let $0<\theta<\frac\pi 2$ and $\alpha\in \sqaut(\caA)$ satisfying
$
\omega_0\circ\alpha\circ\beta_g=\omega_0\circ\alpha$
for all $g\in G$.
We would like to show that
$\IG(\omega_0\circ\alpha,\theta)$ is not empty.

Let us set $\theta_{2.2}:=\theta$ and 
consider 
 $\theta_{0.8}$, $\theta_1$, $\theta_{1.2}$, $\theta_{1.8}$, $\theta_2$,
$\theta_{2.8}$, $\theta_3$, $\theta_{3.2}$
 satisfying (\ref{thetas1}) for this $\theta_{2.2}$.
 Because $\alpha\in \sqaut(\caA)$, there is a decompotision
 given by (\ref{sqaut}), (\ref{sqaut2}), (\ref{sqaut3}).
 Using this decomposition, set 
 \begin{align}
 \begin{split}
 &\alpha_1:=\alpha_{1D}\otimes\alpha_{1U}\\
  &\alpha_{1\zeta}:=
 \lmk
\alpha_{(\theta_1,\theta_2],\zeta}
\otimes \alpha_{(\theta_2,\theta_3],\zeta}\otimes
\alpha_{(\theta_3,\frac\pi 2],\zeta}
\rmk
\circ
\lmk
\alpha_{(\theta_{0.8}, \theta_{1.2}],\zeta}\otimes
\alpha_{(\theta_{1.8},\theta_{2.2}],\zeta}
\otimes \alpha_{(\theta_{2.8},\theta_{3.2}],\zeta}
\rmk\in\Aut\lmk
\caA_{\lmk \lmk C_{\theta_{0.8}}\rmk^c\rmk_\zeta}
\rmk
,\quad\zeta=U,D,\\
&\alpha_2:=\alpha_{[0,\theta_1]}
\in\Aut\lmk\caA_{C_{\theta_1}}\rmk.
 \end{split}
 \end{align}
%
%\begin{align}\label{sqaut}
%&\alpha=\inn\circ\lmk
%\alpha_{[0,\theta_1]}\otimes\alpha_{(\theta_1,\theta_2]}
%\otimes \alpha_{(\theta_2,\theta_3]}\otimes
%\alpha_{(\theta_3,\frac\pi 2]}
%\rmk
%\circ
%\lmk
%\alpha_{(\theta_{0.8}, \theta_{1.2}]}\otimes
%\alpha_{(\theta_{1.8},\theta_{2.2}]}
%\otimes \alpha_{(\theta_{2.8},\theta_{3.2}]}
%\rmk\notag\\
%&\alpha_{Y}=\alpha_{Y,D}\otimes\alpha_{Y,U},
%\quad\alpha_{X}=\alpha_{X,L}\otimes\alpha_{X,R},
%\end{align}
We have
$\alpha=\inn\circ\alpha_2\circ\alpha_1$.

We would like to show that 
$\lmk\alpha\circ
\beta_g^U\circ \alpha^{-1}, \alpha\circ
\beta_g\circ \alpha^{-1}\rmk$
satisfy the conditions of $(\alpha,\hat\alpha)$ in Lemma \ref{lem11s}.
We first show that they satisfy a decomposition corresponding to
 (\ref{hatal}) and (\ref{ald}).
 For $\Gamma=\bbZ^2, H_U$, we have
\begin{align}\label{saiku}
\lmk \beta_g^\Gamma\rmk^{-1}\alpha\circ\beta_g^\Gamma\circ\alpha^{-1}
=\inn\circ\lmk \beta_g^\Gamma\rmk^{-1}\circ\lmk \alpha_1\beta_g^\Gamma\alpha_1^{-1}\rmk
\lmk \alpha_1\beta_g^\Gamma\alpha_1^{-1}\rmk^{-1}
\alpha_2\alpha_1\beta_g^{\Gamma}\alpha_1^{-1}\alpha_2^{-1}.
\end{align}
The latter part $\lmk \alpha_1\beta_g^\Gamma\alpha_1^{-1}\rmk^{-1}
\alpha_2\alpha_1\beta_g^{\Gamma}\alpha_1^{-1}\alpha_2^{-1}$ decomposes to left and right.
To see this, first note that 
\begin{align}
\alpha_1^{-1}\alpha_2\alpha_1
=\alpha_{(\theta_{0.8}, \theta_{1.2}]}^{-1}\alpha_{[0,\theta_1]}\alpha_{(\theta_{0.8}, \theta_{1.2}]}
\in \Aut\lmk
\caA_{C_{\theta_{1.2}}}
\rmk.
\end{align}
Because the conjugation $\lmk \beta_g^\Gamma\rmk^{-1}\cdot  \beta_g^{\Gamma}$ does not change the
support of an automorphism, 
$\lmk \beta_g^\Gamma\rmk^{-1}\lmk
 \alpha_1^{-1}\alpha_2\alpha_1\rmk  \beta_g^{\Gamma}$
 is also supported on ${C_{\theta_{1.2}}}$.
 Therefore, we have
 \begin{align}
 \alpha_1\lmk \lmk \beta_g^\Gamma\rmk^{-1}\lmk
 \alpha_1^{-1}\alpha_2\alpha_1\rmk  \beta_g^{\Gamma}\rmk
\alpha_1^{-1}
=\alpha_{(\theta_1,\theta_2]}\alpha_{(\theta_{0.8}, \theta_{1.2}]}
\lmk \beta_g^\Gamma\rmk^{-1}
\alpha_{(\theta_{0.8}, \theta_{1.2}]}^{-1}\alpha_{[0,\theta_1]}\alpha_{(\theta_{0.8}, \theta_{1.2}]}
 \beta_g^\Gamma\alpha_{(\theta_{0.8}, \theta_{1.2}]}^{-1}\alpha_{(\theta_1,\theta_2]}^{-1}
 \end{align}
 Hence we get the left-right decomposition:
\begin{align}
\begin{split}
&\lmk \alpha_1\beta_g^\Gamma\alpha_1^{-1}\rmk^{-1}
\alpha_2\alpha_1\beta_g^{\Gamma}\alpha_1^{-1}\alpha_2^{-1}
=\alpha_1\lmk \lmk \beta_g^\Gamma\rmk^{-1}\lmk
 \alpha_1^{-1}\alpha_2\alpha_1\rmk  \beta_g^{\Gamma}\rmk
\alpha_1^{-1}\alpha_2^{-1}\\
&=\alpha_{(\theta_1,\theta_2]}\alpha_{(\theta_{0.8}, \theta_{1.2}]}
\lmk \beta_g^\Gamma\rmk^{-1}
\alpha_{(\theta_{0.8}, \theta_{1.2}]}^{-1}\alpha_{[0,\theta_1]}\alpha_{(\theta_{0.8}, \theta_{1.2}]}
 \beta_g^\Gamma\alpha_{(\theta_{0.8}, \theta_{1.2}]}^{-1}\alpha_{(\theta_1,\theta_2]}^{-1}\circ\alpha_{[0,\theta_1]}^{-1}\\
 &=\bigotimes_{\sigma=L,R}
\lmk
\alpha_{(\theta_1,\theta_2],\sigma}\alpha_{(\theta_{0.8}, \theta_{1.2}],\sigma}
\lmk \beta_g^{\Gamma_\sigma}\rmk^{-1}
\alpha_{(\theta_{0.8}, \theta_{1.2}],\sigma}^{-1}\alpha_{[0,\theta_1],\sigma
}\alpha_{(\theta_{0.8}, \theta_{1.2}],\sigma}
 \beta_g^{\Gamma_\sigma}\alpha_{(\theta_{0.8}, \theta_{1.2}],\sigma}^{-1}\alpha_{(\theta_1,\theta_2],\sigma}^{-1}\circ\alpha_{[0,\theta_1],\sigma}^{-1}
\rmk\\
&=:\bigotimes_{\sigma=L,R}\Xi_{\Gamma,g,\sigma}.
\end{split}
\end{align}
Here we set
\begin{align}
\Xi_{\Gamma,g,\sigma}
=\lmk
\alpha_{(\theta_1,\theta_2],\sigma}\alpha_{(\theta_{0.8}, \theta_{1.2}],\sigma}
\lmk \beta_g^{\Gamma_\sigma}\rmk^{-1}
\alpha_{(\theta_{0.8}, \theta_{1.2}],\sigma}^{-1}\alpha_{[0,\theta_1],\sigma
}\alpha_{(\theta_{0.8}, \theta_{1.2}],\sigma}
 \beta_g^{\Gamma_\sigma}\alpha_{(\theta_{0.8}, \theta_{1.2}],\sigma}^{-1}\alpha_{(\theta_1,\theta_2],\sigma}^{-1}\circ\alpha_{[0,\theta_1],\sigma}^{-1}
\rmk\in\Aut\lmk\caA_{\lmk C_{\theta_2}\rmk_\sigma}\rmk.
\end{align}
On the other hand, the first part of 
(\ref{saiku}) with $\Gamma=\bbZ^2,H_U$ satisfies 
\begin{align}
\begin{split}
&\beta_g^{-1}\alpha_1\beta_g\alpha_1^{-1}
% =\lmk \beta_g^{D}\rmk^{-1}\alpha_1\beta_g^{D}\alpha_1^{-1}
% \otimes
% \lmk \beta_g^{U}\rmk^{-1}\alpha_1\beta_g^{U}\alpha_1^{-1}
 =
 \xi_D\otimes\xi_U,
\quad
\lmk \beta_g^{U}\rmk^{-1}\alpha_1\beta_g^U\alpha_1^{-1}
% =\id_{\caA_{H_D}}
% \otimes
% \lmk \beta_g^{U}\rmk^{-1}\alpha_1\beta_g^{U}\alpha_1^{-1}
 =
 \id_{\caA_{H_D}}\otimes \xi_U
\end{split}
\end{align}
where 
\begin{align}
\xi_\zeta:=\lmk \beta_g^{\zeta}\rmk^{-1}\alpha_{1,\zeta}\beta_g^{\zeta}\alpha_{1,\zeta}^{-1}\in
\Aut\lmk
\caA_{\lmk \lmk C_{\theta_{0.8}}\rmk^c\rmk_\zeta}
\rmk,\quad
\zeta=U,D
\end{align}
Hence we obtain decompositions
\begin{align}\label{aahok1}
\begin{split}
&\lmk \beta_g^U\rmk^{-1}\circ\alpha\circ
\beta_g^U\circ \alpha^{-1}
=\inn\circ \lmk \id_{\caA_{H_D}}\otimes \xi_U
\rmk\circ \lmk \Xi_{H_U,g,L}\otimes  \Xi_{H_U,g,R}\rmk,\\
&
\lmk \beta_g\rmk^{-1}\circ\alpha\circ
\beta_g\circ \alpha^{-1}
=\inn\circ \lmk  \xi_D\otimes\xi_U\rmk
\circ  \lmk \Xi_{\bbZ^2,g,L}\otimes  \Xi_{\bbZ^2,g,R}\rmk.
\end{split}
\end{align}
Because  $\xi_{\zeta}\in \Aut\lmk
\caA_{\lmk \lmk C_{\theta_{0.8}}\rmk^c\rmk_\zeta}
\rmk$ commutes with $\beta_{g}^{C_{[0,\theta_{0.8}]}}$
and $\beta_{g}^{C_{[0,\theta_{0.8}],U}}$, we get
\begin{align}\label{aahok}
\begin{split}
&\alpha\circ
\beta_g^U\circ \alpha^{-1}
=\inn\circ \lmk \id_{\caA_{H_D}}\otimes 
\beta_{g}^{C_{(\theta_{0.8},\frac\pi 2],U}}\xi_U
\rmk\circ \lmk \beta_{g}^{C_{[0,\theta_{0.8}],L,U}}\Xi_{H_U,g,L}\otimes \beta_{g}^{C_{[0,\theta_{0.8}],R,U}} \Xi_{H_U,g,R}\rmk,\\
&
\alpha\circ
\beta_g\circ \alpha^{-1}
=\inn\circ \lmk \beta_{g}^{C_{(\theta_{0.8},\frac\pi 2],D}} \xi_D\otimes\beta_{g}^{C_{(\theta_{0.8},\frac\pi 2],U}}\xi_U\rmk
\circ  \lmk  \beta_{g}^{C_{[0,\theta_{0.8}],L}}\Xi_{\bbZ^2,g,L}\otimes  
\beta_{g}^{C_{[0,\theta_{0.8}],R}}\Xi_{\bbZ^2,g,R}\rmk.
\end{split}
\end{align}

Furthermore, from $\beta_g$-invariance of 
$\omega_0\circ\alpha$, we have
\begin{align}\label{invt}
\omega_0\circ \alpha\circ
\beta_g\circ \alpha^{-1}
=\omega_0.
\end{align}

Now we apply Lemma \ref{lem11s} 
for $\mkA_\sigma$, $\mkB_{\sigma\zeta}$,
$\mkC_\zeta$ replaced by $\caA_{\lmk C_{[0,\theta_{0.8}]}\rmk_\sigma}$
 $\caA_{\lmk C_{(\theta_{0.8},\theta_{2}]}\rmk_{\sigma,\zeta}}$
$\caA_{\lmk C_{(\theta_2,\frac\pi 2]}\rmk_\zeta}$,
for $\sigma=L,R$, $\zeta=D,U$.
By (\ref{invt}) and (\ref{aahok}),
$\lmk\alpha\circ
\beta_g^U\circ \alpha^{-1}, \alpha\circ
\beta_g\circ \alpha^{-1}\rmk$
satisfy the conditions of $(\alpha,\hat\alpha)$ in Lemma \ref{lem11s},
for $\omega_0$ and its restrictions.
Applying Lemma \ref{lem11s}, there are $\tilde \eta_{\sigma,g}\in\Aut\lmk
\caA_{\lmk C_{\theta_2}\rmk_\sigma}\rmk$, $g\in G$, $\sigma=L,R$
such that
\begin{align}
\omega_0\circ\alpha\circ
\beta_g^U\circ \alpha^{-1}
\sim_{q.e.}
\omega_0\circ
\lmk
\tilde\eta_{Lg}\otimes \tilde\eta_{Rg}
\rmk,\quad g\in G.
\end{align}
Because both of $\omega_0\circ\alpha\circ
\beta_g^U\circ \alpha^{-1}$ and
$\omega_0\circ
\lmk
\tilde\eta_{Lg}\otimes \tilde\eta_{Rg}
\rmk$
are pure,
by Kadison's transitibity theorem, there exists
a unitary $\tilde v_g\in\caU(\caA)$
such
that
\begin{align}\label{isu}
\omega_0\circ\alpha\circ
\beta_g^U\circ \alpha^{-1}
=
\omega_0\circ\Ad_{\tilde v_g}\circ
\lmk
\tilde\eta_{Lg}\otimes \tilde\eta_{Rg}
\rmk,\quad g\in G.
\end{align}
We define
\begin{align}\label{bgan}
\tilde \beta_g:=
\Ad\lmk \alpha^{-1}\lmk\tilde v_{g^{-1}}\rmk\rmk
\circ\alpha^{-1}
\circ\lmk
\tilde\eta_{Lg^{-1}}\otimes \tilde\eta_{Rg^{-1}}
\rmk\circ\alpha\circ\beta_g^U,\quad g\in G.
\end{align}
It suffices to show that $(\tilde\beta_g)\in \IG(\omega_0\circ\alpha,\theta)=\IG(\omega_0\circ\alpha,\theta_{2.2})$.
By (\ref{isu}), we have $\omega_0\circ\alpha\circ\tilde\beta_g=\omega_0\circ\alpha$.
Therefore, what is left to be proven is that
there are
$\eta_{g}^\sigma\in\Aut\lmk \lmk C_\theta\rmk_\sigma\rmk$, $g\in G$, $\sigma=L,R$
such that
\begin{align}\label{decomgan}
\tilde\beta_g=\inn\circ\lmk \eta_g^L\otimes\eta_g^R\rmk\circ\beta_g^U,\; \text{for all}\;g\in G
\end{align}
By the decomposition (\ref{sqaut}) and the fact that
$\tilde\eta_{Lg^{-1}}\otimes \tilde\eta_{Rg^{-1}}$
has support in $C_{\theta_2}$,
we have
\begin{align}
\begin{split}
&\alpha^{-1}
\circ\lmk
\tilde\eta_{Lg^{-1}}\otimes \tilde\eta_{Rg^{-1}}
\rmk\circ\alpha\\
&=\inn\circ
\lmk
\alpha_{(\theta_{0.8}, \theta_{1.2}]}\otimes
\alpha_{(\theta_{1.8},\theta_{2.2}]}
\rmk^{-1}
\lmk
\alpha_{[0,\theta_1]}\otimes\alpha_{(\theta_1,\theta_2]}
\rmk^{-1}
\lmk
\tilde\eta_{Lg^{-1}}\otimes \tilde\eta_{Rg^{-1}}
\rmk
\lmk
\alpha_{[0,\theta_1]}\otimes\alpha_{(\theta_1,\theta_2]}
\rmk
\circ
\lmk
\alpha_{(\theta_{0.8}, \theta_{1.2}]}\otimes
\alpha_{(\theta_{1.8},\theta_{2.2}]}
\rmk\\
&=\inn\circ\lmk \eta_g^L\otimes\eta_g^R\rmk,
\end{split}
\end{align}
where 
\begin{align}
\begin{split}
&\eta_{g}^\sigma=\lmk
\alpha_{(\theta_{0.8}, \theta_{1.2}],\sigma}\otimes
\alpha_{(\theta_{1.8},\theta_{2.2}],\sigma}
\rmk^{-1}
\lmk
\alpha_{[0,\theta_1],\sigma}\otimes\alpha_{(\theta_1,\theta_2],\sigma}
\rmk^{-1}
\lmk
\tilde\eta_{\sigma g^{-1}}
\rmk
\lmk
\alpha_{[0,\theta_1],\sigma}\otimes\alpha_{(\theta_1,\theta_2],\sigma}
\rmk
\circ
\lmk
\alpha_{(\theta_{0.8}, \theta_{1.2}],\sigma}\otimes
\alpha_{(\theta_{1.8},\theta_{2.2}], \sigma}
\rmk\\
&\quad\quad\quad\in \Aut\lmk \lmk C_{\theta_{2.2}}\rmk_\sigma\rmk,\quad
\sigma=L,R
\end{split}
\end{align}
Substituting this to (\ref{bgan}), we obtain (\ref{decomgan}).
This completes the proof.
\end{proofof}

\section{The stability of the index $h(\omega)$ }\label{stabilitysec}
In this section we prove the stability of the index $h(\omega)$
with respect to $\gamma\in \guaut(\caA)$.
\begin{thm}\label{stabilitythm}
Let $\omega\in\QLS$ with $\IG(\omega)\neq\emptyset$.
Let $\gamma\in \guaut\lmk\caA\rmk$.
Then we have $\omega\circ\gamma\in\QLS$ with $\IG(\omega\circ\gamma)\neq\emptyset$
and 
\begin{align}
h(\omega\circ\gamma)=h(\omega).
\end{align}
\end{thm}
\begin{proof}
{\it Step 1.}
From $\omega\in \QLS$, there is an 
$ \alpha\in \eaut(\omega)$.
For any $0<\theta<\frac\pi 2$ fixed,
we show that $\caD^{\theta}_{\alpha\circ\gamma}\neq\emptyset$, hence
$\alpha\circ\gamma\in\qaut(\caA)$ and 
$\omega\circ\gamma=\omega_{0}\circ\alpha\gamma\in \QLS$.
Set $\theta_{1.2}:=\theta$ and choose 
\begin{align}
0<\theta_{0}<\theta_{0.8}<\theta_1<\theta_{1.2}:=\theta<\theta_{1.8}<\theta_2<\theta_{2.2}<
\theta_{2.8}<\theta_3<\theta_{3.2}<\frac\pi 2.
\end{align}
Because $\alpha\in\qaut(\caA)$, there exists some
$(\alpha_L,\alpha_R,\Theta)\in \caD_{\alpha}^{\theta_{2}}$.
Setting $\alpha_0:=\alpha_L\otimes\alpha_R$, we have $\alpha=\inn\circ\alpha_0\circ\Theta$.
Because $\gamma\in \guaut\lmk\caA\rmk$, 
there are 
$\gamma_{H}\in \haut(\caA)$ and 
$\gamma_{C}\in \gsqaut(\caA)$
{such that}
\begin{align}\label{gcz}
\gamma=\gamma_{C}\circ\gamma_{H}.
\end{align}
Because $\gamma_{H}\in \haut(\caA)$, we may decompose $\gamma_{H}$
as 
\begin{align}\label{ght}
\gamma_{H}=\inn\circ\lmk \gamma_{H, {L}}\otimes \gamma_{H,R}\rmk
=\inn\circ\gamma_{0}
\end{align}
with
some $\gamma_{H, {\sigma}}\in \Aut\lmk \caA_{{\lmk C_{\theta_{0}}\rmk_\sigma}}\rmk$, $\sigma=L,R$.
We set $\gamma_{0}:=\gamma_{H, {L}}\otimes \gamma_{H,R}\in 
\Aut\lmk \caA_{{ C_{\theta_{0}}}}\rmk$.
By definition, $\gamma_{C}\in \gsqaut(\caA)$, allows a
decomposition
\begin{align}\label{sqautg}
&\gamma_{C}=\inn\circ\gamma_{CS}\notag\\
&\gamma_{CS}=\lmk
\gamma_{[0,\theta_1]}\otimes\gamma_{(\theta_1,\theta_2]}
\otimes \gamma_{(\theta_2,\theta_3]}\otimes
\gamma_{(\theta_3,\frac\pi 2]}
\rmk
\circ
\lmk
\gamma_{(\theta_{0.8}, \theta_{1.2}]}\otimes
\gamma_{(\theta_{1.8},\theta_{2.2}]}
\otimes \gamma_{(\theta_{2.8},\theta_{3.2}]}
\rmk
\end{align}
with  \begin{align}\label{sqaut2g}
  \begin{split}
&  \gamma_X:=\bigotimes_{\sigma=L,R,\zeta=D,U} \gamma_{X,\sigma,\zeta},\quad
 \gamma_{[0,\theta_{1}]}:=\bigotimes_{\sigma=L,R}\gamma_{[0,\theta_{1}],\sigma},\quad
 \gamma_{(\theta_3,\frac\pi 2]}:=\bigotimes_{\zeta=D,U}  \gamma_{(\theta_3,\frac\pi 2],\zeta}\\
 &\gamma_{X,\sigma,\zeta}\in \Aut\lmk\caA_{C_{X,\sigma,\zeta}}\rmk,\quad
 \gamma_{X,\sigma}:=\bigotimes_{\zeta=U,D}\gamma_{X,\sigma,\zeta},\quad
\gamma_{X,\zeta}:=\bigotimes_{\sigma=L,R}\gamma_{X,\sigma,\zeta}\\
&\gamma_{[0,\theta_{1}],\sigma}\in \Aut\lmk\caA_{C_{[0,\theta_{1}],\sigma}}\rmk,\quad
 \gamma_{(\theta_3,\frac\pi 2],\zeta}\in \Aut\lmk\caA_{C_{(\theta_3,\frac\pi 2],\zeta}}\rmk, 
  \end{split} 
  \end{align}
 for
 \begin{align}\label{sqaut3g}
 X=(\theta_1,\theta_2], (\theta_2,\theta_3],
% (\theta_{-0.2},\theta_{0.2}],
 (\theta_{0.8},\theta_{1.2}],
 (\theta_{1.8},\theta_{2.2}], 
(\theta_{2.8},\theta_{3.2}],\quad \sigma=L,R,\quad \zeta=D,U.
 \end{align}
 Here we have
 \begin{align}
\gamma_{I}\circ\beta_g^{U}=\beta_g^{U}\circ\gamma_{I}\;\quad\text{for all}\quad g\in G,
\end{align}
for any 
\begin{align}\label{sqaut3g}
I=[0,\theta_1],(\theta_1,\theta_2], (\theta_2,\theta_3], \left(\theta_3,\frac\pi 2\right],
(\theta_{0.8}, \theta_{1.2}], (\theta_{1.8},\theta_{2.2}], (\theta_{2.8},\theta_{3.2}].
 \end{align}
%\begin{align}\label{sqautg}
%&\gamma_{CS}:=\lmk
%\gamma_{[0,\theta_1]}\otimes\gamma_{(\theta_1,\theta_2]}
%\otimes \gamma_{(\theta_2,\theta_3]}\otimes
%\gamma_{(\theta_3,\frac\pi 2]}
%\rmk
%\circ
%\lmk
%\gamma_{(\theta_{0.8}, \theta_{1.2}]}\otimes
%\gamma_{(\theta_{1.8},\theta_{2.2}]}
%\otimes \gamma_{(\theta_{2.8},\theta_{3.2}]}
%\rmk\notag\\
%&\gamma_{Y}=\gamma_{Y,D}\otimes\gamma_{Y,U},
%\quad\gamma_{X}=\gamma_{X,L}\otimes\gamma_{X,R},
%\end{align}
%with
%\begin{align}
%&\gamma_{I}\in \Aut\lmk\caA_{\caC_I}\rmk,\notag\\
%&\gamma_{Y,D}\in \Aut\lmk\caA_{(\caC_Y)_D}\rmk,\quad \gamma_{Y,U}\in \Aut\lmk\caA_{(\caC_Y)_U}\rmk,\notag\\
%&\gamma_{X,L}\in \Aut\lmk\caA_{(\caC_X)_L}\rmk,\quad \gamma_{X,R}\in \Aut\lmk\caA_{(\caC_X)_R}\rmk,
%\end{align}
%for 
%\begin{align}
%&I=[0,\theta_1],(\theta_1,\theta_2], (\theta_2,\theta_3], \left(\theta_3,\frac\pi 2\right],
%(\theta_{0.8}, \theta_{1.2}], (\theta_{1.8},\theta_{2.2}], (\theta_{2.8},\theta_{3.2}],\\
%&Y=(\theta_1,\theta_2], (\theta_2,\theta_3], \left(\theta_3,\frac\pi 2\right],
%(\theta_{0.8}, \theta_{1.2}], (\theta_{1.8},\theta_{2.2}], (\theta_{2.8},\theta_{3.2}],\\
%&X=[0,\theta_1], (\theta_1,\theta_2], \theta_2,\theta_3], (\theta_{0.8}, \theta_{1.2}], \theta_{1.8},\theta_{2.2}], (\theta_{2.8},\theta_{3.2}].
%\end{align}
Set
\begin{align}\label{athens}
\hat\Theta:=\Theta\circ
\lmk
\gamma_{(\theta_2,\theta_3]}\otimes
\gamma_{(\theta_3,\frac\pi 2]}
\rmk
\circ
\lmk
\gamma_{(\theta_{1.8},\theta_{2.2}]}
\otimes \gamma_{(\theta_{2.8},\theta_{3.2}]}
\rmk
\in\Aut\lmk \caA_{C_{\theta_{1.8}}^{c}}\rmk
\subset \Aut\lmk \caA_{C_{\theta_{1.2}}^{c}}\rmk,
\end{align}
and
\begin{align}\label{rome}
\hat\alpha_{\sigma}:=
\alpha_{\sigma}\circ
\lmk
\gamma_{[0,\theta_1],\sigma}\otimes\gamma_{(\theta_1,\theta_2],\sigma}
\rmk
\circ
\gamma_{(\theta_{0.8}, \theta_{1.2}],\sigma}
\circ\gamma_{H,\sigma}\in\Aut(\caA_{H_{\sigma}}),\quad\sigma=L,R.
\end{align}
We claim
\begin{align}\label{algm}
\alpha\circ\gamma=\inn\circ\lmk\hat\alpha_{L}\otimes\hat \alpha_{R}\rmk\circ\hat\Theta.
\end{align}
This means $(\hat\alpha_{L},\hat\alpha_{R},\hat\Theta)\in\caD_{\alpha\gamma}^{\theta_{1.2}}$,
hence $\caD_{\alpha\gamma}^{\theta}=\caD_{\alpha\gamma}^{\theta_{1.2}}\neq \emptyset$.
The claim (\ref{algm}) can be checked as follows.
Note that $\gamma_{(\theta_2,\theta_3]}\otimes
\gamma_{(\theta_3,\frac\pi 2]}$ and $\gamma_{(\theta_{0.8}, \theta_{1.2}]}$
commute because of their disjoint supports.
Because 
$\Theta\in\Aut\lmk\caA_{C_{\theta_{2}}^{c}}\rmk$,
it commutes with $\gamma_{[0,\theta_1]}\otimes\gamma_{(\theta_1,\theta_2]}$
and $\gamma_{(\theta_{0.8}, \theta_{1.2}]}$.
Therefore, we have
\begin{align}
&\alpha\circ\gamma=\inn\circ\alpha_{0}\circ \Theta\circ\lmk
\gamma_{[0,\theta_1]}\otimes\gamma_{(\theta_1,\theta_2]}
\otimes \gamma_{(\theta_2,\theta_3]}\otimes
\gamma_{(\theta_3,\frac\pi 2]}
\rmk
\circ
\lmk
\gamma_{(\theta_{0.8}, \theta_{1.2}]}\otimes
\gamma_{(\theta_{1.8},\theta_{2.2}]}
\otimes \gamma_{(\theta_{2.8},\theta_{3.2}]}
\rmk\circ\gamma_{0}\notag\\
&=\inn\circ\alpha_{0}\circ\lmk
\gamma_{[0,\theta_1]}\otimes\gamma_{(\theta_1,\theta_2]}
\rmk\circ\gamma_{(\theta_{0.8}, \theta_{1.2}]}\circ
 \Theta\circ\lmk \gamma_{(\theta_2,\theta_3]}\otimes
\gamma_{(\theta_3,\frac\pi 2]}\rmk
\circ\lmk
\gamma_{(\theta_{1.8},\theta_{2.2}]}
\otimes \gamma_{(\theta_{2.8},\theta_{3.2}]}
\rmk\circ\gamma_{0}\notag\\
&=\inn\circ\alpha_{0}\circ\lmk
\gamma_{[0,\theta_1]}\otimes\gamma_{(\theta_1,\theta_2]}
\rmk\circ\gamma_{(\theta_{0.8}, \theta_{1.2}]}\circ\hat\Theta\circ\gamma_{0}\label{paris}
\end{align}
Because $\gamma_{0}\in \Aut\lmk \caA_{{ C_{\theta_{0}}}}\rmk$
and
$\hat\Theta
\in\Aut\lmk \caA_{C_{\theta_{1.8}}^{c}}\rmk$ commute,
we have
\begin{align}
\alpha\circ\gamma=(\ref{paris})=
\inn\circ\alpha_{0}\circ\lmk
\gamma_{[0,\theta_1]}\otimes\gamma_{(\theta_1,\theta_2]}
\rmk\circ\gamma_{(\theta_{0.8}, \theta_{1.2}]}\circ\gamma_{0}\circ\hat\Theta
=\inn\circ\lmk\hat\alpha_{L}\otimes\hat \alpha_{R}\rmk\circ\hat\Theta,
\end{align}
proving (\ref{algm}).
\\
{\it Step 2.}
From $\IG(\omega)\neq\emptyset$, we fix a
$0<\theta_{0}<\frac\pi 2$ such that
$\IG(\omega,\theta_{0})\neq\emptyset$.
We choose $\theta_{0.8},\theta_1,\theta_{1.2},\theta_{1.8},\theta_2,\theta_{2.2},
\theta_{2.8},\theta_3,\theta_{3.2}$ such that
\begin{align}
0<\theta_{0}<\theta_{0.8}<\theta_1<\theta_{1.2}<\theta_{1.8}<\theta_2<\theta_{2.2}<
\theta_{2.8}<\theta_3<\theta_{3.2}<\frac\pi 2.
\end{align}
For these $\theta$s, we associate the decomposition of $\gamma$ in {\it Step 1}.
Fix $\lmk \tilde \beta_{g}\rmk\in \IG(\omega,\theta_{0})$ and
$
(\eta_{g}^{\sigma})\in
\caT(\theta_{0}, (\tilde\beta_g))$.
Set $\eta_{g}:=\eta_{g}^{L}\otimes\eta_{g}^{R}$.
Note that  $(\eta_{g}^{\sigma})$ also belongs to
$\caT(\theta_{2}, (\tilde\beta_g))$.
Set
\begin{align}\label{miyajima}
\hat\eta_{g}^{\sigma}:=
\lmk \gamma_{[0,\theta_1],\sigma}\gamma_{(\theta_{0.8}, \theta_{1.2}],\sigma}\gamma_{H,\sigma}\rmk^{-1}
\eta_{g}^{\sigma}\lmk
\beta_{g}^{\sigma U}
\gamma_{[0,\theta_1],\sigma}\gamma_{(\theta_{0.8}, \theta_{1.2}],\sigma}\gamma_{H,\sigma}
\lmk
\beta_{g}^{\sigma U}
\rmk^{-1}\rmk\in\Aut\lmk \caA_{\lmk C_{\theta_{1.2}}\rmk_{\sigma}}\rmk,
\end{align}
for $\sigma=L,R$.
We also set $\hat\eta_{g}:=\hat \eta_{g}^{L}\otimes \hat \eta_{g}^{R}$.
We claim that $\lmk \gamma^{-1}\tilde \beta_{g}\gamma\rmk\in \IG(\omega\circ\gamma,\theta_{1.2})$
with $(\hat\eta_{g}^{\sigma})\in \caT\lmk\theta_{1.2}, \lmk \gamma^{-1}\tilde \beta_{g}\gamma\rmk\rmk$.
Clearly we have
\begin{align}
\omega\circ\gamma\circ\lmk \gamma^{-1}\tilde \beta_{g}\gamma\rmk
=\omega\circ\tilde \beta_{g}\circ\gamma
=\omega\circ\gamma.
\end{align}
Therefore, what remains to be shown is 
\begin{align}
\gamma^{-1}\tilde \beta_{g}\gamma
=\inn \circ\lmk\hat\eta_{g}^L\otimes \hat\eta_{g}^R\rmk\circ \beta_{g}^{U}
\end{align}
To see this, we first have 
\begin{align}
&\gamma^{-1}\circ\eta_{g}\circ\gamma\notag\\
&=\inn\circ\gamma_{0}^{-1}\circ
\lmk
\gamma_{(\theta_{0.8}, \theta_{1.2}]}\otimes
\gamma_{(\theta_{1.8},\theta_{2.2}]}
\otimes \gamma_{(\theta_{2.8},\theta_{3.2}]}
\rmk^{-1}\circ
\lmk
\gamma_{[0,\theta_1]}\otimes\gamma_{(\theta_1,\theta_2]}
\otimes \gamma_{(\theta_2,\theta_3]}\otimes
\gamma_{(\theta_3,\frac\pi 2]}
\rmk^{-1}\notag\\
&\circ\eta_{g}\circ\lmk
\gamma_{[0,\theta_1]}\otimes\gamma_{(\theta_1,\theta_2]}
\otimes \gamma_{(\theta_2,\theta_3]}\otimes
\gamma_{(\theta_3,\frac\pi 2]}
\rmk
\circ
\lmk
\gamma_{(\theta_{0.8}, \theta_{1.2}]}\otimes
\gamma_{(\theta_{1.8},\theta_{2.2}]}
\otimes \gamma_{(\theta_{2.8},\theta_{3.2}]}
\rmk
\gamma_{0}\label{sweden}
\end{align}
from the decomposition (\ref{gcz}), (\ref{ght}) (\ref{sqautg}).
Because $\gamma_{(\theta_1,\theta_2]}
\otimes \gamma_{(\theta_2,\theta_3]}\otimes
\gamma_{(\theta_3,\frac\pi 2]}$ commutes with
$\eta_{g}\in \Aut\lmk \caA_{{ C_{\theta_{0}}}}\rmk$ and
$\gamma_{(\theta_{1.8},\theta_{2.2}]}
\otimes \gamma_{(\theta_{2.8},\theta_{3.2}]}$
commutes with $\lmk \gamma_{[0,\theta_1]}\rmk^{-1}\eta_{g}\gamma_{[0,\theta_1]}\in
 \Aut\lmk \caA_{{ C_{\theta_{1}}}}\rmk$, we have
 \begin{align}
 &\gamma^{-1}\circ\eta_{g}\circ\gamma\notag\\
&=(\ref{sweden})
=
\inn\circ\gamma_{0}^{-1}\circ
\lmk
\gamma_{(\theta_{0.8}, \theta_{1.2}]}\otimes
\gamma_{(\theta_{1.8},\theta_{2.2}]}
\otimes \gamma_{(\theta_{2.8},\theta_{3.2}]}
\rmk^{-1}\notag\\
&\circ
\lmk
\gamma_{[0,\theta_1]}
\rmk^{-1}\circ\eta_{g}
\circ\lmk
\gamma_{[0,\theta_1]}
\rmk
\circ
\lmk
\gamma_{(\theta_{0.8}, \theta_{1.2}]}\otimes
\gamma_{(\theta_{1.8},\theta_{2.2}]}
\otimes \gamma_{(\theta_{2.8},\theta_{3.2}]}
\rmk
\gamma_{0}\notag\\
&=\inn\circ\gamma_{0}^{-1}\circ
\lmk
\gamma_{(\theta_{0.8}, \theta_{1.2}]}
\rmk^{-1}\circ
\lmk
\gamma_{[0,\theta_1]}
\rmk^{-1}\circ\eta_{g}
\circ\lmk
\gamma_{[0,\theta_1]}
\rmk
\circ
\lmk
\gamma_{(\theta_{0.8}, \theta_{1.2}]}
\rmk
\gamma_{0}\label{spain}
 \end{align}
On the other hand, because $\gamma_{CS}$ and $\beta_{g}^{U}$ commute, we have
\begin{align}
&\gamma^{-1}\circ\beta_{g}^{U}\circ\gamma=\inn\gamma_{0}^{-1}\circ\gamma_{CS}^{-1}\beta_{g}^{U}\gamma_{CS}\gamma_{0}
=\inn\gamma_{0}^{-1}\circ\beta_{g}^{U}\gamma_{0}.
\label{italy}
\end{align}
Combining (\ref{spain}) and (\ref{italy}), we obtain
\begin{align}
&\gamma^{-1}\tilde \beta_{g}\gamma
=\inn\circ \gamma_{0}^{-1}
\lmk
\gamma_{(\theta_{0.8}, \theta_{1.2}]}
\rmk^{-1}\circ
\lmk
\gamma_{[0,\theta_1]}
\rmk^{-1}\circ\eta_{g}
\circ\lmk
\gamma_{[0,\theta_1]}
\rmk
\circ
\lmk
\gamma_{(\theta_{0.8}, \theta_{1.2}]}
\rmk
\gamma_{0}
\circ
\gamma_{0}^{-1}\circ\beta_{g}^{U}\gamma_{0}\notag\\
&=\inn\circ \gamma_{0}^{-1}
\lmk
\gamma_{(\theta_{0.8}, \theta_{1.2}]}
\rmk^{-1}\circ
\lmk
\gamma_{[0,\theta_1]}
\rmk^{-1}\circ\eta_{g}\beta_{g}^{U}
\circ\lmk
\gamma_{[0,\theta_1]}
\rmk
\circ
\lmk
\gamma_{(\theta_{0.8}, \theta_{1.2}]}
\rmk
\circ\gamma_{0}\notag\\
&=\inn \circ\lmk\hat\eta_{g}^L\otimes \hat\eta_{g}^R\rmk\circ \beta_{g}^{U}.
\end{align}
In the second equality, we used the fact that $
\gamma_{[0,\theta_1]}
\gamma_{(\theta_{0.8}, \theta_{1.2}]}
$
and $\beta_{g}^{U}$ commute.
This completes the proof of the claim.
\\
{\it Step 3.}
We use the setting and notation of {\it Step 1.}(with $\theta_{0}$ chosen in {\it Step 2.})
and  {\it Step. 2}.
By Lemma \ref{ichi}, there exists 
\begin{align}
\lmk (W_g), (u_\sigma(g,h))\rmk\in \IP\lmk
\omega, \alpha, \theta_{2},
(\tilde\beta_g), (\eta_{g}^\sigma),
(\alpha_L,\alpha_R,\Theta)
\rmk.
\end{align}
Now we have
\begin{align}\label{berlin}
\omega\circ\gamma\in \QLS,\;
\alpha\circ\gamma\in \eaut(\omega\circ\gamma),\;
\lmk \gamma^{-1}\circ\tilde \beta_{g}\circ\gamma\rmk\in \IG(\omega\circ\gamma,\theta_{1.2}),\;
(\hat\eta_{g}^{\sigma})\in \caT\lmk\theta_{1.2}, \lmk \gamma^{-1}\tilde \beta_{g}\gamma\rmk\rmk,\;
(\hat \alpha_L,\hat  \alpha_R,\hat\Theta)\in\caD_{\alpha\gamma}^{\theta_{1.2}}
.
\end{align}
We claim
\begin{align}\label{onomichi}
\lmk (W_g), (u_\sigma(g,h))\rmk\in \IP\lmk
\omega\circ\gamma, \alpha\circ\gamma, \theta_{1.2},
(\gamma^{-1}\tilde\beta_g\gamma), (\hat \eta_{g}^\sigma),
(\hat \alpha_L,\hat \alpha_R,\hat \Theta)
\rmk.
\end{align}
This immediately implies $h(\omega)=h(\omega\circ\gamma)$.
To prove the claim, we first see from (\ref{athens}) and (\ref{rome}) that
\begin{align}
&\lmk \hat \alpha_L\otimes\hat \alpha_R\rmk\circ
\hat\Theta\circ \gamma_{0}^{-1}
\lmk
\gamma_{(\theta_{0.8}, \theta_{1.2}]}
\rmk^{-1}\circ
\lmk
\gamma_{[0,\theta_1]}
\rmk^{-1}\notag\\
&=\alpha_{0}\circ
\lmk
\gamma_{[0,\theta_1]}\otimes\gamma_{(\theta_1,\theta_2]}
\rmk
\circ
\gamma_{(\theta_{0.8}, \theta_{1.2}]}
\circ\gamma_{0}
\circ\Theta\circ
\lmk
\gamma_{(\theta_2,\theta_3]}\otimes
\gamma_{(\theta_3,\frac\pi 2]}
\rmk
\circ
\lmk
\gamma_{(\theta_{1.8},\theta_{2.2}]}
\otimes \gamma_{(\theta_{2.8},\theta_{3.2}]}
\rmk
\circ \gamma_{0}^{-1}
\lmk
\gamma_{(\theta_{0.8}, \theta_{1.2}]}
\rmk^{-1}\circ
\lmk
\gamma_{[0,\theta_1]}
\rmk^{-1}\notag\\
&=\alpha_{0}\circ
\lmk
\gamma_{[0,\theta_1]}\otimes\gamma_{(\theta_1,\theta_2]}
\rmk
\circ
\Theta\circ
\lmk
\gamma_{(\theta_2,\theta_3]}\otimes
\gamma_{(\theta_3,\frac\pi 2]}
\rmk
\circ
\lmk
\gamma_{(\theta_{1.8},\theta_{2.2}]}
\otimes \gamma_{(\theta_{2.8},\theta_{3.2}]}
\rmk
\circ
\lmk
\gamma_{[0,\theta_1]}
\rmk^{-1}\label{greece}
\end{align}
because $\gamma_{(\theta_{0.8}, \theta_{1.2}]}\circ \gamma_{0}\in  \Aut\lmk \caA_{{ C_{\theta_{1.2}}}}\rmk$
and $\Theta\circ
\lmk
\gamma_{(\theta_2,\theta_3]}\otimes
\gamma_{(\theta_3,\frac\pi 2]}
\rmk
\circ
\lmk
\gamma_{(\theta_{1.8},\theta_{2.2}]}
\otimes \gamma_{(\theta_{2.8},\theta_{3.2}]}
\rmk\in  \Aut\lmk \caA_{{ C_{\theta_{1.8}}}^{c}}\rmk$
commute.
Furthermore, because $\gamma_{[0,\theta_1]}$
and $\Theta\circ
\lmk
\gamma_{(\theta_2,\theta_3]}\otimes
\gamma_{(\theta_3,\frac\pi 2]}
\rmk
\circ
\lmk
\gamma_{(\theta_{1.8},\theta_{2.2}]}
\otimes \gamma_{(\theta_{2.8},\theta_{3.2}]}
\rmk
\in  \Aut\lmk \caA_{{ C_{\theta_{1.8}}}^{c}}\rmk$
commute,
while $\gamma_{(\theta_1,\theta_2]}$ and 
$\Theta\in  \Aut\lmk \caA_{{ C_{\theta_{2}}}^{c}}\rmk$
commute,
we have
\begin{align}
&\lmk \hat \alpha_L\otimes\hat \alpha_R\rmk\circ
\hat\Theta\circ \gamma_{0}^{-1}
\lmk
\gamma_{(\theta_{0.8}, \theta_{1.2}]}
\rmk^{-1}\circ
\lmk
\gamma_{[0,\theta_1]}
\rmk^{-1}
=(\ref{greece})
=\alpha_{0}\circ
\gamma_{(\theta_1,\theta_2]}
\circ
\Theta\circ
\lmk
\gamma_{(\theta_2,\theta_3]}\otimes
\gamma_{(\theta_3,\frac\pi 2]}
\rmk
\circ
\lmk
\gamma_{(\theta_{1.8},\theta_{2.2}]}
\otimes \gamma_{(\theta_{2.8},\theta_{3.2}]}
\rmk\notag\\
&=\alpha_{0}\circ
\Theta\circ\gamma_{(\theta_1,\theta_2]}
\circ\lmk
\gamma_{(\theta_2,\theta_3]}\otimes
\gamma_{(\theta_3,\frac\pi 2]}
\rmk
\circ
\lmk
\gamma_{(\theta_{1.8},\theta_{2.2}]}
\otimes \gamma_{(\theta_{2.8},\theta_{3.2}]}
\rmk
=\alpha_{0}\circ
\Theta
\circ\hat\gamma.
\end{align}
Here $\hat\gamma:=\gamma_{(\theta_1,\theta_2]}
\circ\lmk
\gamma_{(\theta_2,\theta_3]}\otimes
\gamma_{(\theta_3,\frac\pi 2]}
\rmk
\circ
\lmk
\gamma_{(\theta_{1.8},\theta_{2.2}]}
\otimes \gamma_{(\theta_{2.8},\theta_{3.2}]}
\rmk\in  \Aut\lmk \caA_{{ C_{\theta_{1}}}^{c}}\rmk$
commutes with $\beta_{g}^{U}$.
Combining this and 
\begin{align}
\hat \eta_{g}\beta_{g}^{U}=
\lmk \gamma_{[0,\theta_1]}\gamma_{(\theta_{0.8}, \theta_{1.2}]}\gamma_{0}\rmk^{-1}
\eta_{g}
\beta_{g}^{ U}
\gamma_{[0,\theta_1]}\gamma_{(\theta_{0.8}, \theta_{1.2}]}\gamma_{0},
\end{align}
we obtain
\begin{align}\label{madrid}
\pi_{0}\circ 
\lmk \hat \alpha_L\otimes\hat \alpha_R\rmk\circ
\hat\Theta\circ\hat\eta_{g}\beta_{g}^{U}
\lmk \hat\Theta\rmk^{-1}
\lmk \hat \alpha_L\otimes\hat \alpha_R\rmk^{-1}
=\pi_{0}\circ\alpha_{0}\circ
\Theta
\circ\hat\gamma\circ \eta_{g}
\beta_{g}^{ U}\circ
\hat\gamma^{-1}\circ\Theta^{-1}\circ \alpha_{0}^{-1}.
\end{align}
Because 
$\hat\gamma$
commutes with $\beta_{g}^{U}$
and $ \eta_{g}\in  \Aut\lmk \caA_{{ C_{\theta_{0}}}}\rmk$ commutes with 
$\hat\gamma\in \Aut\lmk \caA_{{ C_{\theta_{1}}}^{c}}\rmk$,
we have
\begin{align}
\pi_{0}\circ 
\lmk \hat \alpha_L\otimes\hat \alpha_R\rmk\circ
\hat\Theta\circ\hat\eta_{g}\beta_{g}^{U}
\lmk \hat\Theta\rmk^{-1}
\lmk \hat \alpha_L\otimes\hat \alpha_R\rmk^{-1}=(\ref{madrid})
=\pi_{0}\circ\alpha_{0}\circ
\Theta
\circ \eta_{g}
\beta_{g}^{ U}
\circ\Theta^{-1}\circ \alpha_{0}^{-1}
=\Ad\lmk W_{g}\rmk\circ\pi_{0}
\end{align}
Hence the condition for $W_{g}$ in (\ref{onomichi}) is checked.
On the other hand, 
substituting  (\ref{rome}) and (\ref{miyajima}), we get
\begin{align}
&\pi_R\circ\hat \alpha_R\circ\hat \eta_g^R\beta_g^{R U}
\hat \eta_h^R\lmk\beta_g^{R U}\rmk^{-1}\lmk \hat \eta_{gh}^R\rmk^{-1}\hat\alpha_{R}^{-1}\notag\\
&=
\pi_{R}\circ\alpha_R\circ
\lmk
\gamma_{[0,\theta_1],R}\otimes\gamma_{(\theta_1,\theta_2],R}
\rmk
\circ
\gamma_{(\theta_{0.8}, \theta_{1.2}],R}
\circ\gamma_{H,R}
\circ
\lmk \gamma_{[0,\theta_1],R}\circ\gamma_{(\theta_{0.8}, \theta_{1.2}],R}\circ\gamma_{H,R}\rmk^{-1}\notag\\
& \eta_g^R\beta_g^{RU}
 \eta_h^R\lmk\beta_g^{R U}\rmk^{-1}\lmk \eta_{gh}^R\rmk^{-1}\circ
\gamma_{[0,\theta_1],R}\circ\gamma_{(\theta_{0.8}, \theta_{1.2}],R}\circ\gamma_{H,R}
\circ
\lmk
\lmk
\gamma_{[0,\theta_1],R}\otimes\gamma_{(\theta_1,\theta_2],R}
\rmk
\circ
\gamma_{(\theta_{0.8}, \theta_{1.2}],R}
\circ\gamma_{H,R}
\rmk^{-1}\alpha_R^{-1}\notag\\
&=
\pi_{R}\circ\alpha_R\circ
\gamma_{(\theta_1,\theta_2],R}
\circ
 \eta_g^R\beta_g^{RU}
 \eta_h^R\lmk\beta_g^{R U}\rmk^{-1}\lmk \eta_{gh}^R\rmk^{-1}\circ
\lmk
\gamma_{(\theta_1,\theta_2],R}
\rmk^{-1}\circ\alpha_R^{-1}.\label{totori}
\end{align}
Because
$ \eta_g^R\beta_g^{RU}
 \eta_h^R\lmk\beta_g^{R U}\rmk^{-1}\lmk \eta_{gh}^R\rmk^{-1}\in  \Aut\lmk \caA_{{ C_{\theta_{0}}}}\rmk$ commutes with  $\gamma_{(\theta_1,\theta_2],R}$, we obtain
\begin{align}
\pi_R\circ\hat \alpha_R\circ\hat \eta_g^R\beta_g^{R U}
\hat \eta_h^R\lmk\beta_g^{R U}\rmk^{-1}\lmk \hat \eta_{gh}^R\rmk^{-1}\hat\alpha_{R}^{-1}
=
(\ref{totori})
=\pi_{R}\circ\alpha_R\circ
 \eta_g^R\beta_g^{RU}
 \eta_h^R\lmk\beta_g^{R U}\rmk^{-1}\lmk \eta_{gh}^R\rmk^{-1}
 \alpha_R^{-1}
 =\Ad\lmk u_{R}(g,h)\rmk\circ\pi_{R}.
\end{align}
An analogous statement for $\sigma=L$ also holds.
This completes the proof of  (\ref{onomichi}).
Hence the statement of the Theorem is proven.
\end{proof}
\section{Proof of Theorem \ref{main}}\label{mainthproofsec}
In this section, we prove Theorem \ref{main}.
%In fact, our index can be considered as an invariant of short range entangled $\beta$-invariant states.
%In this paper, we say a state has short range entanglement 
%if it can be created from a product state via a quasi-local automorphism.
The proof relies heavily on the machinery of quasi-local automorphisms 
developed in \cite{bmns}~\cite{NSY},~\cite{MO}.
(Summary is given in Appendix \ref{quasilocalsec}.)
We use terminology and facts from  Appendix \ref{ffunc}, \ref{quasilocalsec}, freely.
%\begin{defn}
We introduce a set of $F$-functions with fast decay, $\caF_a$ as 
Definition \ref{fadef}. 
Crucial point for us is the following.
\begin{thm}\label{mo}
Let $\Phi_0,\Phi_1\in\caP_{UG}$ and $\omega_{\Phi_0}$, $\omega_{\Phi_1}$
be their unique gapped ground states.
 Suppose that
$\Phi_0\sim\Phi_1$ holds, via a path $\Phi : [0,1]\to \caP_{UG}$.
Then there exists some 
$\Psi\in\hat\caB_F([0,1])$ with $\Psi_{1}\in \hat\caB_{F}([0,1])$ for some $F\in \caF_a$ of the form
$F(r)=\frac{\exp\lmk {-r^{\theta}}\rmk}{(1+r)^{4}}$ with a constant $0<\theta<1$,
such that
$\omega_{\Phi_{1}}=\omega_{\Phi_0}\circ\tau_{1,0}^{\Psi}$.
If 
$\Phi_0,\Phi_1\in \caP_{UG\beta}$ and 
$\Phi\sim_\beta\Phi_0$,
we may take $\Psi$ to be $\beta$-invariant.
\end{thm}
For the proof,
see Appendix \ref{quasilocalsec}.

%We set
%\begin{align}
%{ {QAut}_0}:=
%\left\{
%\tau_{1,0}^\Phi\mid
%\Phi\in \hat\caB_F([0,1]), \; F\in\caF_a
%\right\},
%\end{align}
%and denote by $QAut_1$ the set of all
%automorphism given by composition of finite number of
%
%We say that a state $\omega$ on $\caA_{\bbZ}$ is a
%short range entangled state if 
%there are an $F\in\caF$, a $\Phi\in\hat\caB_F([0,1])$ (see definition \ref{hbfdef}), and a product state
%$
% \omega_1:=\bigotimes_{\bbZ^2} \rho_\xi
% $   such that 
%$\omega=\omega_1\circ\tau_{1,0}^\Phi$.
%\end{defn}
%We require the existence of $\beta$-invariant pure 
%product state, which exists if and only if $U$ has an invariant vector.
%\begin{defn}
%Suppose that there is a $\beta$-invariant  pure product state 
% $
% \omega_0:=\bigotimes_{\bbZ^2} \rho_\xi
% $.
%We define
%\begin{align}
%\caS_{SL\beta}:=
%\left\{
%\omega\mid \omega=\omega_0\circ\tau_{1,0}^\Phi,\quad \text{for some}\quad
%\Phi\in \hat\caB_F([0,1]), \; F\in\caF_a
%\right\}
%\end{align}
%We say $\omega_1,\omega_2$ are equivalent (and write $\omega_1\sim_{SL\beta}\omega_2$) if
%there is $\Phi\in \hat\caB_F([0,1])$ and  $F\in\caF_a$
%such that 
%\begin{description}
%\item[(i)]
%$\omega_2=\omega_1\circ\tau_{1,0}^\Phi$, and
%\item[(ii)]
%$\Phi(\cdot,t)\in\caP_{UF\beta}$ for any $t\in[0,1]$.
%\end{description}
%
%\end{defn}
%\begin{rem}
%Because there is a strictly local automorphism connecting any two 
%product states, the definition of $\caS_{SL\beta}$ does not depend on the
%choice of $\omega_0$.
%\end{rem}
%
From this and Theorem \ref{defindexspt} and Theorem \ref{stabilitythm},
in order to show Theorem \ref{main}, it suffices to show
the following.

\begin{thm}
	\label{thm:quasiauto}
Let 
$F\in \caF_a$ be an $F$-function of the form
$F(r)=\frac{\exp\lmk {-r^{\theta}}\rmk}{(1+r)^{4}}$ with a constant  $0<\theta<1$.
Let $\Psi\in \hat \caB_{F}([0,1])$ be a path of interactions satisfying $\Psi_1\in \hat \caB_F([0,1])$.
Then we have $\tau_{1,0}^{\Psi}\in \sqaut(\caA)$.
Furthermore, if $\Psi$ is $\beta_{g}^{U}$-invariant, i.e., 
$\beta_{g}^U\lmk \Psi(X;t)\rmk=\Psi(X;t)$ for any $X\in{\mathfrak S}_{\bbZ^2}$, $t\in [0,1]$, and
$g\in G$, then 
we have $\tau_{1,0}^{\Psi}\in \gsqaut(\caA)$.
\end{thm}
\begin{proof}
Fix arbitrary 
\begin{align}\label{thetas}
0<\theta_{0.8}<\theta_1<\theta_{1.2}<\theta_{1.8}<\theta_2<\theta_{2.2}<
\theta_{2.8}<\theta_3<\theta_{3.2}<\frac\pi 2.
\end{align}
We show the existence of the decomposition 
\begin{align}
\begin{split}
\tau_{1,0}^\Psi=&\Ad(u)\circ
\lmk
 \alpha_{(0,\theta_1]}\otimes
\alpha_{(\theta_1,\theta_2]}\otimes \alpha_{(\theta_2,\theta_3]}
\otimes  \alpha_{(\theta_3,\frac\pi 2]}
\rmk\\
&\circ
\lmk
\alpha_{(\theta_{0.8},\theta_{1.2}]}
\alpha_{(\theta_{1.8},\theta_{2.2}]}\otimes \alpha_{(\theta_{2.8},\theta_{3.2}]}
\rmk,
\end{split}
\end{align}
with $\alpha$s of the form (\ref{sqaut2}) and (\ref{sqaut3}).
We follow the strategy in \cite{NO}.
\\
{\it Step 1.}
Fix some $0<\theta'<\theta$, and set
\begin{align}\label{apple}
\tilde F(r):=\frac{\exp\lmk {-r^{\theta'}}\rmk}{(1+r)^{4}}.
\end{align}
With suitably chosen constant $c_{1}>0$,
we have
\begin{align}\label{tildef}
\max\left\{ F\lmk\frac r 3\rmk, \lmk  F\lmk \lcm \frac r 3 \rcm \rmk\rmk^{\frac 12}\right\}\le
c_{1}\tilde F(r),\quad\quad  r\ge 0.
\end{align}
Namely, $c_{1}\tilde F(r)$ satisfy the condition of $\tilde F_{\theta}$
in Definition \ref{fadef} (ii) for our $F=\frac{\exp\lmk {-cr^{\theta}}\rmk}{(1+r)^{4}}$ and $\theta=\frac 12$.\\
Set 
\begin{align}
&\caC_{0}:=\left\{
\begin{gathered}
%C_{[0,\theta_{0}],\sigma}\;,
C_{[0,\theta_1],\sigma}\;,
C_{(\theta_1,\theta_2], \sigma,\zeta}\;,
C_{(\theta_2,\theta_3], \sigma,\zeta}\;,
C_{(\theta_3,\frac\pi 2], \zeta},\\
\sigma=L,R,\quad \zeta=D,U.
\end{gathered}
\right\},\\
&\caC_{1}:=\left\{
\begin{gathered}
%C_{(\theta_{-0.2},\theta_{0.2}),\sigma,\zeta}\;,
C_{(\theta_{0.8},\theta_{1.2}),\sigma,\zeta}\;,
C_{(\theta_{1.8},\theta_{2.2}), \sigma,\zeta}\;,
C_{(\theta_{2.8},\theta_{3.2}), \sigma,\zeta},\\
\sigma=L,R,\quad \zeta=D,U.
\end{gathered}
\right\}.
\end{align}
Define $\Psi^{(0)}, \Psi^{(1)}\in \hat\caB_{F}([0,1])$ by
 \begin{align}
 \begin{split}
 &\Psi^{(0)}\lmk X; t\rmk:=
 \left\{
\begin{gathered}
 \Psi\lmk X; t\rmk,\quad \text{if there exists a}\quad  C\in\caC_0\quad
 \text{such that }\quad X\subset C
 \\
 0,\quad \text{otherwise}
 \end{gathered}
 \right.,\\
& \Psi^{(1)}\lmk X; t\rmk:=\Psi^{(0)}\lmk X; t\rmk-\Psi\lmk X; t\rmk,
 \end{split}
 \end{align}
 for each $X\in{\mathfrak S}_{\bbZ^2}$, $t\in[0,1]$.

First we would like to represent $ \lmk \tau_{1,0}^{\Psi^{(0)}}\rmk^{-1}\circ \tau_{1,0}^{ \Psi}$
as some quasi-local automorphism.
 %Let $\{\Lambda_{n}\}_{n=1}^{\infty}\subset\caP_{0}\lmk {\bbZ^2}\rmk$ be an increasing sequence $\Lambda_{n}\nearrow{\bbZ^2}$. 
Let $t,s\in[0,1]$.
We apply Proposition \ref{tokyo1877} for $\Psi$ replaced by  $\Psi^{(1)}$,and $\tilde\Psi$
by $\Psi$.
%with $\caK_{t}$, $\caK_{t}^{(n)}$, $\Psi$ replaced by $\tau_{t,s}^{\Psi}$, $\tau_{t,s}^{(\Lambda_{n}), \Psi}$, $\Psi^{(1)}$.
%Note that the conditions {\it 1,2,3,4,5} of Theorem \ref{tsan} with $p=0,q=r=1$, $G=G_{F}$ is 
%clear from Theorem \ref{tni}.
%Note that we have $\Psi^{(1)}_{1}\in\caB_{F}([0,1])$ where $1=\max\{ p=0,q=1, r=1\}$.
Hence we set
\begin{align}\label{psisdef}
\Xi^{(s)}\lmk
Z, t
\rmk:=
\sum_{m\ge 0} \sum_{X\subset Z,\; X(m)=Z}
\Delta_{X(m)}\lmk
\tau_{t,s}^{\Psi}\lmk
 \Psi^{(1)}\lmk X; t\rmk
\rmk
\rmk
\end{align}
and
\begin{align}
\Xi^{(n)(s)}\lmk
Z, t
\rmk:=
\sum_{m\ge 0} \sum_{X\subset Z, X(m)\cap\Lambda_{n}=Z}
\Delta_{X(m)}\lmk
\tau_{t,s}^{(\Lambda_n)\Psi}\lmk
 \Psi^{{(1)}}\lmk X; t\rmk
\rmk
\rmk.
\end{align}

Corresponding to (\ref{psio}), we obtain
\begin{align}
\tau_{t,s}^{(\Lambda_n),\Psi} \lmk
H_{\Lambda_n,\Psi^{(1)}}(t)
\rmk
= H_{\Lambda_n,\Xi^{(n)(s)}}(t).
\end{align}
Applying Proposition \ref{tokyo1877}.
we have $\Xi^{(n)(s)}, \Xi^{(s)}\in \hat \caB_{\tilde F}([0,1])$, and
\begin{align}\label{convconv1}
\lim_{n\to\infty}\lV
\tau_{t,u}^{\Xi^{(n)(s)}}\lmk A\rmk
-\tau_{t,u}^{\Xi^{(s)}}\lmk A\rmk
\rV=0,\quad A\in\caA,\quad t,u\in [0,1]
\end{align}
holds.
%Note that
%\begin{align}
%\begin{split}
%\frac{d}{dt} \hat\tau_{t,s}^{(\Lambda_n), \Xi^{(n)(s)}}(A)
%&=-i\left[ H_{\Lambda_n, \Xi^{(n)(s)}}(t), \hat\tau_{t,s}^{(\Lambda_n), \Xi^{(n)(s)}}(A)
%\right] \\
%&=-i \left[
%\tau_{t,s}^{(\Lambda_n),\Psi} \lmk
%H_{\Lambda_n, \Psi^{(1)}}(t)
%\rmk,
%\hat\tau_{t,s}^{(\Lambda_n), \Xi^{(n)(s)}}(A)
%\right].
%\end{split}
%\end{align}
%On the other hand, we have
%\begin{align}
%\begin{split}
%\frac{d}{dt} \tau_{t,s}^{(\Lambda_n), \Psi}&\circ \lmk \tau_{t,s}^{(\Lambda_n),\Psi^{(0)}}\rmk^{-1}(A) \\
%&=\tau_{t,s}^{(\Lambda_n),\Psi}
%\lmk
%i\left[
%H_{\Lambda_n,\Psi}(t)-H_{\Lambda_n,\Psi^{(0)}}(t), \lmk \tau_{t,s}^{(\Lambda_n),\Psi^{(0)}}\rmk^{-1}(A)
%\right]
%\rmk \\
%&=-i\left[
%\tau_{t,s}^{(\Lambda_n),\Psi} \lmk
%H_{\Lambda_n, \Psi^{(1)}}(t)
%\rmk,
%\tau_{t,s}^{(\Lambda_n), \Psi}\circ \lmk \tau_{t,s}^{(\Lambda_n),\Psi^{(0)}}\rmk^{-1}(A)
%\right].
%\end{split}
%\end{align}
Two functions $\hat\tau_{t,s}^{(\Lambda_n), \Xi^{(n)(s)}}(A)$ and $ \tau_{t,s}^{(\Lambda_n), \Psi}\circ \lmk \tau_{t,s}^{(\Lambda_n),\Psi^{(0)}}\rmk^{-1}(A)$
satisfy the same differential equation and the initial condition.
Therefore we obtain
\begin{align}\label{ata}
\hat\tau_{t,s}^{(\Lambda_n), \Xi^{(n)(s)}}(A)= \tau_{t,s}^{(\Lambda_n), \Psi}\circ \lmk \tau_{t,s}^{(\Lambda_n),\Psi^{(0)}}\rmk^{-1}(A),\quad t\in [0,1],\quad A\in\caA.
\end{align}
From the fact that $
%\hat\tau_{t,u}^{\Xi^{(n)(s)}}\lmk A\rmk=
\hat \tau_{t,u}^{(\Lambda_n), \Xi^{(n)(s)}}=\tau_{u,t}^{(\Lambda_n), \Xi^{(n)(s)}}=\tau_{u,t}^{\Xi^{(n)(s)}}$ converges strongly to an automorphism $\tau_{u,t}^{\Xi^{(s)}}$ on 
$\caA$ (\ref{convconv1}), 
we have
\begin{align}\label{convconvh}
\lim_{n\to\infty}\lV
\hat\tau_{t,s}^{(\Lambda_{n})\Xi^{(n)(s)}}\lmk A\rmk
-\tau_{s,t}^{\Xi^{(s)}}\lmk A\rmk
\rV=0,\quad A\in\caA.
\end{align}
On the other hand, by Theorem~\ref{tni},  
we have for $t \in [0,1]$ and $A \in \caA$
\begin{align}
\lim_{n\to\infty}\lV
 \tau_{t,s}^{(\Lambda_n), \Psi}\circ \lmk \tau_{t,s}^{(\Lambda_n),\Psi^{(0)}}\rmk^{-1}(A)
 -\tau_{t,s}^{ \Psi}\circ \lmk \tau_{t,s}^{\Psi^{(0)}}\rmk^{-1}(A)
 \rV=0.
% ,\quad t\in [0,1],\quad A\in\caA.
\end{align}
Therefore, taking $n\to\infty$ limit in (\ref{ata}), we obtain
\begin{align}
\tau_{s,t}^{ \Xi^{(s)}}(A)= \tau_{t,s}^{\Psi}\circ \lmk \tau_{t,s}^{\Psi^{(0)}}\rmk^{-1}(A),\quad t,s\in [0,1],\quad A\in\caA.
\end{align}
Hence we have
\begin{align}
\tau_{s,t}^{\Psi}=\lmk \tau_{t,s}^{\Psi}\rmk^{-1}
=\lmk \tau_{t,s}^{\Psi^{(0)}}\rmk^{-1}\lmk\tau_{s,t}^{ \Xi^{(s)}}\rmk^{-1}=
\tau_{s,t}^{\Psi^{(0)}}\tau_{t,s}^{ \Xi^{(s)}}
\end{align}
In particular, we get
\begin{align}\label{ttt}
\tau_{1,0}^{\Psi}=
\tau_{1,0}^{\Psi^{(0)}}\tau_{0,1}^{ \Xi^{(1)}}.
\end{align}
{\it Step 2.}
We show 
\begin{align}\label{vtc}
\sum_{\substack{Z\in\pg, \\\not\exists C\in\caC_{1} s.t. Z\subset C}}
\sup_{t\in [0,1]}\lV \Xi^{(1)}\lmk Z,t\rmk\rV<\infty.
\end{align}
From this,
\begin{align}\label{vt}
V(t):=\sum_{\substack{Z\in\pg, \\\not\exists C\in\caC_{1} s.t. Z\subset C}}
\Xi^{(1)}\lmk Z,t\rmk
\in\caA
\end{align}
converges absolutely in the norm topology
and define an element in $\caA$.
Furthermore, for
\begin{align}
V_{n}(t):=\sum_{\substack{Z\in\pg,\; Z\subset \Lambda_{n} \\\not\exists C\in\caC_{1} s.t. Z\subset C} } 
\Xi^{(1)}\lmk Z,t\rmk
\in\caA_{\Lambda_{n}},\quad n\in\nan,
\end{align}
we get
\begin{align}\label{kub}
\lim_{n\to\infty}\sup_{t\in[0,1]}\lV V_{n}(t)-V(t)\rV=0,
\end{align}
from (\ref{vtc}).

To prove (\ref{vtc}), we first bound
\begin{align}
\begin{split}
&\sum_{\substack{Z\in\pg,
%\; Z\cap \Lambda_{n}^{c}\neq \emptyset
\\
\not\exists C\in\caC_{1} s.t. Z\subset C
} } 
\sup_{t\in[0,1]}
\lV
\Xi^{(1)}\lmk Z,t\rmk
\rV
\\
&\le
\sum_{\substack{Z\in\pg,
%\; Z\cap \Lambda_{n}^{c}\neq \emptyset
\\
\not\exists C\in\caC_{1} s.t. Z\subset C
} } 
\sum_{m\ge 0} \sum_{\substack{X: X\subset Z,\\ X(m)=Z}} 
\left[
\sup_{t\in[0,1]}\lV
\Delta_{X(m)}\lmk
\tau_{t,1}^{\Psi}\lmk
 \Psi^{(1)}\lmk X; t\rmk
\rmk
\rmk
\rV
\right]\\
&\le
\sum_{m\ge 0} \sum_{\substack{X:  
%X(m)\cap \Lambda_{n}^{c}\neq \emptyset
\\
\not\exists C\in\caC_{1} s.t. X(m)\subset C}} 
\sup_{t\in[0,1]}\lV
\Delta_{X(m)}\lmk
\tau_{t,1}^{\Psi}\lmk
 \Psi^{(1)}\lmk X; t\rmk
\rmk
\rmk
\rV\\
%&\le2
%\sum_{m\ge 0} \sum_{\substack{X:  X(m)\cap \Lambda_{n}^{c}\neq \emptyset\\
%\not\exists C\in\caC_{1} s.t. X(m)\subset C}} 
%\sup_{t\in[0,1]}\lV\Delta_{X(m)}\lmk
%\tau_{t,1}^{\Psi}\lmk
% \Psi^{(1)}\lmk X; t\rmk
%\rmk
%\rmk
%\rV\\
& \le
\sum_{m\ge 0} \sum_{\substack{X:  
%X(m)\cap \Lambda_{n}^{c}\neq \emptyset
\\
\not\exists C\in\caC_{1} s.t. X(m)\subset C}} 
 \left[
\sup_{t\in[0,1]}
\frac{8\lV \Psi^{(1)}\lmk X; t\rmk \rV}{C_{F}}\lmk e^{2I_{F}(\Psi)}-1\rmk\lv X\rv G_{F}\lmk m\rmk
\right]
\\
&=\frac{8}{C_{F}}\lmk e^{2I_F(\Psi)}-1\rmk
\sum_{m\ge 0} \sum_{\substack{X:  
%X(m)\cap \Lambda_{n}^{c}\neq \emptyset
\\
\not\exists C\in\caC_{1} s.t. X(m)\subset C}} 
\\
&\quad\quad\quad\quad\quad\quad\quad
\left[
\sup_{t\in[0,1]}\lmk \lV \Psi^{(1)}\lmk X; t\rmk \rV\rmk
\lv X\rv G_{F}\lmk m\rmk
\right].
\end{split}
\label{remi}
\end{align}
For the third inequality, we used Theorem~\ref{tni}~3.
For any cone $C_1,C_2$ of $\bbZ^2$ with apex at the origin, we set
\begin{align}\label{mc1c2}
M(C_1,C_2):=
\sum_{m\ge 0} 
\sum_{\substack{X:  
%X(m)\cap \Lambda_{n}^{c}\neq \emptyset
\\
\forall  C\in\caC_{1},
X\cap \lmk \lmk C^{c}\rmk(m)\rmk\neq\emptyset,\\
%\not\exists C\in\caC_{1} s.t. X(m)\subset C\\
 X\cap C_{1}\neq\emptyset,\quad  X\cap C_{2}\neq\emptyset
}} 
\left[
\sup_{t\in[0,1]}\lmk \lV \Psi^{(1)}\lmk X; t\rmk \rV\rmk
\lv X\rv G_{F}\lmk m\rmk
\right].
\end{align}
From the definition of $ \Psi^{(1)}$, we have
$
 \Psi^{(1)}\lmk X; t\rmk =0, 
$
unless
$X$ has a non-empty intersection with at least two of
elements in $\caC_{0}$.
Therefore, if $X$ gives a non-zero contribution in 
(\ref{remi}), then it has to satisfy
\begin{align*}
% X(m)\cap \Lambda_{n}^{c}\neq \emptyset,\\
X\cap \lmk \lmk C^{c}\rmk(m)\rmk\neq\emptyset,\quad \text{for all}\quad  C\in\caC_{1},\\
\exists C_{1},C_{2}\in\caC_{0}\quad\text{such that},\; C_1\neq C_2, \
\quad X\cap C_{1}\neq\emptyset,\quad  X\cap C_{2}\neq\emptyset.
\end{align*}
Hence we have
\begin{align}
\begin{split}
&(\ref{remi})\le
\frac{8}{C_{F}}\lmk e^{2I_F(\Psi)}-1\rmk
\sum_{\substack{C_{1},C_{2}\in\caC_{0}\\C_1\neq C_2
} }
M({C_1,C_2})
\end{split}
\end{align}
Hence it suffice so show that
$M({C_1,C_2})<\infty$ for all $C_{1},C_{2}\in\caC_{0}$
with $C_1\neq C_2$.

In order to proceed, we prepare two estimates.
We will freely identify $\bbC$ and $\bbR^{2}$ in an obvious manner.
In particular, $\arg z$ of $z\in\bbZ^{2}\subset\bbR^{2}$
in the following definition
is considered with this identification:
for $\varphi_{1}<\varphi_{2}$, we set 
\begin{align}\label{ccheck}
\check C_{[\varphi_{1},\varphi_{2}]}:=
\left\{
z\in \bbZ^{2}\mid \arg z\in [\varphi_{1},\varphi_{2}]
\right\}.
\end{align}
We define $\check C_{(\varphi_{1},\varphi_{2})}$ etc. analogously.
Set
\begin{align}\label{czdefn}
\begin{split}
&{c^{(0)}}_{\zeta_1,\zeta_2,\zeta_3,\zeta_4}
:=\sqrt{1-\max\left\{\cos (\zeta_3-\zeta_2),\cos (\zeta_4-\zeta_1),0\right\}},\quad
\zeta_1,\zeta_2,\zeta_3,\zeta_4 \in \bbR.
%&K_{F}^{(0)}(c)
%:=\sum_{r=0}^\infty
%F\lmk 
%c
%r
%\rmk (r+2)^3,\quad c>0,
%{K^{(1)}}(m,c):=&64\cdot 24\cdot (m+1)
%\sum_{r_1=0}^\infty\sum_{\substack{r\in{\bbZ_{\ge 0}}: \\\sqrt{r^2+r_1^2}c\ge (m+1)} }
%(r_1+1)
% F\lmk
%  \sqrt{r^2+r_1^2}c-(m+1)
%  \rmk\\
%&  +
%3\cdot 64\cdot  48\cdot \frac{(m+1)^4}{c^3}F(0),\\
%&\quad c>0,\quad m\in\bbZ_{\ge 0}.
\end{split}
\end{align}
%for $F\in\caF_{a}$.
%From  (\ref{bzeroest}), we have $K_{F}^{(0)}(c)<\infty$ for any $c>0$ and $F\in\caF_{a}$.
\begin{lem}\label{gomap}
Let $\varphi_1<\varphi_2<\varphi_3<\varphi_4$ with $\varphi_4-\varphi_1<2\pi$.
Then
\begin{align}
\begin{split}
&b_0(\varphi_1,\varphi_2,\varphi_3,\varphi_4)\\
&:=
\sum_{m\ge 0} \sum_{\substack{X: \\
X\cap \check
 C_{[\varphi_1,\varphi_2]}\neq\emptyset,\\ X\cap \check
 C_{[\varphi_3,\varphi_4]}\neq\emptyset
}} 
\left[
 \sup_{t\in[0,1]}\lmk\lV \Psi\lmk X; t\rmk\rV\rmk \lv X\rv G_{F}\lmk m\rmk
\right]
\notag\\
&\le
(64)^3
\frac {3^{4}\kappa_{1, 4, F}}{\lmk c^{(0)}_{\varphi_1,\varphi_2,\varphi_3,\varphi_4}\rmk ^{4}}
\lmk \lV\lv \Psi_1\rV\rv_{F}\rmk
\lmk \sum_{m\ge 0}G_{F}\lmk m\rmk\rmk<\infty.
%&
%\le
% (64)^3\lmk \lV\lv \Psi_1\rV\rv_F\rmk\lmk \sum_{m\ge 0}G_{F}\lmk m\rmk\rmk
%K_{F}^{(0)}(c^{(0)}_{\varphi_1,\varphi_2,\varphi_3,\varphi_4})<\infty.
%
% (64)^2\lmk \lV\lv \Psi_1^{(1)}\rV\rv\rmk\lmk \sum_{m\ge 0}G_{F}\lmk m\rmk\rmk\\
%&\sum_{r=0}^\infty
%F\lmk 
%\sqrt{1- \max\left\{
%\cos\lmk \varphi_3-\varphi_2\rmk,
%\cos\lmk \varphi_4-\varphi_1\rmk,0
%\right\}}
%r)
%\rmk (r+1)^3\\
\end{split}
\end{align}
\end{lem}
\begin{proof}
%Let $x=s_1e^{i\phi_1}\in \check
%{C}_{[\varphi_1,\varphi_2]}$ 
%and $y=s_2e^{i\phi_2}\in \check
%{C}_{[\varphi_3,\varphi_4]}$ with $s_1,s_2\ge 0$.
%If $\cos\lmk \phi_2-\phi_1\rmk\ge 0$, then
%we have
%\begin{align}
%\begin{split}
%&{\dist}(x,y)=\sqrt{s_1^2+s_2^2-2s_1s_2\cos\lmk \phi_2-\phi_1\rmk}
%\ge 
%\sqrt{s_1^2+s_2^2}
%\sqrt{1-\cos\lmk \phi_2-\phi_1\rmk}
%\\
%%&\ge \sqrt{s_1^2+s_2^2-2s_1s_2
%%\max\left\{
%%\cos\lmk \varphi_3-\varphi_2\rmk,
%%\cos\lmk \varphi_4-\varphi_1\rmk
%%\right\}}\notag\\
%&\ge \sqrt{1- \max\left\{
%\cos\lmk \varphi_3-\varphi_2\rmk,
%\cos\lmk \varphi_4-\varphi_1\rmk,0
%\right\}}
%\sqrt{s_1^2+s_2^2}
%.
%\end{split}
%\end{align}
%If $\cos\lmk \phi_2-\phi_1\rmk<0$, then
%we have
%\begin{align}
%\begin{split}
%&{\dist}(x,y)=\sqrt{s_1^2+s_2^2-2s_1s_2\cos\lmk \phi_2-\phi_1\rmk}
%\ge \sqrt{s_1^2+s_2^2}.\\
%%&\ge 
%% \sqrt{1- \max\left\{
%%\cos\lmk \varphi_3-\varphi_2\rmk,
%%\cos\lmk \varphi_4-\varphi_1\rmk
%%\right\}}
%%\sqrt{s_1^2+s_2^2}.
%%&\ge \sqrt{s_1^2+s_2^2-2s_1s_2
%%\max\left\{
%%\cos\lmk \varphi_3-\varphi_2\rmk,
%%\cos\lmk \varphi_4-\varphi_1\rmk
%%\right\}}\notag\\
%\end{split}
%\end{align}
%Hence for any $x=s_1e^{i\phi_1}\in \check
%{C}_{[\varphi_1,\varphi_2]}$
%and $y=s_2e^{i\phi_2}\in \check
%{C}_{[\varphi_3,\varphi_4]}$ with $s_1,s_2\ge 0$
%we have
%\begin{align}
%{\dist}(x,y)\ge  \sqrt{1- \max\left\{
%\cos\lmk \varphi_3-\varphi_2\rmk,
%\cos\lmk \varphi_4-\varphi_1\rmk,0
%\right\}}
%\sqrt{s_1^2+s_2^2}
%=c^{(0)}_{\varphi_1,\varphi_2,\varphi_3,\varphi_4}\sqrt{s_1^2+s_2^2}.
%\end{align}
Substituting Lemme \ref{gomap1}, we obtain
\begin{align}
&b_0(\varphi_1,\varphi_2,\varphi_3,\varphi_4)\\
&:=
\sum_{m\ge 0} \sum_{\substack{X: \\
X\cap \check
{C}_{[\varphi_1,\varphi_2]}\neq\emptyset,\\ X\cap \check
{C}_{[\varphi_3,\varphi_4]}\neq\emptyset
}} 
\left[
\sup_{t\in[0,1]}\lmk \lV \Psi\lmk X; t\rmk \rV\rmk
\lv X\rv G_{F}\lmk m\rmk
\right]
\notag\\
&\le 
\sum_{m\ge 0} \sum_{\substack{x\in  \check
{C}_{[\varphi_1,\varphi_2]},\\y\in \check
{C}_{[\varphi_3,\varphi_4]}}}
\sum_{X\ni x,y}\left[
\sup_{t\in[0,1]}\lmk \lV \Psi\lmk X; t\rmk \rV\rmk
\lv X\rv G_{F}\lmk m\rmk
\right]\notag\\
&\le \lmk \lV\lv \Psi_1\rV\rv_{F}\rmk
 \sum_{\substack{x\in \check
{C}_{[\varphi_1,\varphi_2]},\\y\in \check
{C}_{[\varphi_3,\varphi_4]}}}
F\lmk {\dist}(x,y)\rmk
\lmk \sum_{m\ge 0}G_{F}\lmk m\rmk\rmk\notag\\
&\le
(64)^3
\frac {3^{4}\kappa_{1, 4, F}}{\lmk c^{(0)}_{\varphi_1,\varphi_2,\varphi_3,\varphi_4}\rmk ^{4}}
\lmk \lV\lv \Psi_1\rV\rv_{F}\rmk
\lmk \sum_{m\ge 0}G_{F}\lmk m\rmk\rmk<\infty.
%&\le \lmk \lV\lv \Psi_1^{(1)}\rV\rv\rmk\lmk \sum_{m\ge 0}G_{F}\lmk m\rmk\rmk
%\sum_{r_1=0}^\infty\sum_{r_2=0}^\infty
%\sum_{\substack{x\in S_{r_1,1}^{[\varphi_1,\varphi_2]}\cap \bbZ^2\\
%y\in S_{r_2,1}^{[\varphi_3,\varphi_4]}\cap \bbZ^2}}
%F\lmk {\dist}(x,y)\rmk
%\notag\\
%\begin{split}
%&\le\lmk \lV\lv \Psi_1^{(1)}\rV\rv\rmk\lmk \sum_{m\ge 0}G_{F}\lmk m\rmk\rmk\\
%&\sum_{r_1=0}^\infty\sum_{r_2=0}^\infty
%F\lmk 
%{c^{(0)}_{\varphi_1,\varphi_2,\varphi_3,\varphi_4}}
%\sqrt{r_1^2+r_2^2}.)
%\rmk\#\lmk
%S_{r_1,1}^{[\varphi_1,\varphi_2]}\cap \bbZ^2
%\rmk
%\#\lmk
%S_{r_2,1}^{[\varphi_3,\varphi_4]}\cap \bbZ^2
%\rmk
%\end{split}\\
%\begin{split}
%&\le (64)^2\lmk \lV\lv \Psi_1^{(1)}\rV\rv\rmk\lmk \sum_{m\ge 0}G_{F}\lmk m\rmk\rmk\\
%&\sum_{r_1=0}^\infty\sum_{r_2=0}^\infty
%F\lmk 
%{c^{(0)}_{\varphi_1,\varphi_2,\varphi_3,\varphi_4}}
%\sqrt{r_1^2+r_2^2})
%\rmk(r_1+1)(r_2+1)\\
%& 
%\le (64)^2\lmk \lV\lv \Psi_1^{(1)}\rV\rv\rmk\lmk \sum_{m\ge 0}G_{F}\lmk m\rmk\rmk\\
%&\sum_{r=0}^\infty
%F\lmk 
%{c^{(0)}_{\varphi_1,\varphi_2,\varphi_3,\varphi_4}}
%r)
%\rmk (r+2)^2
%\#\lmk
%S_{r}^{0,\frac\pi 2}\cap \bbZ^2
%\rmk
%\\
%&\le
%(64)^3\lmk \lV\lv \Psi_1^{(1)}\rV\rv\rmk\lmk \sum_{m\ge 0}G_{F}\lmk m\rmk\rmk\\
%&\sum_{r=0}^\infty
%F\lmk 
%{c^{(0)}_{\varphi_1,\varphi_2,\varphi_3,\varphi_4}}
%r)
%\rmk (r+2)^3
%\\
%&= (64)^3\lmk \lV\lv \Psi_1^{(1)}\rV\rv\rmk\lmk \sum_{m\ge 0}G_{F}\lmk m\rmk\rmk
%K^{(0)}(c^{(0)}_{\varphi_1,\varphi_2,\varphi_3,\varphi_4}).
%\end{split}
 \end{align}
 We used Lemma \ref{gomap} at the last inequality.
 The last value is finite by (\ref{ugugug}) for our $F\in\caF_{a}$.
\end{proof}
Set\begin{align}
{c^{(1)}}_{\zeta_1,\zeta_2,\zeta_3}
:=\sqrt{1-\max\left\{\cos (\zeta_1-\zeta_2),\cos (\zeta_1-\zeta_3)\right\}},\quad
\zeta_1,\zeta_2,\zeta_3\in [0,2\pi).
\end{align}
\begin{lem}\label{yomap}
For $\varphi_1<\varphi_2<\varphi_3$ with $\varphi_3-\varphi_1<\frac \pi 2$, we have
\begin{align}
\begin{split}
&b_1(\varphi_1,\varphi_2,\varphi_3)
:=\sum_{m\ge 0} \sum_{\substack{X: \\
X\subset \check
{C}_{[\varphi_1,\varphi_3]}\\
X\cap \check
{C}_{[\varphi_1,\varphi_2]}\neq\emptyset\\
X\cap \check
 C_{[\varphi_2,\varphi_3]}\neq\emptyset\\
X\cap \lmk \lmk
\lmk \check
 C_{(\varphi_1,\varphi_3)}\rmk^c\rmk(m)\rmk\neq \emptyset\\
}} 
\left[
\sup_{t\in[0,1]}\lmk \lV \Psi\lmk X; t\rmk \rV\rmk
\lv X\rv G_{F}\lmk m\rmk
\right]\\
%&\le
%\lmk \lV\lv \Psi_1^{(1)}\rV\rv\rmk
%\sum_{m\ge 0}  G_{F}\lmk m\rmk
%\lmk
%{K^{(1)}}(m,{c^{(1)}}_{\varphi_1, \varphi_2,\varphi_3})
%+{K^{(1)}}(m,{c^{(1)}}_{\varphi_3, \varphi_1,\varphi_2})
%\rmk
&\le64\cdot 144\cdot 24\cdot 
\lmk
\pi \kappa_{1,2,F}+F(0)
\rmk
\lmk \lV\lv \Psi_1\rV\rv_F\rmk
\lmk
\sum_{m\ge 0}  (m+1)^{4}G_{F}\lmk m\rmk\rmk
\lmk
\lmk {c^{(1)}}_{\varphi_{1},\varphi_{2},\varphi_{3}}\rmk^{-4}
+\lmk {c^{(1)}}_{\varphi_{3},\varphi_{1},\varphi_2}\rmk^{-4}
\rmk
<\infty.
\end{split}
\end{align}
\end{lem}
\begin{proof}
Set
\begin{align}
L_\varphi:=
\left\{
z\in\bbR^2\mid
\arg z=\varphi
\right\},\quad \varphi\in [0,2\pi).
\end{align}
Note that
if $X\in{\mathfrak S}_{\bbZ^2}$ satisfies
$X\subset \check
{C}_{[\varphi_1,\varphi_3]}$ and 
$X\cap \lmk \lmk
\lmk \check
 C_{(\varphi_1,\varphi_3)}\rmk^c\rmk(m)\rmk\neq \emptyset$,
then
we have
\begin{align}
d(X, L_{\varphi_1})\le m,\quad\text{or}\quad
d(X, L_{\varphi_3})\le m.
\end{align}
%
%X\subset \hat{C}_{[\varphi_1,\varphi_3]}\\
%X\cap \hat{C}_{[\varphi_1,\varphi_2]}\neq\emptyset\\
%X\cap C_{[\varphi_2,\varphi_3]}\neq\emptyset\\
%X\cap \lmk
%C_{(\varphi_1,\varphi_3)}^c(m)\rmk\neq \emptyset\\
%
Therefore, we have
\begin{align}
\begin{split}
&\sum_{m\ge 0} \sum_{\substack{X: \\
X\subset \check
{C}_{[\varphi_1,\varphi_3]}\\
X\cap \check
{C}_{[\varphi_1,\varphi_2]}\neq\emptyset\\
X\cap \check
 C_{[\varphi_2,\varphi_3]}\neq\emptyset\\
X\cap \lmk \lmk
\lmk \check
 C_{(\varphi_1,\varphi_3)}\rmk^c\rmk(m)\rmk\neq \emptyset\\
}} 
\left[
\sup_{t\in[0,1]}\lmk \lV \Psi\lmk X; t\rmk \rV\rmk
\lv X\rv G_{F}\lmk m\rmk
\right]\\
&\le\sum_{m\ge 0}  G_{F}\lmk m\rmk
\lmk
 \sum_{\substack{X: \\
%X\subset \hat{C}_{[\varphi_1,\varphi_3]}\\
%X\cap \hat{C}_{[\varphi_1,\varphi_2]}\neq\emptyset\\
X\cap \check C_{[\varphi_2,\varphi_3]}\neq\emptyset\\
%X\cap \lmk \lmk\lmk C_{(\varphi_1,\varphi_3)}\rmk^c\rmk(m)\rmk\neq \emptyset\\
d(X, L_{\varphi_1})\le m
}} 
+\sum_{\substack{X: \\
%X\subset \hat{C}_{[\varphi_1,\varphi_3]}\\
X\cap \check{C}_{[\varphi_1,\varphi_2]}\neq\emptyset\\
%X\cap C_{[\varphi_2,\varphi_3]}\neq\emptyset\\
%X\cap \lmk \lmk\lmk C_{(\varphi_1,\varphi_3)}\rmk^c\rmk(m)\rmk\neq \emptyset\\
d(X, L_{\varphi_3})\le m
}} \rmk
\left[
\sup_{t\in[0,1]}\lmk \lV \Psi\lmk X; t\rmk \rV\rmk
\lv X\rv
\right]\\
&\le
\sum_{m\ge 0}  G_{F}\lmk m\rmk
\lmk
\sum_{\substack{x\in \check
 C_{[\varphi_2,\varphi_3]}\\
y\in  L_{\varphi_1}(m)
}}
+\sum_{\substack{x\in \check
{C}_{[\varphi_1,\varphi_2]}\\
y\in  L_{\varphi_3}(m)
}}
 \rmk
 \sum_{\substack{X: X\ni x,y
%X\subset \hat{C}_{[\varphi_1,\varphi_3]}\\
%X\cap \hat{C}_{[\varphi_1,\varphi_2]}\neq\emptyset\\
%X\cap C_{[\varphi_2,\varphi_3]}\neq\emptyset\\
%X\cap \lmk \lmk\lmk C_{(\varphi_1,\varphi_3)}\rmk^c\rmk(m)\rmk\neq \emptyset\\
%d(X, L_{\varphi_1})\le m
}} 
\left[
\sup_{t\in[0,1]}\lmk \lV \Psi\lmk X; t\rmk \rV\rmk
\lv X\rv
\right]\\
&\le
\lmk \lV\lv \Psi_1\rV\rv_F\rmk
\sum_{m\ge 0}  G_{F}\lmk m\rmk
\lmk
\sum_{\substack{x\in \check
 C_{[\varphi_2,\varphi_3]}\\
y\in  L_{\varphi_1}(m)
}}
+\sum_{\substack{x\in \check
{C}_{[\varphi_1,\varphi_2]}\\
y\in  L_{\varphi_3}(m)
}}
 \rmk
F\lmk {\dist}(x,y)\rmk\\
%&\le
%\lmk \lV\lv \Psi_1^{(1)}\rV\rv\rmk
%\sum_{m\ge 0}  G_{F}\lmk m\rmk\\
%&\lmk
%\sum_{r_1=0}^\infty
%\sum_{x\in S_{r_1}^{\varphi_2,\varphi_3}}\sum_{\substack{
%y\in  L_{\varphi_1}(m)
%}}F\lmk {\dist}(x,y)\rmk+
%\sum_{r_1=0}^\infty
%\sum_{x\in S_{r_1,1}^{[\varphi_1,\varphi_2]}}\sum_{\substack{
%y\in  L_{\varphi_3}(m)
%}}F\lmk {\dist}(x,y)\rmk
%\rmk \\
%&\le
%\lmk \lV\lv \Psi_1^{(1)}\rV\rv\rmk
%\sum_{m\ge 0}  G_{F}\lmk m\rmk
%\lmk
%\sum_{\substack{x\in C_{[\varphi_2,\varphi_3]}\\
%y\in  L_{\varphi_1}(m)
%}}
%+\sum_{\substack{x\in \hat{C}_{[\varphi_1,\varphi_2]}\\
%y\in  L_{\varphi_3}(m)
%}}
% \rmk
%F\lmk {\dist}(x,y)\rmk\\
&\le
64\cdot 144\cdot 24\cdot 
\lmk
\pi \kappa_{1,2,F}+F(0)
\rmk
\lmk \lV\lv \Psi_1\rV\rv_F\rmk
\lmk
\sum_{m\ge 0}  (m+1)^{4}G_{F}\lmk m\rmk\rmk
\lmk
\lmk {c^{(1)}}_{\varphi_{1},\varphi_{2},\varphi_{3}}\rmk^{-4}
+\lmk {c^{(1)}}_{\varphi_{3},\varphi_{1},\varphi_2}\rmk^{-4}
\rmk
\\
%\lmk
%{K^{(1)}}(m,{c^{(1)}}_{\varphi_1, \varphi_2,\varphi_3})
%+{K^{(1)}}(m,{c^{(1)}}_{\varphi_3, \varphi_1,\varphi_2})
%\rmk
\\
\end{split}
\end{align}
At the last inequality, we used Lemma \ref{lcone} with
$\varphi_3-\varphi_1<\frac\pi 2$.
Because of $\varphi_3-\varphi_1<\frac\pi 2$ and (\ref{ugugug}), the last value is finite.
\end{proof}
Now let us go back to the estimate of (\ref{mc1c2}).
If $C_1,C_2\in \caC_0$ are
$C_1=\check
{C}_{[\varphi_1,\varphi_2]}$,
$C_2=\check
{C}_{[\varphi_3,\varphi_4]}$
with  $\varphi_1<\varphi_2<\varphi_3<\varphi_4$, $\varphi_4-\varphi_1<2\pi$,
then from Lemma \ref{gomap},
we have
\begin{align}
\begin{split}
&
M(C_1,C_2)
%\sum_{m\ge 0} 
%\sum_{\substack{X:  
%%X(m)\cap \Lambda_{n}^{c}\neq \emptyset
%\\
%\forall  C\in\caC_{1},
%X\cap \lmk \lmk C^{c}\rmk(m)\rmk\neq\emptyset,\\
%%\not\exists C\in\caC_{1} s.t. X(m)\subset C\\
% X\cap C_{1}\neq\emptyset,\quad  X\cap C_{2}\neq\emptyset
%}} 
%\left[
%\sup_{t\in[0,1]}\lmk \lV \Psi^{(1)}\lmk X; t\rmk \rV\rmk
%\lv X\rv G_{F}\lmk m\rmk
%\right]\\
%
%
%M(C_1,C_2):=
%\sum_{m\ge 0} 
%\sum_{\substack{X:  X(m)\cap \Lambda_{n}^{c}\neq \emptyset\\
%\not\exists C\in\caC_{1} s.t. X(m)\subset C\\
% X\cap C_{1}\neq\emptyset,\quad  X\cap C_{2}\neq\emptyset
%}} 
%\left[
%\sup_{t\in[0,1]}\lmk \lV \Psi^{(1)}\lmk X; t\rmk \rV\rmk
%\lv X\rv G_{F}\lmk m\rmk
%\right]\\
\le
b_0(\varphi_1,\varphi_2,\varphi_3,\varphi_4)
%&:=\sum_{m\ge 0} \sum_{\substack{X: \\
%X\cap \hat{C}_{[\varphi_1,\varphi_2]}\neq\emptyset,\\ X\cap \hat{C}_{[\varphi_3,\varphi_4]}\neq\emptyset
%}} 
%\left[
%\sup_{t\in[0,1]}\lmk \lV \Psi^{(1)}\lmk X; t\rmk \rV\rmk
%\lv X\rv G_{F}\lmk m\rmk
%\right]\\
%&\le 
%(64)^2\lmk \lV\lv \Psi_1^{(1)}\rV\rv\rmk\lmk \sum_{m\ge 0}G_{F}\lmk m\rmk\rmk
%K^{(0)}(c^{(0)}_{\varphi_1,\varphi_2,\varphi_3,\
%varphi_4})
<\infty.
\end{split}
\end{align}
Now suppose that $C_1,C_2\in \caC_0$ are
$C_1=\check
{C}_{[\varphi_1,\varphi_2]}$,
$C_2=\check
 C_{[\varphi_2,\varphi_3]}$
with  $\varphi_1<\varphi_2<\varphi_3$, $\varphi_3-\varphi_1<2\pi$.
By the definition of $\caC_0$ and $\caC_1$,
there is some $C=C_{(\zeta_1,\zeta_2)}\in\caC_1$
such that $\varphi_1<\zeta_1<\varphi_2<\zeta_2<\varphi_3$ and $\zeta_2-\zeta_1<\frac\pi 2$.
For $X\in{\mathfrak S}_{\bbZ^2}$ 
to give a nonzero contribution in (\ref{mc1c2}),
it have to satisfy
\begin{align}
X(m)\cap \lmk \check
 C_{[\zeta_1,\zeta_2]}\rmk^c\neq\emptyset ,\quad
X\cap \check
{C}_{[\varphi_1,\varphi_2]}\neq\emptyset,\quad  X\cap \check C_{[\varphi_2,\varphi_3]}\neq\emptyset.
\end{align}
For such an $X$, one of the following occurs:
\begin{description}
\item[(i)] $X\cap \check C_{[\zeta_2,\varphi_3]}\neq\emptyset$ and $X\cap \check {C}_{[\varphi_1,\varphi_2]}\neq\emptyset$
\item[(ii)] $X\cap \check C_{[\varphi_1,\zeta_1]}\neq\emptyset$ and
$X\cap \check C_{[\varphi_2,\varphi_3]}\neq\emptyset$
\item[(iii)] $X\cap \check C_{[\varphi_2, \zeta_2]}\neq\emptyset$ 
(and 
$X\cap\check  C_{[\zeta_1,\varphi_2]}\neq\emptyset$ )
and
$X\cap \check C_{[\varphi_3,\varphi_1+2\pi]}\neq\emptyset$,
%\item[(iv)]($X\cap C_{[\varphi_2, \zeta_2]}\neq\emptyset$ and )
%$X\cap C_{\zeta_1,\varphi_2}\neq\emptyset$ and
%$X\cap C_{[\zeta_2,\varphi_3]}\neq\emptyset$
%\item[(v)]
%$X\cap C_{[\varphi_2, \zeta_2]}\neq\emptyset$ 
%(and 
%$X\cap C_{\zeta_1,\varphi_2}\neq\emptyset$ )
%and 
%$X\cap C_{[\varphi_1, \zeta_1]}\neq\emptyset$ 
\item[(iv)]
$X\subset \check C_{\zeta_1,\zeta_2}$,
$X\cap \lmk \lmk \check C_{\zeta_1,\zeta_2}\rmk^c\rmk(m)\neq\emptyset$,
$X\cap \check C_{[\varphi_2, \zeta_2]}\neq\emptyset$ ,
and 
$X\cap \check C_{[\zeta_1,\varphi_2]}\neq\emptyset$.
\end{description}
Hence we get 
\begin{align}
\begin{split}
&M({C_1,C_2})\\
&\le 
b_0(\varphi_1,\varphi_2, \zeta_2,\varphi_3)
+b_0(\varphi_1,\zeta_1, \varphi_2,\varphi_3)
+b_0(\varphi_2, \zeta_2, \varphi_3,\varphi_1+2\pi)\\
&+b_1(\zeta_1,\varphi_2,\zeta_2)\\
&<\infty.
\end{split}
\end{align}
Hence we have proven the claim of {\it Step 2}.\\
{\it Step 3.}
Next we set
\begin{align}\label{tpsidef}
\tilde\Xi(Z,t):=\left\{
\begin{gathered}
\Xi^{(1)}(Z,t),\quad\text{if}\;\quad \;\exists C\in\caC_{1} \;s.t. \; Z\subset C\\
0\quad\text{otherwise}
\end{gathered}.\right.
\end{align}

 Clearly, we have $\tilde \Xi\in \hat\caB_{\tilde F}([0,1])$.
 Note that
 \begin{align}\label{hvh}
 H_{\Lambda_{n}, \tilde\Xi}(t)+V_{n}(t)
 = H_{\Lambda_{n}, \Xi^{(1)}}(t).
 \end{align}
 
 As a uniform limit of $[0,1]\ni t\mapsto V_{n}(t)\in \caA$, (\ref{kub}),
$[0,1]\ni t\mapsto V(t)\in \caA$ is norm-continuous.
Because of $\tilde \Xi\in \hat\caB_{\tilde F}([0,1])$,
$[0,1]\ni t\mapsto \tau_{t,s}^{\tilde\Xi}\lmk V(t)\rmk\in \caA$ is also norm-continuous,
for each $s\in[0,1]$.
Therefore, for each $s\in [0,1]$, there is a unique norm-differentiable map 
$[0,1]\ni t \mapsto W^{(s)}(t) \in \caU\lmk \caA\rmk$
such that
\begin{align}
\frac{d}{dt} W^{(s)}(t)=-i \tau_{t,s}^{\tilde\Xi}\lmk V(t)\rmk W^{(s)}(t),\quad
W^{(s)}(s)=\unit.
\end{align}
It is given by
\begin{align}\label{wsexp}
W^{(s)}(t)
:=\sum_{k=0}^{\infty }(-i)^{k}
\int_{s}^{t}ds_{1}\int_{s}^{s_{1}}ds_{2}\cdots \int_{s}^{s_{k-1}}ds_{k}
\tau_{s_{1},s}^{\tilde\Xi}\lmk V(s_{1})\rmk
\cdots \tau_{s_{k},s}^{\tilde\Xi}\lmk V(s_{k})\rmk.
\end{align}
Analogously, for each $s\in[0,1]$ and $n\in\nan$, we define
a unique norm-differentiable map from
$[0,1]$ to $ \caU\lmk \caA\rmk$
such that
\begin{align}
\frac{d}{dt} W_{n}^{(s)}(t)=-i \tau_{t,s}^{(\Lambda_{n})\tilde\Xi}\lmk V_{n}(t)\rmk W_{n}^{(s)}(t),\quad
W_{n}^{(s)}(s)=\unit.
\end{align}
It is given by
\begin{align}\label{wsexpn}
W_{n}^{(s)}(t)
:=\sum_{k=0}^{\infty }(-i)^{k}
\int_{s}^{t}ds_{1}\int_{s}^{s_{1}}ds_{2}\cdots \int_{s}^{s_{k-1}}ds_{k}
\tau_{s_{1},s}^{(\Lambda_{n})\tilde\Xi}\lmk V_{n}(s_{1})\rmk
\cdots \tau_{s_{k},s}^{(\Lambda_{n})\tilde\Xi}\lmk V_{n}(s_{k})\rmk.
\end{align}
By the uniform convergence (\ref{kub}) and Lemma \ref{tni},
we have 
\begin{align}\lim_{n\to\infty}
\sup_{t\in[0,1]}\lV
\tau_{t,s}^{(\Lambda_{n})\tilde\Xi}\lmk V_{n}(t)\rmk
-\tau_{t,s}^{\tilde\Xi}\lmk V(t)\rmk\rV=0.
\end{align}
From this and (\ref{wsexp}), (\ref{wsexpn}),
we obtain
\begin{align}
\lim_{n\to\infty}\sup_{t\in[0,1]}\lV W_{n}^{(s)}(t)- W^{(s)}(t) \rV =0.
\end{align}
This and Theorem \ref{tni} 4 for $\Xi^{(1)}, \tilde \Xi\in \caB_{\tilde F}([0,1])$
imply
\begin{align}\label{limlim}
\begin{split}
\lim_{n\to\infty}  \tau_{s,t}^{(\Lambda_{n}), \tilde\Xi}\circ \Ad \lmk W_{n}^{(s)}(t)\rmk (A)=
\tau_{s,t}^{ \tilde\Xi}\circ \Ad \lmk W^{(s)}(t)\rmk (A),\\
\lim_{n\to\infty}  \tau_{s,t}^{(\Lambda_{n}), \Xi^{(1)}}(A)=
\tau_{s,t}^{ \Xi^{(1)}}(A),
\end{split}
\end{align}
for any $A\in\caA$.

Note that for any $A\in\caA$
\begin{align}
\begin{split}
&\frac{d}{dt} \tau_{s,t}^{(\Lambda_{n}), \tilde\Xi}\circ \Ad \lmk W_{n}^{(s)}(t)\rmk (A)\\
&=-i\left[
H_{\Lambda_{n}, \tilde\Xi}(t),
\tau_{s,t}^{(\Lambda_{n}), \tilde\Xi}\circ \Ad \lmk W_{n}^{(s)}(t)\rmk (A)
\right] \\
&\quad\quad\quad\quad\quad\quad\quad
-i\tau_{s,t}^{(\Lambda_{n}), \tilde\Xi}\lmk
\left[
\tau_{t,s}^{(\Lambda_{n}), \tilde\Xi}\lmk V_{n}(t)\rmk,
\Ad \lmk W_{n}^{(s)}(t)\rmk (A)
\right]
\rmk
\\
&=-i\left[
H_{\Lambda_{n}, \tilde\Xi}(t)+V_{n}(t),
\tau_{s,t}^{(\Lambda_{n}), \tilde\Xi}\circ \Ad \lmk W_{n}^{(s)}(t)\rmk (A)
\right]\nonumber\\
&=-i\left[
H_{\Lambda_{n}, \Xi^{(1)}}(t),
\tau_{s,t}^{(\Lambda_{n}), \tilde\Xi}\circ \Ad \lmk W_{n}^{(s)}(t)\rmk (A)
\right].
\end{split}
\end{align}
We used  (\ref{inv})
for the second equality and (\ref{hvh}) for the third equality.
On the other hand, for any $A\in\caA$, we have 
\begin{align}
\frac{d}{dt} \tau_{s,t}^{(\Lambda_{n}), \Xi^{(1)}} (A)=-i\left[
H_{\Lambda_{n}, \Xi^{(1)}}(t),
\tau_{s,t}^{(\Lambda_{n}), \Xi^{(1)}} (A)
\right].
\end{align}
Therefore, $\tau_{s,t}^{(\Lambda_{n}), \tilde\Xi}\circ \Ad \lmk W_{n}^{(s)}(t)\rmk (A)$
and $ \tau_{s,t}^{(\Lambda_{n}), \Xi^{(1)}} (A)$ satisfy the same differential equation.
Also note that we have $\tau_{s,s}^{(\Lambda_{n}), \tilde\Xi}\circ \Ad \lmk W_{n}^{(s)}(s)\rmk (A)=
 \tau_{s,s}^{(\Lambda_{n}), \Xi^{(1)}} (A)=A$.
 Therefore, we get
 \begin{align}
 \tau_{s,t}^{(\Lambda_{n}), \tilde\Xi}\circ \Ad \lmk W_{n}^{(s)}(t)\rmk (A)
 =\tau_{s,t}^{(\Lambda_{n}), \Xi^{(1)}} (A).
 \end{align}
By (\ref{limlim}), we obtain
\begin{align}
 \tau_{s,t}^{\tilde\Xi}\circ \Ad \lmk W^{(s)}(t)\rmk (A)
 =\tau_{s,t}^{\Xi^{(1)}} (A),\quad A\in\caA, \; t,s\in[0,1].
 \end{align}
 Taking inverse, we get 
\begin{align}\label{www}
 \Ad \lmk W^{(s)^{*}}(t)\rmk\circ\tau_{t,s }^{\tilde\Xi}  
 =\tau_{t,s}^{\Xi^{(1)}},\; t,s\in[0,1].
 \end{align}
 {\it Step 4.}
 Combining (\ref{ttt}) and (\ref{www}) we have
 \begin{align}
\tau_{1,0}^{\Psi}=
\tau_{1,0}^{\Psi^{(0)}}\tau_{0,1}^{ \Xi^{(1)}}
=\tau_{1,0}^{\Psi^{(0)}}\circ\Ad \lmk \lmk W^{(1)}(0)\rmk^{*}\rmk\circ\tau_{0,1 }^{\tilde\Xi}.
\end{align}
By the definition of $\Psi^{(0)}$ and $\tilde\Xi$,
we obtain
decompositions
\begin{align}\label{inu}
\begin{split}
&\tau_{1,0}^{\Psi^{(0)}}
=\alpha_{[0,\theta_{1}]}\otimes \alpha_{(\theta_1,\theta_2]}\otimes \alpha_{(\theta_2,\theta_3]}
\otimes  \alpha_{(\theta_3,\frac\pi 2]}\\
&\tau_{0,1}^{\tilde\Xi}
= \alpha_{(\theta_{0.8},\theta_{1.2}]}\otimes
\alpha_{(\theta_{1.8},\theta_{2.2}]}\otimes \alpha_{(\theta_{2.8},\theta_{3.2}]}
\end{split}
\end{align}
with $\alpha$s of the form (\ref{sqaut2}) and (\ref{sqaut3}).
This completes the proof of the first part.
\\
{\it Step 5.}
Suppose that $\beta_{g}^{U}\lmk \Psi(X;t)\rmk=\Psi(X;t)$ for any $X\in{\mathfrak S}_{\bbZ^2}$, $t\in [0,1]$, and
$g\in G$. Then clearly we have $\beta_{g}^{U}\lmk \Psi^{(0)}(X;t)\rmk=\Psi^{(0)}(X;t)$ for any $X\in{\mathfrak S}_{\bbZ^2}$, $t\in [0,1]$, and
$g\in G$.
By Theorem \ref{tni}, 5, this implies
$\tau_{1,0}^{\Psi^{(0)}}{\beta_g^U}={\beta_g^U}\tau_{1,0}^{\Psi^{(0)}}$.
From the decomposition (\ref{inu}),
this means all of $\alpha_{[0,\theta_{1}],\sigma}$, 
$\alpha_{(\theta_1,\theta_2],\sigma,\zeta}$, $\alpha_{(\theta_2,\theta_3],\sigma,\zeta}$,
$\alpha_{(\theta_3,\frac\pi 2],\zeta}$, $\sigma=L,R$, $\zeta=U,D$ 
commute with ${\beta_g^U}$.
Because $\Pi_{X}$ commutes with $\beta_{g}^{U}$
$\tau_{t,s}^{\Psi}$ commutes with $\beta_{g}^{U}$  (Theorem \ref{tni}, 5)
, and $\Psi^{(1)}$ is $\beta_{g}^{U}$-invariant,
$\Xi^{(s)}$ is $\beta_{g}^{U}$-invariant 
from the definition (\ref{psisdef}).
Therefore, from the definition (\ref{tpsidef}), $\tilde\Xi$ is also
$\beta_{g}^{U}$-invariant.
Hence by Theorem \ref{tni}, 5
$\tau_{0,1}^{ \tilde \Xi}$ commutes with ${\beta_g^U}$.
The decomposition
(\ref{inu}) then implies that
$\alpha_{(\theta_{0.8},\theta_{1.2}],\sigma,\zeta}$
$\alpha_{(\theta_{1.8},\theta_{2.2}],\sigma,\zeta}$
 $\alpha_{(\theta_{2.8},\theta_{3.2}],\sigma,\zeta}
$,$\sigma=L,R$, $\zeta=U,D$ 
commute with ${\beta_g^U}$.
\end{proof}
An analogous proof shows the following.
\begin{prop}\label{prophaut}
Let 
$F\in \caF_a$ be an $F$-function of the form
$F(r)=\frac{\exp\lmk {-r^{\theta}}\rmk}{(1+r)^{4}}$ with a constant  $0<\theta<1$.
Let $\Psi\in \hat \caB_{F}([0,1])$ be a path of interactions satisfying $\Psi_1\in \hat \caB_F([0,1])$.
Define $\Psi^{(0)}\in \hat \caB_{F}([0,1])$ by
 \begin{align}\label{pzpo}
 \begin{split}
& \Psi^{(0)}\lmk X; t\rmk:=
 \left\{
\begin{gathered}
 \Psi\lmk X; t\rmk,\quad \text{if }\quad X\subset H_U\quad\text{or}\quad
 X\subset H_D
 \\
 0,\quad \text{otherwise}
 \end{gathered}
 \right.,
 \end{split}
 \end{align}
 for each $X\in{\mathfrak S}_{\bbZ^2}$, $t\in[0,1]$.
Then
 $
 \lmk\tau_{1,0}^{\Psi^{(0)}}\rmk^{-1}\tau_{1,0}^{\Psi}$ belongs to $\haut(\caA)$.
\end{prop}
\begin{proof}
Define $\tilde F$ as in (\ref{apple}) with some $0<\theta'<\theta$.
%Setting
%\begin{align}
% \Psi^{(1)}\lmk X; t\rmk:=\Psi^{(0)}\lmk X; t\rmk-\Psi\lmk X; t\rmk,
% \end{align}
The same argument as in Theorem \ref{thm:quasiauto} {\it Step 1.} implies
that  there exists $\Xi^{(1)}\in \hat \caB_{\tilde F}[0,1]$ with $\tilde F\in\caF_a$,
such that
\begin{align}
\tau_{1,0}^{\Psi}=
\tau_{1,0}^{\Psi^{(0)}}\tau_{0,1}^{ \Xi^{(1)}}.
\end{align}
This $\Xi^{(1)}$ is given by the formula (\ref{psisdef})
for current  $\Psi$ and $\Psi^{(1)}\lmk X; t\rmk:=\Psi^{(0)}\lmk X; t\rmk-\Psi\lmk X; t\rmk$.
To prove the theorem, it suffices to show that 
$\tau_{0,1}^{ \Xi^{(1)}}$ belongs to $\haut(\caA)$.
Indeed, for any $0<\theta_0<\frac\pi 4$,
as in Theorem \ref{thm:quasiauto} {\it Step 2}, we have 
\begin{align}\label{tanuki}
\begin{split}
&\sum_{\substack{Z: Z\nsubseteq C_{[0,\theta_0],L}
\\\text{and}\;\;Z\nsubseteq C_{[0,\theta_0],R}
}}
\sup_{t\in [0,1]}\lV \Xi^{(1)}\lmk Z,t\rmk\rV\\
&\le
\frac{8}{C_{F}}\lmk e^{2I_F(\Psi)}-1\rmk
\sum_{m\ge 0} 
\sum_{\substack{X:  
X(m)\nsubseteq C_{[0,\theta_0],L}
\\\text{and}\;\;X(m)\nsubseteq C_{[0,\theta_0],R}
} }
\left[
\sup_{t\in[0,1]}\lmk \lV \Psi^{(1)}\lmk X; t\rmk \rV\rmk
\lv X\rv G_{F}\lmk m\rmk
\right]
<\infty.
\end{split}
\end{align}
To see this, note that if $X$ in the last line has a  non-zero contribution to the sum,
then at least one of the following occurs.
\begin{description}
\item[(i)] $X\cap C_{[\theta_{0}, \frac\pi 2],U}\neq\emptyset$,
and $X\cap H_{D}\neq\emptyset$
\item[(ii)]$X\cap  C_{[\theta_{0}, \frac\pi 2],D}\neq\emptyset$,
and $X\cap H_{U}\neq\emptyset$
\item[(iii)] $X\subset C_{[0,\theta_{0}]}$ and
\begin{description}
\item[(1)]
$X\cap C_{[0,\theta_{0}],L}\neq\emptyset$ and
$X\cap C_{[0,\theta_{0}],R}\neq\emptyset$, or 
\item[(2)]
$X\subset C_{[0,\theta_{0}], R}$,
$X\cap \check
 C_{[0,\theta_{0}]}\neq\emptyset$,
$X\cap \check
 C_{[-\theta_{0},0]}\neq\emptyset$
and $X{(m)}\cap \lmk C_{[0,\theta_{0}], R}\rmk^{c}\neq \emptyset$,
\item[(3)]
$X\subset C_{[0,\theta_{0}], L}$,
$X\cap \check
 C_{[\pi-\theta_{0},\pi ]}\neq\emptyset$,
$X\cap \check
 C_{[\pi, \pi+\theta_{0}]}\neq\emptyset$
and $X{(m)}\cap \lmk C_{[0,\theta_{0}], L}\rmk^{c}\neq \emptyset$.
\end{description}
\end{description}
Therefore, the summation in the second line  of (\ref{tanuki}) is bounded by
\begin{align*}
\frac{8}{C_{F}}\lmk e^{2I_F(\Psi)}-1\rmk
\lmk
\begin{gathered}
b_{0}(\theta_{0}, \pi-\theta_{0}, \pi, 2\pi)
+b_{0}(0,\pi, \pi+\theta_{0}, 2\pi-\theta_{0})
+b_{0}(-\theta_{0}, \theta_{0}, \pi-\theta_{0}, \pi+\theta_{0})\\
+b_{1}(-\theta_{0}, 0,\theta_{0})
+b_{1}(\pi-\theta_{0}, \pi, \pi+\theta_{0})
\end{gathered}
\rmk<\infty,
\end{align*}
from Lemma \ref{gomap} and \ref{yomap},
proving (\ref{tanuki}).

Therefore, as in {\it Step 3.} of Theorem \ref{thm:quasiauto},
setting
\begin{align}
\tilde\Xi(Z,t):=\left\{
\begin{gathered}
\Xi^{(1)}(Z,t),\quad\text{if}\quad 
 Z\subseteq C_{[0,\theta_0],L}\quad
\text{or}\;\;Z\subseteq C_{[0,\theta_0],R}
\\
0\quad\text{otherwise}
\end{gathered}\right.,
\end{align}
we obtain $\tau_{0,1}^{ \Xi^{(1)}}=\inn\circ\tau_{0,1}^{ \tilde \Xi}$.
By the definition, $\tau_{0,1}^{ \tilde \Xi}$ decomposes as
$\tau_{0,1}^{ \tilde \Xi}= \zeta_L\otimes\zeta_R$, with some $\zeta_\sigma\in\Aut\lmk\caA_{C_{[0,\theta_0],\sigma}}\rmk$, $\sigma=L,R$.
As this holds for any $0<\theta_0<\frac\pi 4$, we conclude $\tau_{0,1}^{ \Xi^{(1)}}\in\haut(\caA)$.
\end{proof}
\begin{thm}\label{thmqautreal}
Let 
$F\in \caF_a$ be an $F$-function of the form
$F(r)=\frac{\exp\lmk {-r^{\theta}}\rmk}{(1+r)^{4}}$ with a constant  $0<\theta<1$.
Let $\Psi\in \hat \caB_{F}([0,1])$ be a path of interactions satisfying $\Psi_1\in \hat \caB_F([0,1])$.
 If $\Psi$ is $\beta$-invariant, then $\tau_{1,0}^{\Psi}$
belongs to $\guaut(\caA)$.
\end{thm}
\begin{proof}
Define $\Psi^{(0)}$ as in (\ref{pzpo}) for our $\Psi$.
By Proposition \ref{prophaut}, we have
$\lmk\tau_{1,0}^{\Psi^{(0)}}\rmk^{-1}\tau_{1,0}^{\Psi}\in \haut(\caA)$.
On the other hand, applying Theorem  \ref{thm:quasiauto} to $\Psi^{(0)}\in\hat\caB_F([0,1])$,
we see that $\tau_{1,0}^{\Psi^{(0)}}$ belongs to 
$\sqaut(\caA).$
Note that $\Psi^{(0)}(X;t)$ is non-zero only if $X\subset H_U$ or $X\subset H_D$,
and it coincides with $\Psi(X;t)$ when it is non-zero.
Therefore, if $\Psi$ is $\beta$-invariant, $\Psi^{(0)}$ is $\beta_g^U$-invariant. 
Therefore, by Theorem \ref{thm:quasiauto}, we have
$\tau_{1,0}^{\Psi^{(0)}}\in \gsqaut(\caA)$.
Hence we have
$\tau_{1,0}^{\Psi}\in \guaut(\caA)$.
%we obtain a decomposition
%\begin{align}\label{usagi2}
%\begin{split}
%\tau_{1,0}^{\Psi^{(0)}}=&\inn\circ
%\lmk
% \alpha_{(0,\theta_1]}\otimes
%\alpha_{(\theta_1,\theta_2]}\otimes \alpha_{(\theta_2,\theta_3]}
%\otimes  \alpha_{(\theta_3,\frac\pi 2]}
%\rmk\\
%&\circ
%\lmk
%\alpha_{(\theta_{0.8},\theta_{1.2}]}
%\alpha_{(\theta_{1.8},\theta_{2.2}]}\otimes \alpha_{(\theta_{2.8},\theta_{3.2}]}
%\rmk,
%\end{split}
%\end{align}
%with $\alpha_X$ such that (\ref{fukurou}) and (\ref{azarashi}).
%On the other hand,
%we 
\end{proof}

\begin{proofof}[Theorem \ref{main}]
Let $\Phi_0\in\caP_{UG}$ be the fixed trivial interaction with a unique gapped ground state.
Its ground state $\omega_0:=\omega_{\Phi_0}$ is of a product form (\ref{ozs}). 
For any $\Phi\in \caP_{SL\beta}$, we have 
$\Phi_0\sim\Phi$.
%, via a path $\Phi : [0,1]\to \caP_{UF\beta}$.
Then by Theorem \ref{mo},
 there exists some 
$\Psi\in\hat\caB_F([0,1])$ with $\Psi_{1}\in \hat\caB_{F}([0,1])$ for some $F\in \caF_a$ of the form
$F(r)=\frac{\exp\lmk {-r^{\theta}}\rmk}{(1+r)^{4}}$ with a $0<\theta<1$,
such that
$\omega_{\Phi}=\omega_{\Phi_0}\circ\tau_{1,0}^{\Psi}$.
From Theorem \ref{thm:quasiauto}, $\tau_{1,0}^\Psi$ belongs to
$\sqaut(\caA)$. 
Because  $\Phi\in \caP_{SL\beta}$, 
$\omega_{\Phi}=\omega_{\Phi_0}\circ\tau_{1,0}^{\Psi}$ is $\beta$-invariant.
Then, by Theorem  \ref{defindexspt}, 
$\IG(\omega_{\Phi})$ is not empty.
Therefore, we may define $h_\Phi:=h(\omega_{\Phi})$ by Definition \ref{theindexdef}.

To see that $h_\Phi$ is an invariant of $\sim_\beta$,
let $\Phi_1,\Phi_2\in \caP_{SL\beta}$
with $\Phi_1\sim_\beta\Phi_2$.
Then by Theorem \ref{mo},
there exists some  $\beta$-invariant
$\Psi\in\hat\caB_F([0,1])$ with $\Psi_{1}\in\hat\caB_F([0,1])$
for some $F\in \caF_a$ of the form
$F(r)=\frac{\exp\lmk {-r^{\theta}}\rmk}{(1+r)^{4}}$ with a constant $0<\theta<1$ 
such that
$\omega_{\Phi_2}=\omega_{\Phi_1}\circ\tau_{1,0}^{\Psi}$.
Applying Theorem  \ref{thmqautreal}, 
to this $\Psi$,
$\tau_{1,0}^{\Psi}$
belongs to $\guaut(\caA)$.
Then Theorem \ref{stabilitythm}
implies 
\begin{align}
h_{\Phi_2}=h(\omega_{\Phi_2})
=h(\omega_{\Phi_1}\circ\tau_{1,0}^{\Psi})=h(\omega_{\Phi_1})
=h_{\Phi_1},
\end{align}
proving the stability.
\end{proofof}

\section{Automorphisms with factorized $d^{0}_{H_{U}}\alpha$}\label{autocasesec}
When $\alpha\in \eaut(\omega)$ has some good factorization property with respect to
the action of $\beta_{g}^{U}$, the index $h(\omega)$
can be calculated without going through GNS representations.
\begin{defn}
For $\alpha\in \Aut\lmk\caA\rmk$, we set
\begin{align}
\lmk d^{0}_{H_{U}}\alpha\rmk (g):=\alpha^{-1}\beta_{g}^{U}\circ \alpha\circ \lmk \beta_{g}^{U}\rmk^{-1},\quad g\in G.
\end{align}
We say that $d^{0}_{H_{U}}\alpha$ is factorized into left and right if there are automorphisms
$\gamma_{g,\sigma}\in\Aut\lmk\caA_{H_{\sigma}}\rmk$, $g\in G$, $\sigma=L,R$
such that 
\begin{align}\label{dzd}
\lmk d^{0}_{H_{U}}\alpha\rmk (g)=\inn\circ \lmk \gamma_{g,L}\otimes\gamma_{g,R}\rmk,\quad
g\in G.
\end{align}
\end{defn}
For known examples of $2$-dimensional SPT-phases like 
 \cite{cglw} and
\cite{MillerMIyake2016}
 and \cite{Beni2016}
or injective PEPS \cite{molnar},
this property holds.
From such an automorphism, we can derive an outer action of $G$.
\begin{lem}\label{seta}
Let $\alpha\in \Aut\lmk\caA\rmk$ be an automorphism.
Suppose that $d^{0}_{H_{U}}\alpha$ is factorized into left and right i.e., there are automorphisms
$\gamma_{g,\sigma}\in\Aut\lmk\caA_{H_{\sigma}}\rmk$, $g\in G$, $\sigma=L,R$
such that 
\begin{align}\label{autofac}
\lmk d^{0}_{H_{U}}\alpha\rmk (g)=\inn\circ \lmk \gamma_{g,L}\otimes\gamma_{g,R}\rmk,\quad
g\in G.
\end{align}
Then there are unitaries $v_{\sigma}(g,h)\in \caU\lmk \caA_{H_{\sigma}}\rmk$,
$g,h\in G$, $\sigma=L,R$ such that
\begin{align}\label{usec}
\gamma_{g,{\sigma}}\beta_{g}^{{\sigma}U}\gamma_{h,{\sigma}}\beta_{h}^{{\sigma}U}
 \lmk \gamma_{gh,{\sigma}}\beta_{gh}^{{\sigma}U}\rmk^{-1}
 =\Ad\lmk v_{\sigma}(g,h)\rmk.
\end{align}
\end{lem}
\begin{proof}
Because $\beta_{g}^{U}$ is a group action, substituting (\ref{autofac}), we get
\begin{align}
&\id_{\caA}
=\alpha^{-1}\beta_{g}^{U}\alpha\circ\alpha^{-1}\beta_{h}^{U}\alpha\circ
\lmk \alpha^{-1}\beta_{gh}^{U}\alpha\rmk^{-1}\notag\\
&=\inn\circ 
 \lmk \gamma_{g,L}\beta_{g}^{LU}\otimes\gamma_{g,R}\beta_{g}^{RU}\rmk
 \circ 
 \lmk \gamma_{h,L}\beta_{h}^{LU}\otimes\gamma_{h,R}\beta_{h}^{RU}\rmk
 \circ 
 \lmk \gamma_{gh,L}\beta_{gh}^{LU}\otimes\gamma_{gh,R}\beta_{gh}^{RU}\rmk^{-1}\notag\\
 &=\inn\circ 
 \lmk \gamma_{g,L}\beta_{g}^{LU}\gamma_{h,L}\beta_{h}^{LU}
 \lmk \gamma_{gh,L}\beta_{gh}^{LU}\rmk^{-1}
 \otimes\gamma_{g,R}\beta_{g}^{RU} \gamma_{h,R}\beta_{h}^{RU}\lmk \gamma_{gh,R}\beta_{gh}^{RU}\rmk^{-1}\rmk 
\end{align}
By Lemma \ref{bimp}, we then see that there are 
unitaries $v_{\sigma}(g,h)\in \Aut\lmk \caA_{H_{\sigma}}\rmk$,
$g\in G$, $\sigma=L,R$ satisfying (\ref{usec}).
\end{proof}
 It is well known that 
a third cohomology class  can be associated to cocycle actions \cite{Connes}\cite{jones}.
\begin{lem}\label{iriomote}
Let $\alpha\in \Aut\lmk\caA\rmk$ be an automorphism
such that $d^{0}_{H_{U}}\alpha$ is factorized into left and right as (\ref{autofac}).
Let
 $v_{\sigma}(g,h)\in \caU\lmk \caA_{H_{\sigma}}\rmk$,
$g,h\in G$, $\sigma=L,R$ 
 be unitaries satisfying (\ref{usec}), given in Lemma \ref{seta}.
 Then
 there is some $c_{{\sigma}}\in C^{3}(G,\bbT)$, $\sigma=L,R$
 such that
 \begin{align}\label{sado}
 v_{{\sigma}}(g,h)
 v_{{\sigma}}(gh,k)
 =c_{{\sigma}}(g,h,k) 
 \lmk \gamma_{g,{\sigma}}\circ\beta_{g}^{{\sigma}U}\lmk v_{{\sigma}}(h,k)\rmk\rmk
 v_{{\sigma}}\lmk g, hk\rmk,\quad g,h,k\in G.
 \end{align}
\end{lem}
\begin{proof}
By (\ref{usec}), we have
\begin{align}\hat\gamma_{g,\sigma}\hat\gamma_{h,\sigma}
 =\Ad\lmk v_{\sigma}(g,h)\rmk\circ \hat\gamma_{gh,\sigma}
\end{align}
for
$
\hat\gamma_{g,\sigma}:=\gamma_{g,{\sigma}}\beta_{g}^{{\sigma}U}
$.
Using this, we have
\begin{align}
&\Ad\lmk v_{\sigma}(g,h)\rmk\circ\Ad\lmk v_{\sigma}(gh,k)\rmk\circ\hat\gamma_{ghk,\sigma}
=\Ad\lmk v_{\sigma}(g,h)\rmk\circ
\hat\gamma_{gh,\sigma}\circ\hat\gamma_{k,\sigma}
=\hat\gamma_{g,\sigma}\hat\gamma_{h,\sigma}\hat\gamma_{k,\sigma}
=\hat\gamma_{g,\sigma}\circ \Ad\lmk v_{\sigma}(h,k)\rmk\circ\hat\gamma_{hk,\sigma}\notag\\
&=\Ad\lmk \hat\gamma_{g,\sigma}\lmk  v_{\sigma}(h,k)\rmk\rmk
\hat\gamma_{g,\sigma}\circ \hat\gamma_{hk,\sigma}
=\Ad\lmk
 \hat\gamma_{g,\sigma}\lmk  v_{\sigma}(h,k)\rmk  v_{\sigma}(g,hk)
\rmk\circ\hat\gamma_{ghk,\sigma}.
\end{align}
Because $\caA'\cap \caA=\unit_{\caA}$,
$\hat\gamma_{g,\sigma}\lmk  v_{\sigma}(h,k)\rmk  v_{\sigma}(g,hk)$
and $v_{\sigma}(g,h)v_{\sigma}(gh,k)$ are proportional to each other, proving the Lemma.
\end{proof}
By the same argument as Lemma \ref{lem3}, we can show that this $c_{R}$ is actually a $3$-cocycle.
If $\omega\in \QLS$ is given by
an automorphism $\alpha\in \eaut(\omega)$
 with factorized $d^{0}_{H_{U}}\alpha$, and if 
$\omega_{0}$ is invariant under $\beta_{g}^{U}$,
then we have
 $h(\omega)=[c_{R}]_{H^{3}(G,\bbT)}$, for $c_R$ given in Lemma \ref{iriomote}.
\begin{thm}
Let $\omega_{0}$ be a reference state of the form (\ref{ozs}), and assume that
$\omega_{0}\circ\beta_{g}^{U}=\omega_{0}$ for any $g\in G$.
Let $\alpha\in \qaut\lmk \caA\rmk$
be an automorphism.
Suppose that $d^{0}_{H_{U}}\alpha$ is factorized into left and right as in (\ref{autofac})
with some $\gamma_{g,\sigma}\in \Aut\lmk\caA_{C_{\theta_{0}},\sigma}\rmk$
and $0<\theta_{0}<\frac\pi 2$, for $\sigma=L,R$.
Let
 $v_{\sigma}(g,h)\in \caU\lmk \caA_{H_{\sigma}}\rmk$,
$g,h\in G$, $\sigma=L,R$ 
 be unitaries satisfying (\ref{usec}) (given in Lemma \ref{seta})
and $c_{R}\in C^{3}(G,\bbT)$ satisfying (\ref{sado})
 for these $v_{R}(g,h)$ (given in Lemma \ref{iriomote}).
 Then we have $\omega_{0}\circ\alpha\in\QLS$ with $\IG(\omega_{0}\circ\alpha)\neq\emptyset$, 
 $c_{R}\in Z^{3}(G,\bbT)$,
 and $h(\omega_{0}\circ\alpha)=[c_{R}]_{H^{3}(G,\bbT)}$.
\end{thm}
\begin{proof}
That $\omega_{0}\circ\alpha\in\QLS$  is by definition.
Because 
\begin{align}
\Ad\lmk v_{\sigma}(g,h)\rmk
=
\gamma_{g,{\sigma}}\beta_{g}^{{\sigma}U}\gamma_{h,{\sigma}}\beta_{h}^{{\sigma}U}
 \lmk \gamma_{gh,{\sigma}}\beta_{gh}^{{\sigma}U}\rmk^{-1}
 \in  \Aut\lmk\caA_{C_{\theta_{0}},\sigma}\rmk,
\end{align}
our $ v_{\sigma}(g,h)$ belongs to $\caU\lmk\caA_{C_{\theta_{0}},\sigma}\rmk$.
Because
\begin{align}
\omega_{0}\alpha\circ\alpha^{-1}\beta_{g}^{U}\alpha
=\omega_{0}\beta_{g}^{U}\alpha
=\omega_{0}\alpha
\end{align}
and 
\begin{align}
\alpha^{-1}\beta_{g}^{U}\alpha
=\inn\circ \lmk \gamma_{g,L}\otimes\gamma_{g,R}\rmk\circ\beta_{g}^{U},
\end{align}
with $\gamma_{g,\sigma}\in \Aut\lmk\caA_{C_{\theta_{0},\sigma}}\rmk$,
we have $(\alpha^{-1}\beta_{g}^{U}\alpha)\in \IG(\omega_{0}\alpha, \theta_{0})$,
and $(\gamma_{g,\sigma})\in\caT\lmk\theta_{0}, \alpha^{-1}\beta_{g}^{U}\alpha\rmk$.
Clearly $\alpha\in \eaut(\omega_{0}\circ\alpha)$ and there is $(\alpha_L,\alpha_R,\Theta)\in \caD_{\alpha}^{\theta_{0}}$ because $\alpha\in \qaut(\caA)$.
Set $\gamma_{g}:=\gamma_{g,L}\otimes\gamma_{g,R}$.
From Lemma \ref{ichi}, there is some $W_{g}\in \caU(\caH_{0})$ $g\in G$ satisfying
\begin{align}\label{cento}
\Ad\lmk W_g\rmk\circ\pi_0
=\pi_0\circ\lmk \alpha_L\otimes\alpha_{R}\rmk
\circ\Theta\circ\gamma_g\beta_g^U\circ\Theta^{-1}\circ
\lmk \alpha_L\otimes\alpha_{R}\rmk^{-1},\quad
g\in G.
\end{align}
In particular, because $ v_{R}(h,k)$ belongs to $\caU\lmk\caA_{\lmk C_{\theta_{0}}\rmk_{R}}\rmk$,
$\Theta\in \Aut\lmk\caA_{C_{\theta_{0}}^{c}}\rmk$, and $\gamma_g\beta_g^U$ preserves
$\caA_{\lmk C_{\theta_{0}}\rmk_{R}}$,
we have
\begin{align}\label{cento1}
&\Ad\lmk W_g\rmk\circ
\pi_0\circ\lmk \alpha_L\otimes\alpha_{R}\rmk
\lmk \id_{\caA_{{L}}}\otimes \lmk v_{R}(h,k)\rmk\rmk
=\pi_0\circ\lmk \alpha_L\otimes\alpha_{R}\rmk
\circ\Theta\circ\gamma_g\beta_g^U\circ\Theta^{-1}\circ
\lmk \id_{\caA_{{L}}}\otimes \lmk v_{R}(h,k)\rmk\rmk\notag\\
&=\pi_0\circ\lmk \alpha_L\otimes\alpha_{R}\rmk
\lmk \id_{\caA_{{L}}}\otimes \gamma_{g,R}\beta_g^{RU}\lmk v_{R}(h,k)\rmk\rmk
=\unit_{\caH_{L}}\otimes \pi_R\circ \alpha_{R}\circ
 \gamma_{g,R}\beta_g^{RU}\lmk v_{R}(h,k)\rmk
\end{align}

On the other hand, (\ref{usec}) means 
\begin{align}\label{hitsuji}
\Ad\lmk
\pi_{\sigma}\circ\alpha_{\sigma}\lmk v_{\sigma}(g,h)\rmk
\rmk\pi_{\sigma}
=\pi_{\sigma}\circ\alpha_{\sigma}\circ
\gamma_{g,{\sigma}}\beta_{g}^{{\sigma}U}\gamma_{h,{\sigma}}\lmk \beta_{g}^{{\sigma}U}\rmk^{-1}
 \lmk \gamma_{gh,{\sigma}}\rmk^{-1}\circ\alpha_{\sigma}^{-1}.
\end{align}
From (\ref{cento}) and (\ref{hitsuji}) we have
\begin{align}
\lmk (W_g), (\pi_{\sigma}\circ\alpha_{\sigma}\lmk v_{\sigma}(g,h)\rmk)\rmk\in \IP\lmk
\omega_{0}\circ\alpha, \alpha, \theta_{0},
(\alpha^{-1}\beta_{g}^{U}\alpha), (\gamma_{g,\sigma}),
(\alpha_L,\alpha_R,\Theta)
\rmk.
\end{align}
Now from (\ref{sado}) and then (\ref{cento1}), we obtain
\begin{align}
&\unit_{\caH_{L}}\otimes\pi_{R}\circ\alpha_{R}\lmk v_{R}(g,h)v_{R}(gh,k)\rmk
=c_{{R}}(g,h,k) 
 \unit_{\caH_{L}}\otimes\pi_{R}\circ\alpha_{R}\lmk
 \lmk \gamma_{g,{R}}\circ\beta_{g}^{{R}U}\lmk v_{{R}}(h,k)\rmk\rmk
 v_{{R}}\lmk g, hk\rmk\rmk\notag\\
& =c_{{R}}(g,h,k)
 \lmk
 \Ad\lmk W_g\rmk
\lmk \id_{\caH_{{L}}}\otimes \pi_{R}\alpha_{R}\lmk v_{R}(h,k)\rmk\rmk\rmk
\cdot \lmk \unit_{\caH_{L}}\otimes\pi_{R}\circ\alpha_{R}\lmk
 v_{{R}}\lmk g, hk\rmk\rmk\rmk.
\end{align}
This means 
\begin{align}
c_{R}=c_R\lmk
\omega_{0}\circ\alpha, \alpha, \theta_{0},
(\alpha^{-1}\beta_g^{U}\alpha), (\gamma_{g,\sigma}),
(\alpha_L,\alpha_R,\Theta),\lmk (W_g), ((\pi_{\sigma}\circ\alpha_{\sigma}\lmk v_{\sigma}(g,h)\rmk))\rmk
\rmk
\end{align}
in the Definition \ref{nagasaki}.
Hence we get 
$c_{R}\in Z^{3}(G,\bbT)$,
 and $h(\omega_{0}\circ\alpha)=[c_{R}]_{H^{3}(G,\bbT)}$.
\end{proof}

{\bf Acknowledgment.}\\
{
The author is grateful to Hal Tasaki for a stimulating discussion over two-dimensional
Dijkgraaf-Witten model.
The author is grateful to Yasuyuki Kawahigashi for introducing the author
various papers from operator algebra. 
This work was supported by JSPS KAKENHI Grant Number 16K05171 and 19K03534.
It was also supported by JST CREST Grant Number JPMJCR19T2.
The present result and the main idea of the proof were announced publicly on 15 December 2020 at IAMP One World Mathematical Physics Seminar (see you-tube video)\cite{IAMP}.
It was also presented in the international meeting {\it Theoretical studies of topological phases of matter}
on 17 December 2020, and
 Current Developments in Mathematics 4th January 2021 via zoom with a lecture note.
 \cite{IAMP}.
}

\appendix
\section{Basic Notations}\label{notasec}
 For a finite
set $S$, $\#S$ indicates the number of elements in $S$.
For $t\in\bbR$, $[t]$ denotes the smallest integer less than or equal to $t$.

For a Hilbert space $\caH$, $B(\caH)$ denotes the set of all bounded operators on $\caH$.
If $V:\caH_1\to\caH_2$ is a linear map from a Hilbert space $\caH_1$ to 
another Hilbert space $\caH_2$,
then $\Ad (V):B(\caH_1)\to B(\caH_2)$ denotes the map
$\Ad(V)(x):=V x V^*$, $x\in B(\caH_1)$.
Occasionally we write $\Ad_V$ instead of $\Ad(V)$.
For a $C^{*}$-algebra $\caB$ and $v\in \caB$, we set 
$\Ad(v)(x):=\Ad_{v}(x):=vxv^{*}$, $x\in \caB$.

%For a $C^{*}$-algebra $\caB$, we denote by
%$\caB_{1}$ the set of all elements
%in  $\caB$ with norm less than or equal to $1$ and by $\caB_{+,1}$ the set of all positive elements
%in  $\caB$ with norm less than or equal to $1$.
For a state $\omega$, $\varphi$ on a $C^{*}$-algebra $\caB$,
we write $\omega\sim_{\rm q.e.}\varphi$ when they are quasi-equivalent. (See \cite{BR1}.)
We denote by $\Aut \caB$ the group of automorphisms on a $C^{*}$-algebra $\caB$.
The group of inner automorphisms on  a unital $C^{*}$-algebra $\caB$ is denoted by $\Inn \caB$.
For $\gamma_1,\gamma_2\in\Aut(\caB)$, $\gamma_1=\inn\circ\gamma_2$
means there is some unitary $u$ in $\caB$ such that $\gamma_1=\Ad(u)\circ\gamma_2$.
For a unital $C^{*}$-algebra $\caB$, the unit of $\caB$ is denoted by $\unit_{\caB}$.
For a Hilbert space we write $\unit_{\caH}:=\unit_{\caB(\caH)}$.
For a unital $C^{*}$-algebra $\caB$, by $\caU(\caB)$, we mean
the set of all unitary elements in $\caB$.
For a Hilbert space we write $\caU(\caH)$ for $\caU(\caB(\caH))$.

For a state $\varphi$ on $\caB$ and a $C^{*}$-subalgebra $\caC$ of $\caB$,
$\varphi\vert_{\caC}$ indicates the restriction of $\varphi$ to $\caC$.

\section{Automorphisms on UHF-algebras}
\begin{lem}\label{bimp}
Let $\mathfrak A$, $\mathfrak B$ be UHF-algebras.
If automorphisms $\gamma_{\mathfrak A}\in\Aut(\mathfrak A)$, $\gamma_{\mathfrak B}\in\Aut(\mathfrak B)$ and a unitary $W\in \caU\lmk{\mathfrak A}\otimes \mathfrak B\rmk$
satisfy
\begin{align}
\lmk
\gamma_{\mathfrak A}\otimes\gamma_{\mathfrak B}
\rmk(X)=
\Ad_W(X),\quad X\in{\mathfrak A\otimes\mathfrak B},
\end{align}
then there are unitaries $u_{\mathfrak A}\in\caU(\mathfrak A)$ and  $u_{\mathfrak B}\in\caU(\mathfrak B)$ such that
\begin{align}
\begin{split}
&\gamma_{\mathfrak A}\lmk X_{\mathfrak A}\rmk
=\Ad_{u_{\mathfrak A}}(X_{\mathfrak A})
,\quad X_{\mathfrak A}\in{\mathfrak A},\\
&\gamma_{\mathfrak B}\lmk X_{\mathfrak B}\rmk=\Ad_{u_{\mathfrak B}}(X_{\mathfrak B})
,\quad X_{\mathfrak B}\in{\mathfrak B}.
\end{split}
\end{align}
\end{lem}
\begin{proof}
Fix some irreducible representations 
$(\caH_{\mathfrak A},\pi_{\mathfrak A})$, $(\caH_{\mathfrak B},\pi_{\mathfrak B})$,
of $\mathfrak A$, $\mathfrak B$.
We claim that there are unitaries $v_{\mathfrak A}\in\caU(\caH_{\mathfrak A})$
and $v_{\mkB}\in\caU(\caH_{\mkB})$ such that 
\begin{align}\label{clu}
\begin{split}
&\Ad_{v_{\mkA}}\lmk \pi_{\mkA}(X_{\mkA})\rmk=\pi_{\mkA}\circ\gamma_{\mkA} (X_{\mkA}),\quad
X_{\mkA}\in\mkA,\\
&\Ad_{v_{\mkB}}\lmk \pi_{\mkB}(X_{\mkB})\rmk=\pi_{\mkB}\circ\gamma_{\mkB} (X_{\mkB}),\quad
X_{\mkB}\in\mkB.
\end{split}
\end{align}
To see this, note that
\begin{align}
\lmk
\pi_{\mkA}\circ\gamma_{\mkA}\otimes \pi_{\mkB}\circ\gamma_{\mkB}
\rmk
=\Ad_{\lmk \pi_{\mkA}\otimes \pi_{\mkB}\rmk(W)}\circ\lmk \pi_{\mkA}\otimes\pi_{\mkB}\rmk.
\end{align}
From this, $\pi_{\mkA}\circ\gamma_{\mkA}$ (resp. $\pi_{\mkB}\circ\gamma_{\mkB}$) is quasi-equivalent to
$\pi_{\mkA}$ (resp. $\pi_{\mkB}$).
Because $\pi_{\mkA}$ and $\pi_{\mkB}$ are irreducible, by the Wigner theorem, 
there are  unitaries $v_{\mathfrak A}\in\caU(\caH_{\mathfrak A})$
and $v_{\mkB}\in\caU(\caH_{\mkB})$
satisfying (\ref{clu}).

We then have
\begin{align}
\Ad_{\lmk \pi_{\mkA}\otimes \pi_{\mkB}\rmk(W)}\circ\lmk \pi_{\mkA}\otimes\pi_{\mkB}\rmk
=\lmk \pi_{\mkA}\circ\gamma_{\mkA}\rmk\otimes \lmk \pi_{\mkB}\circ\gamma_{\mkB}\rmk
=\lmk \Ad_{v_{\mkA}}\circ\pi_{\mkA}\rmk\otimes \lmk \Ad_{v_{\mkB}}\circ\pi_{\mkB}\rmk
=\Ad_{v_{\mkA}\otimes v_{\mkB}}\circ \lmk\pi_{\mkA}\otimes\pi_{\mkB}\rmk.
\end{align}
Because $\pi_{\mkA}\otimes\pi_{\mkB}$ is irreducible, there is a $c\in\bbT$
such that
\begin{align}
\lmk \pi_{\mkA}\otimes\pi_{\mkB}\rmk(W)=c\lmk v_{\mkA}\otimes v_{\mkB}\rmk.
\end{align}

We claim there is a unitary $u_{\mathfrak B}\in\caU(\mathfrak B)$ such that
\begin{align}
%\pi_{\mkA}\lmk u_{\mkA}\rmk=v_{\mkA},\quad
\pi_{\mkB}\lmk u_{\mkB}\rmk=v_{\mkB}.
\end{align}
Choose a unit vector $\xi\in\caH_{\mathfrak A}$ with $\braket{\xi}{ v_{\mkA}\xi}\neq 0$.
For each $x\in\caB(\caH_\mkA\otimes\caH_{\mkB})$, the map
\begin{align}
\caH_{\mathfrak B}\times \caH_{\mathfrak B}\ni 
(\eta_1,\eta_2)\mapsto
\braket{\lmk \xi\otimes \eta_1\rmk}{x \lmk \xi\otimes \eta_2\rmk}
\end{align}
is a bounded sesquilinear form.
Therefore, there is a unique $\Phi_\xi(x)\in\caB(\caH_{\mathfrak B})$ such that
\begin{align}
\braket{\eta_1}{\Phi_\xi(x)\eta_2}
=\braket{\lmk \xi\otimes \eta_1\rmk}{x \lmk \xi\otimes \eta_2\rmk},\quad
(\eta_1,\eta_2)\in \caH_{\mathfrak B}\times \caH_{\mathfrak B}.
\end{align}
The map
$\Phi_\xi:\caB\lmk\caH_\mkA\otimes\caH_{\mkB}\rmk\to \caB(\caH_{\mathfrak B})$  is linear and
\begin{align}\label{norp}
\lV\Phi_\xi(x)\rV\le \lV x\rV,\quad x\in\caB(\caH).
\end{align}
Because $W$ belongs to $\mathfrak A\otimes \mathfrak B$,
there are sequence 
\begin{align}
z_N=\sum_{i=1}^{n_N} a_i^{(N)}\otimes b_i^{(N)},\quad
\text{with}\quad  a_i^{(N)}\in\mkA,\quad b_i^{(N)}\in\mkB
\end{align}
such that
\begin{align}
\lV
W-z_N
\rV<\frac 1N.
\end{align}
Because of (\ref{norp}),
we have
\begin{align}
\lV
\Phi_{\xi}\lmk \lmk \pi_\mkA\otimes\pi_\mkB\rmk\lmk W-z_N\rmk\rmk
\rV<\frac 1N.
\end{align}
Note that 
\begin{align}
\Phi_{\xi}\lmk \lmk \pi_\mkA\otimes\pi_\mkB\rmk\lmk z_N\rmk\rmk
=\sum_{i=1}^{n_N} \braket{\xi}{ \pi_{\mathfrak A}\lmk a_i^{(N)}\rmk\xi}\pi_{\mathfrak B}\lmk b_i^{(N)}\rmk
\in\pi_{\mathfrak B}(\mathfrak B).
\end{align}
Therefore, we have
\begin{align}
c\braket{\xi}{v_{\mkA}\xi} v_{\mkB}
=\Phi_{\xi}\lmk c\lmk v_{\mkA}\otimes v_{\mkB}\rmk
\rmk
=\Phi_{\xi}\lmk
\lmk \pi_{\mkA}\otimes\pi_{\mkB}\rmk(W)\rmk
\in\overline{\pi_{\mathfrak B}(\mathfrak B)}^n,
\end{align}
where $\overline{\cdot}^{n}$ denotes the norm closure.
Because $\pi_{\mkB}\lmk\mathfrak B\rmk$ is norm-closed, we have $\overline{\pi_\mkB\lmk\mathfrak B\rmk}^n
=\pi_\mkB\lmk\mathfrak B\rmk$.
Hence we have $v_{\mkB}\in\pi_\mkB\lmk\mathfrak B\rmk$, i.e.,
there is a unitary $u_\mkB\in\mathfrak B$ such that $v_{\mkB}=\pi_{\mathfrak B}\lmk u_{\mkB}\rmk$.

We then have
\begin{align}
\pi_\mkB\circ\Ad_{u_{\mkB}}(X)
=\Ad_{\pi_\mkB(u_\mkB)}\circ\pi_{\mkB}(X)
=\Ad_{v_\mkB}\circ\pi_{\mkB}(X)
=\pi_\mkB\circ\gamma_\mkB(X),\quad X\in\mkB.
\end{align}
As $\mkB$ is simple, 
$\Ad_{u_{\mkB}}(X)=\gamma_\mkB(X)$
for all $ X\in\mkB$.

The proof for $\mkA$ is the same.

\end{proof}

\section{$F$-functions}\label{ffunc}
In this section, we collect various estimates about $F$-functions.
Let us first start from the definition.
\begin{defn}\label{ffuncdef}
An $F$-function $F$ on $({\bbZ^2}, \dist)$
is a non-increasing function $F:[0,\infty)\to (0,\infty)$
such that
\begin{description}
\item[(i)]
$\lV F\rV:=\sup_{x\in{\bbZ^2}}\lmk \sum_{y\in{\bbZ^2}}F\lmk {\dist}(x,y)\rmk\rmk<\infty$,
and
\item[(ii)]
$C_{F}:=\sup_{x,y\in{\bbZ^2}}\lmk \sum_{z\in{\bbZ^2}}
\frac{F\lmk {\dist}(x,z)\rmk F\lmk {\dist}(z,y)\rmk}{F\lmk {\dist}(x,y)\rmk}\rmk<\infty$.
\end{description}
These are called \emph{uniform integrability} and the \emph{convolution identity}, respectively.
\end{defn}

We denote by $\caF_{a}$ a class of $F$-functions which decay quickly.
\begin{defn}\label{fadef}
We say an $F$-function $F$ belongs to $\caF_{a}$
if
\begin{description}
\item[(i)]
for any $k\in\bbN\cup\{0\}$ and $0<\theta\le 1$, we have
\begin{align}
\kappa_{\theta,k, F}:=\sum_{n=0}^{\infty} (n+1)^{k}\lmk  F(n)\rmk^{\theta}<\infty,
\end{align}
and
\item[(ii)] for any $0<\theta<1$,
there is a $F$-function $\tilde F_{\theta}$ such that
\begin{align}\label{tildef}
\max\left\{ F\lmk\frac r 3\rmk, \lmk  F\lmk \lcm \frac r 3 \rcm \rmk\rmk^{\theta}\right\}\le
\tilde F_{\theta}(r),\quad\quad  r\ge 0.
\end{align}

\end{description}

\end{defn}
For example, a function $F(r)=\frac{\exp\lmk {-r^{\theta}}\rmk}{(1+r)^{4}}$ with a constant  $0<\theta<1$
belongs to $\caF_a$. (See section 8 of \cite{NSY}.)

In this section, we derive inequalities about $F\in \caF_{a}$.
In order for that the following Lemma is useful.
We will freely identify $\bbC$ and $\bbR^{2}$ in an obvious manner.
\begin{lem}\label{sest}
For $0\le \theta_1<\theta_2\le 2\pi$, $c>0$, and $r\ge 0$, set
\begin{align}
S_{r,c}^{[\theta_1,\theta_2]}
:=
\left\{
s e^{i\theta}\in\bbR^2\mid
r\le s< r+c,\quad
\theta\in [\theta_1,\theta_2]
\right\}.
\end{align}
Then we 
have
\begin{align}
\#\lmk
S_{r,c}^{[\theta_1,\theta_2]}\cap \bbZ^2
\rmk
\le \pi \lmk 2\sqrt 2+c\rmk^{2}(r+1).
\end{align}
In particular, we have
\begin{align}
\#\lmk
S_{r,1}^{[\theta_1,\theta_2]}\cap \bbZ^2
\rmk \le 64 (r+1).
\end{align}
\end{lem}
\begin{proof}
Because the diameter of a $2$-dimensional unit square is $\sqrt 2$,
any unit square $B$ of $\bbZ^2$ with $B\cap S_{r,c}^{[\theta_1,\theta_2]}\cap \bbZ^2\neq \emptyset$
satisfies $B\subset \hat S_{r,c}^{[\theta_1,\theta_2]}(\sqrt 2)$.
Therefore, we have
\begin{align}\label{boxn}
\#\left\{B\mid \text{unit square of} \quad \bbZ^2
\quad\text{with }\quad B\cap S_{r,c}^{[\theta_1,\theta_2]}\cap \bbZ^2\neq \emptyset\right\}
=\sum_{B: B\cap S_{r,c}^{[\theta_1,\theta_2]}\cap \bbZ^2\neq \emptyset} 1
\le \lv \hat S_{r,c}^{[\theta_1,\theta_2]}(\sqrt 2)\rv.
\end{align}
Note that the area of $ \hat S_{r,c}^{[\theta_1,\theta_2]}(\sqrt 2)$, denoted by 
$\lv \hat S_{r,c}^{[\theta_1,\theta_2]}(\sqrt 2)\rv$
 is less than
\begin{align}
\lv \hat S_{r,c}^{[\theta_1,\theta_2]}(\sqrt 2)\rv\le 
\pi\lmk (r+c+\sqrt 2)^2-(r-\sqrt 2)^2)\rmk
\le 
\pi
(2r+c) 
(2\sqrt 2+c)
\le \pi\lmk 2\sqrt 2+c\rmk^{2}(r+1)
%\le
%32 (r+1)\le 64 (r+1)
\end{align}
if $r>\sqrt 2$.
For $r\le \sqrt 2$, 
we have 
\begin{align}
\lv \hat S_{r,c}^{[\theta_1,\theta_2]}(\sqrt 2)\rv
\le \pi\lmk (r+c+\sqrt 2)^2\rmk
\le \pi\cdot (2\sqrt 2+c)^2
\le \pi\lmk 2\sqrt 2+c\rmk^{2}(r+1).
\end{align}
Hence in any case, we have 
\begin{align}
\lv \hat S_{r,c}^{[\theta_1,\theta_2]}(\sqrt 2)\rv
\le  \pi\lmk 2\sqrt 2+c\rmk^{2}(r+1).
\end{align}
Substituting this to (\ref{boxn}), we obtain
\begin{align}
\#\left\{B\mid \text{unit square of} \quad \bbZ^2
\quad\text{with }\quad B\cap S_{r,c}^{[\theta_1,\theta_2]}\cap \bbZ^2\neq \emptyset\right\}
\le \pi\lmk 2\sqrt 2+c\rmk^{2}(r+1).
\end{align}
On the other hand, we have
\begin{align}
& \#\left\{S_{r,c}^{[\theta_1,\theta_2]}\cap \bbZ^2\right\}
=\sum_{z\in S_{r,c}^{[\theta_1,\theta_2]}\cap \bbZ^2}1 
= \sum_{z\in S_{r,c}^{[\theta_1,\theta_2]}\cap \bbZ^2}
\sum_{B: \text{unit square of}\;\bbZ^2} \frac 14 \unit_{z\in B}\notag\\
&=\sum_{B: \text{unit square of}\;\bbZ^2} 
\sum_{z\in S_{r,c}^{[\theta_1,\theta_2]}\cap \bbZ^2}\frac 14 \unit_{z\in B}
\le \sum_{\substack{B: \\
\text{unit square of}\;\bbZ^2\\
B\cap S_{r,c}^{[\theta_1,\theta_2]}\cap \bbZ^2\neq \emptyset}} 1\notag\\
&=\#\left\{B\mid \text{unit square of} \quad \bbZ^2
\quad\text{with }\quad B\cap S_{r,c}^{[\theta_1,\theta_2]}\cap \bbZ^2\neq \emptyset\right\}
\le \pi\lmk 2\sqrt 2+c\rmk^{2}(r+1).
\end{align}
\end{proof}

For an $F$-function $F\in\caF_{a}$, 
define a function $G_{F}$ on $t\ge 0$ by
\begin{align}\label{gfdef}
G_{F}(t):= \sup_{x\in{\bbZ^2}}\lmk
\sum_{y\in{\bbZ^2}, {\dist}(x,y)\ge t} F\lmk {\dist}(x,y)\rmk
\rmk,\quad t\ge 0.
\end{align}
Note that by uniform integrability the supremum is finite for all $t$.
In particular, for any $0<\theta<1$, we have
\begin{align}
\begin{split}
&G_F(t)
\le  \sum_{r=[t]}^\infty\sum_{\substack{y\in \bbZ^2 :\\
r\le \dist(0,y)<r+1
}} F\lmk \dist(0,y)\rmk
\le \sum_{r=[t]}^\infty\sum_{y\in S_{r,1}^{[0,2\pi]}\cap\bbZ^2} F(r)
\le \sum_{r=[t]}^\infty \# \lmk S_{r,1}^{[0,2\pi]}\cap\bbZ^2\rmk F(r)\\
%\sum_{r=t}^\infty\sum_{y\in S_r^{[0,2\pi]}} F(r)
&\le 64 \sum_{r=[t]}^\infty (r+1) F(r)
= 64  \sum_{r=[t]}^\infty (r+1) F(r)^{\theta} F(r)^{1-\theta}
\le  64  \lmk \sum_{r=0}^\infty (r+1) F(r)^{\theta}\rmk F([t])^{1-\theta}
\le 64 \cdot \kappa_{\theta, 1, F}\cdot F([t])^{1-\theta}
<\infty.
\end{split}
\end{align}
Substituting this, for any $0<\alpha\le 1$, $0<\theta,\varphi<1$, and $k\in\bbN\cup \{0\}$, we have
\begin{align}\label{ugugug}
\begin{split}
&\sum_{n=0}^{\infty} (1+n)^{k} \lmk G_{F}(n)\rmk^{\alpha}
\le\lmk  64 \cdot \kappa_{\theta, 1, F}\rmk^{\alpha} 
\sum_{n=0}^{\infty}(1+n)^{k} \cdot F(n)^{\alpha\lmk 1-\theta\rmk}
=\lmk  64 \cdot \kappa_{\theta, 1, F}\rmk^{\alpha} 
\kappa_{\alpha\lmk 1-\theta\rmk, k, F}<\infty,\\
&\sum_{n=[\frac r3]}^{\infty} (1+n)^{k} \lmk G_{F}(n)\rmk^{\alpha}
\le\lmk  64 \cdot \kappa_{\theta, 1, F}\rmk^{\alpha} 
\sum_{n=[\frac r 3]}^{\infty}(1+n)^{k} \cdot
\lmk F(n)^{\alpha\lmk 1-\theta\rmk}\rmk^{(1-\varphi)} \lmk F(n)^{\alpha\lmk 1-\theta\rmk}\rmk^{\varphi}\\
&=\lmk  64 \cdot \kappa_{\theta, 1, F}\rmk^{\alpha} 
\kappa_{\alpha\lmk 1-\theta\rmk\lmk1-\varphi\rmk, k, F}
F\lmk \lcm\frac r3\rcm\rmk^{\alpha\lmk 1-\theta\rmk\varphi}.
\end{split}
\end{align}
For any $0<c\le 1$, we have
\begin{align}\label{bzeroest}
\begin{split}
&\sum_{r=0}^{\infty} F(cr) (r+2)^{3}
=\sum_{l=0}^{\infty}
\sum_{\substack{r\in\bbZ_{\ge 0}\\ l\le cr< l+1}}F(cr) (r+2)^{3}
\le \sum_{l=0}^{\infty}
\sum_{\substack{r\in\bbZ_{\ge 0}\\\frac{l}c\le r<\frac{l+1}c}}F(l) \lmk \frac{l+1}c+2\rmk^{3}\\
&\le
\sum_{l=0}^{\infty}
F(l) \lmk \frac{l+1}c+2\rmk^{3}\lmk \frac{l+1}c-(\frac lc-1)+1\rmk
\le
\sum_{l=0}^{\infty}
F(l) \lmk \frac{l+1}c+2\rmk^{4}\\
&\le
\frac 1{c^{4}}\sum_{l=0}^{\infty}
F(l) \lmk l+3\rmk^{4}\le \frac {3^{4}\kappa_{1, 4, F}}{c^{4}} <\infty.
\end{split}
\end{align}
We also have for $m\in\bbZ_{\ge 0}$ and $0<c\le 1$ that
\begin{align}\label{kikiki}
\begin{split}
&\sum_{r_1=0}^\infty\sum_{\substack{r\in{\bbZ_{\ge 0}}: \\\sqrt{r^2+r_1^2}c\ge (m+1)} }
(r_1+1)
 F\lmk
  \sqrt{r^2+r_1^2}c-(m+1)
  \rmk\\
&  \le 
  \sum_{l=0}^{\infty}
  \sum_{\substack{r_{1}, r\in\bbZ_{\ge 0}\\
  l\le \sqrt{r^2+r_1^2}c- (m+1)< l+1
  }}
  (r_1+1)
 F\lmk
  \sqrt{r^2+r_1^2}c-(m+1)
  \rmk\\
  &\le\sum_{l=0}^{\infty}
  \sum_{\substack{r_{1}, r\in\bbZ_{\ge 0}\\
  l\le \sqrt{r^2+r_1^2}c- (m+1)< l+1
  }}
  \lmk \frac{l+m+2}{c}+1\rmk \cdot F(l)\\
  &\le
    \sum_{l=0}^{\infty}
 \#\left\{
 \bbZ^{2}\cap S_{\frac{l+m+1}{c}, \frac 1c}^{[0,2\pi]}
 \right\}
  \lmk \frac{l+m+2}{c}+1\rmk \cdot F(l)\\
&\le 
  \sum_{l=0}^{\infty}\pi \lmk 2\sqrt 2+\frac 1c\rmk^{2}\lmk \frac{l+m+1}{c}+1\rmk\cdot
  \lmk \frac{l+m+2}{c}+1\rmk \cdot F(l)\\
  &\le 
  \sum_{l=0}^{\infty}\pi \lmk 2\sqrt 2+\frac 1c\rmk^{2}
  \lmk \frac{l+m+3}{c}\rmk^{2}\cdot F(l)\\
  &\le 
\pi \lmk 2\sqrt 2+\frac 1c\rmk^{2}
\frac{(m+3)^{2}}{c^{2}}
  \sum_{l=0}^{\infty}\lmk l+1\rmk^{2} F(l)\\
  &\le\pi \lmk 2\sqrt 2+\frac 1c\rmk^{2}
\frac{(m+3)^{2}}{c^{2}}
\kappa_{1,2,F}
\le \lmk\frac 3 c\rmk^{2}\lmk 2\sqrt 2+\frac 1c\rmk^{2}\pi
{(m+1)^{2}}
\kappa_{1,2,F}.
\end{split}
\end{align}

Recall (\ref{ccheck}) and (\ref{czdefn}).
\begin{lem}\label{gomap1}
Let $\varphi_1<\varphi_2<\varphi_3<\varphi_4$ with $\varphi_4-\varphi_1<2\pi$.
Then
we have
\begin{align}
\sum_{\substack{x\in \check C_{[\varphi_1,\varphi_2]},\\y\in \check C_{[\varphi_3,\varphi_4]}}}
F\lmk {\dist}(x,y)\rmk\le
(64)^3
\frac {3^{4}\kappa_{1, 4, F}}{\lmk c^{(0)}_{\varphi_1,\varphi_2,\varphi_3,\varphi_4}\rmk ^{4}}.
\end{align}

%
% (64)^2\lmk \lV\lv \Phi_1^{(1)}\rV\rv\rmk\lmk \sum_{m\ge 0}G_{F}\lmk m\rmk\rmk\\
%&\sum_{r=0}^\infty
%F\lmk 
%\sqrt{1- \max\left\{
%\cos\lmk \varphi_3-\varphi_2\rmk,
%\cos\lmk \varphi_4-\varphi_1\rmk,0
%\right\}}
%r)
%\rmk (r+1)^3\\
\end{lem}

\begin{proof}
Let $x=s_1e^{i\phi_1}\in \check C_{[\varphi_1,\varphi_2]}$
and $y=s_2e^{i\phi_2}\in \check C_{[\varphi_3,\varphi_4]}$ with $s_1,s_2\ge 0$.
If $\cos\lmk \phi_2-\phi_1\rmk\ge 0$, then
we have
\begin{align}
\begin{split}
&{\dist}(x,y)=\sqrt{s_1^2+s_2^2-2s_1s_2\cos\lmk \phi_2-\phi_1\rmk}
\ge 
\sqrt{s_1^2+s_2^2}
\sqrt{1-\cos\lmk \phi_2-\phi_1\rmk}
\\
%&\ge \sqrt{s_1^2+s_2^2-2s_1s_2
%\max\left\{
%\cos\lmk \varphi_3-\varphi_2\rmk,
%\cos\lmk \varphi_4-\varphi_1\rmk
%\right\}}\notag\\
&\ge \sqrt{1- \max\left\{
\cos\lmk \varphi_3-\varphi_2\rmk,
\cos\lmk \varphi_4-\varphi_1\rmk,0
\right\}}
\sqrt{s_1^2+s_2^2}
.
\end{split}
\end{align}
If $\cos\lmk \phi_2-\phi_1\rmk<0$, then
we have
\begin{align}
\begin{split}
&{\dist}(x,y)=\sqrt{s_1^2+s_2^2-2s_1s_2\cos\lmk \phi_2-\phi_1\rmk}
\ge \sqrt{s_1^2+s_2^2}.\\
%&\ge 
% \sqrt{1- \max\left\{
%\cos\lmk \varphi_3-\varphi_2\rmk,
%\cos\lmk \varphi_4-\varphi_1\rmk
%\right\}}
%\sqrt{s_1^2+s_2^2}.
%&\ge \sqrt{s_1^2+s_2^2-2s_1s_2
%\max\left\{
%\cos\lmk \varphi_3-\varphi_2\rmk,
%\cos\lmk \varphi_4-\varphi_1\rmk
%\right\}}\notag\\
\end{split}
\end{align}
Hence for any $x=s_1e^{i\phi_1}\in \check C_{[\varphi_1,\varphi_2]}$
and $y=s_2e^{i\phi_2}\in \check C_{[\varphi_3,\varphi_4]}$ with $s_1,s_2\ge 0$
we have
\begin{align}
{\dist}(x,y)\ge  \sqrt{1- \max\left\{
\cos\lmk \varphi_3-\varphi_2\rmk,
\cos\lmk \varphi_4-\varphi_1\rmk,0
\right\}}
\sqrt{s_1^2+s_2^2}
=c^{(0)}_{\varphi_1,\varphi_2,\varphi_3,\varphi_4}\sqrt{s_1^2+s_2^2}.
\end{align}
Substituting this estimate, we obtain
\begin{align}
\begin{split}
&\sum_{\substack{x\in \check C_{[\varphi_1,\varphi_2]},\\y\in \check C_{[\varphi_3,\varphi_4]}}}
F\lmk {\dist}(x,y)\rmk
\\
&\le
\sum_{r_1=0}^\infty\sum_{r_2=0}^\infty
\sum_{\substack{x\in S_{r_1,1}^{[\varphi_1,\varphi_2]}\cap \bbZ^2\\
y\in S_{r_2,1}^{[\varphi_3,\varphi_4]}\cap \bbZ^2}}
F\lmk {\dist}(x,y)\rmk
\notag\\
&\le
\sum_{r_1=0}^\infty\sum_{r_2=0}^\infty
F\lmk 
{c^{(0)}_{\varphi_1,\varphi_2,\varphi_3,\varphi_4}}
\sqrt{r_1^2+r_2^2})
\rmk\#\lmk
S_{r_1,1}^{[\varphi_1,\varphi_2]}\cap \bbZ^2
\rmk
\#\lmk
S_{r_2,1}^{[\varphi_3,\varphi_4]}\cap \bbZ^2
\rmk
\end{split}\\
\begin{split}
&\le (64)^2
\sum_{r_1=0}^\infty\sum_{r_2=0}^\infty
F\lmk 
{c^{(0)}_{\varphi_1,\varphi_2,\varphi_3,\varphi_4}}
\sqrt{r_1^2+r_2^2})
\rmk(r_1+1)(r_2+1)\\
&\le  (64)^2
\sum_{r=0}^\infty
\sum_{\substack{r_{1}, r_{2}\in\bbZ_{\ge 0}\\
(r_{1}, r_{2})\in S_{r, 1}^{[0,\frac\pi 2]}\cap\bbZ^{2}
}
}
F\lmk 
{c^{(0)}_{\varphi_1,\varphi_2,\varphi_3,\varphi_4}}
\sqrt{r_1^2+r_2^2})
\rmk(r_1+1)(r_2+1)\\
& 
\le (64)^2
\sum_{r=0}^\infty
F\lmk 
{c^{(0)}_{\varphi_1,\varphi_2,\varphi_3,\varphi_4}}
r
\rmk (r+2)^2\cdot
\#\lmk
S_{r}^{[0,\frac\pi 2]}\cap \bbZ^2
\rmk
\\
&\le
(64)^3
\sum_{r=0}^\infty
F\lmk 
{c^{(0)}_{\varphi_1,\varphi_2,\varphi_3,\varphi_4}}
r
\rmk (r+2)^3
\\
&\le(64)^3
\frac {3^{4}\kappa_{1, 4, F}}{\lmk c^{(0)}_{\varphi_1,\varphi_2,\varphi_3,\varphi_4}\rmk ^{4}}
\end{split}
 \end{align}
 We used Lemma \ref{sest} to bound $\#\lmk
S_{r,1}^{[0,\frac\pi 2]}\cap \bbZ^2
\rmk$ etc.
At the last inequality we used (\ref{bzeroest})
 
\end{proof}

Set
\begin{align}\label{lphidef}
L_\varphi:=
\left\{
z\in\bbR^2\mid
\arg z=\varphi
\right\},\quad \varphi\in [0,2\pi).
\end{align}
and
\begin{align}
{c^{(1)}}_{\zeta_1,\zeta_2,\zeta_3}
:=\sqrt{1-\max\left\{\cos (\zeta_1-\zeta_2),\cos (\zeta_1-\zeta_3)\right\}},\quad
\zeta_1,\zeta_2,\zeta_3\in [0,2\pi).
\end{align}

\begin{lem}\label{lcone}
Let $\varphi, \theta_1,\theta_2\in\bbR$
with $\theta_1<\theta_2$ and
$0<\lv\varphi-\theta_0\rv<\frac \pi 2$
for all $\theta_0\in [\theta_1,\theta_2]
$.
Then we have
\begin{align}
\sum_{x\in \check C_{[\theta_1,\theta_2]}}\sum_{y\in L_\varphi(m)} F\lmk {\dist}(x,y)\rmk
\le
64\cdot 144\cdot 24\cdot 
\lmk {c^{(1)}}_{\varphi,\theta_1,\theta_2}\rmk^{-4}
\lmk
\pi \kappa_{1,2,F}+F(0)
\rmk(m+1)^{4},
\end{align}
for any $m\in\bbN\cup\{0\}$
%
%
%Then for any $x=r_0 e^{i\theta}\in C_{[\theta_1,\theta_2]}$ with 
%$r_0\ge 0$, we have
%\begin{align}
%\begin{split}
%&\sum_{y\in L_\varphi(m)} F\lmk {\dist}(x,y)\rmk\\
%&\le
%\sum_{\substack{r\in{\bbZ_{\ge 0}}: \\\sqrt{r^2+r_0^2}{c^{(1)}}_{\theta,\varphi_1,\varphi_2}\ge (m+1)} }
%24(m+1) F\lmk
%  \sqrt{r^2+r_0^2}{c^{(1)}}_{\theta,\varphi_1,\varphi_2}-(m+1)
%  \rmk\\
%& +24\frac{(m+1)^2}{{c^{(1)}}_{\theta,\varphi_1,\varphi_2}}F(0)\unit_{r_0\le \frac{m+1}{  {c^{(1)}}_{\theta,\varphi_1,\varphi_2}  }}
%
%
%&\le
%\sum_{\substack{r\in\bbZ: \\\sqrt{r^2+r_0^2}{c^{(1)}}_{\theta,\varphi_1,\varphi_2}\ge (m+1)} }
%12(m+1) F\lmk
%  \sqrt{r^2+r_0^2}{c^{(1)}}_{\theta,\varphi_1,\varphi_2}-(m+1)
%  \rmk\\
%& +24\frac{(m+1)^2}{{c^{(1)}}_{\theta,\varphi_1,\varphi_2}}F(0)\unit_{r_0\le \frac{m+1}{  {c^{(1)}}_{\theta,\varphi_1,\varphi_2}  }}
%
%
% \lmk \sum_{r: \sqrt{r^2+r_0^2}
%{c^{(1)}}_{\varphi, \theta_1,\theta_2}
%\ge (m+1)} 
%12(m+1) F\lmk
%  \sqrt{r^2+r_0^2}{c^{(1)}}_{\varphi, \theta_1,\theta_2}
%  -(m+1)
%  \rmk\rmk+24\frac{(m+1)^2}
%{{c^{(1)}}_{\varphi, \theta_1,\theta_2}}
%F(0)
%\sum_{r: \sqrt{r^2+r_0^2}
%\sqrt{1-\max\left\{\cos (\theta_1-\varphi),\cos (\theta_2-\varphi)\right\}}
%\ge (m+1)} 
%12(m+1)\\
%& F\lmk
%  \sqrt{r^2+r_0^2}\sqrt{1-\max\left\{\cos (\theta_1-\varphi),\cos (\theta_2-\varphi)\right\}}-(m+1)
%  \rmk\\
%& +24\frac{(m+1)^2}
%{\sqrt{1- \max\left\{\cos (\theta_1-\varphi),\cos (\theta_2-\varphi)\right\}}}F(0).
\end{lem}
\begin{proof}
For each $r\in\bbZ$, set
\begin{align}
T_{\varphi, r,m}:=
\left\{
se^{i\theta}\in \bbR^{2}\mid
r\le s\cos(\theta-\varphi)\le r+1,\quad
-m\le s\sin(\theta-\varphi)\le m
\right\}.
\end{align}
Note that $s\cos(\theta-\varphi)$ is a projection of $se^{i\theta}$ onto $L_\varphi$
and $|s\sin(\theta-\varphi)|$ is the distance of $se^{i\theta}$ from the line including
 $L_\varphi$.
Then we have
\begin{align}\label{ldecom}
L_\varphi(m)\subset
\cup_{r=-m}^\infty T_{\varphi, r,m}\cap\bbZ^2,\quad\text{and}\quad
\lv 
\hat T_{\varphi, r,m}(\sqrt 2)
\rv\le (2\sqrt 2+1)(2m+2\sqrt 2)\le
12(m+1).
\end{align}
Because the diameter of a $2$-dimensional unit square is $\sqrt 2$,
any unit square $B$ of $\bbZ^2$ with $B\cap T_{\varphi, r,m}
\cap \bbZ^2\neq \emptyset$
satisfies $B\subset \hat T_{\varphi, r,m}(\sqrt 2)$.
Therefore, we have
\begin{align}
\begin{split}
&\#\left\{B\mid \text{unit square of} \quad \bbZ^2
\quad\text{with }\quad B\cap
T_{\varphi, r,m}
\cap \bbZ^2\neq \emptyset\right\}
=\sum_{B: B\cap
T_{\varphi, r,m}
\cap \bbZ^2\neq \emptyset} 1\\
&\le \lv 
\hat T_{\varphi, r,m}(\sqrt 2)
\rv\le 12(m+1).
\end{split}
\end{align}
On the other hand, we have
\begin{align}\label{tnum}
& \#\left\{T_{\varphi,r,m}\cap \bbZ^2\right\}
=\sum_{z\in T_{\varphi,r,m}\cap \bbZ^2}1 
= \sum_{z\in T_{\varphi,r,m}\cap \bbZ^2}
\sum_{B: \text{unit square of}\;\bbZ^2} \frac 14 \unit_{z\in B}\notag\\
&=\sum_{B: \text{unit square of}\;\bbZ^2} 
\sum_{z\in T_{\varphi,r,m}\cap \bbZ^2}\frac 14 \unit_{z\in B}
\le \sum_{\substack{B: \\
\text{unit square of}\;\bbZ^2\\
B\cap T_{\varphi,r,m}\cap \bbZ^2\neq \emptyset}} 1\notag\\
&=\#\left\{B\mid \text{unit square of} \quad \bbZ^2
\quad\text{with }\quad B\cap T_{\varphi,r,m}\cap \bbZ^2\neq \emptyset\right\}
\le 12(m+1).
\end{align}
If  $x\in \check C_{[\theta_1,\theta_2]}$, we have $x=r_0 e^{i\theta_0}$
for some $r_0\ge 0$ and $\theta_0\in [\theta_1,\theta_2]$.
By the assumption, we have $0<|\theta_0-\varphi|<\frac\pi 2$
hence $0<\cos(\varphi-\theta_0)<1$.
Therefore, for any $r\in\bbR$, we have
\begin{align}
\begin{split}
&\dist(x, r e^{i\varphi})
=\sqrt{r^2+r_0^2-2r_0r\cos (\theta_0-\varphi)}
\ge \sqrt{r^2+r_0^2}\sqrt{1-\cos (\theta_0-\varphi)}\\
&\ge
 \sqrt{r^2+r_0^2}\sqrt{1-\max\left\{\cos (\theta_1-\varphi),\cos (\theta_2-\varphi)\right\}}.
\end{split}
\end{align}
Therefore, for 
any $x\in \check C_{[\theta_1,\theta_2]}$  and $y\in T_{\varphi, r,m}$,
we have
\begin{align}
\begin{split}
&{\dist}(x,y)
\ge \dist (x, r e^{i\varphi})-(m+1)
%\ge  \sqrt{r^2+r_0^2}{c^{(1)}}_{\theta,\varphi_1,\varphi_2}-(m+1)\\
=\sqrt{r^2+r_0^2}{c^{(1)}}_{\varphi,\theta_1,\theta_2}-(m+1).
\end{split}
\end{align}
From this and (\ref{ldecom}) and (\ref{tnum}), for any 
$x=r_0 e^{i\theta_0}\in C_{[\theta_1,\theta_2]}$, $r_0\ge 0$, we have 
\begin{align}
\begin{split}
&\sum_{y\in L_\varphi(m)} F\lmk {\dist}(x,y)\rmk
\le 
\sum_{r=-m}^\infty \sum_{y\in\lmk  T_{\varphi, r,m}\cap\bbZ^2\rmk}
 F\lmk {\dist}(x,y)\rmk
 \le
 \sum_{r=-\infty}^\infty \sum_{y\in\lmk  T_{\varphi, r,m}\cap\bbZ^2\rmk}
 F\lmk {\dist}(x,y)\rmk\\
& \le
 \sum_{\substack{r\in\bbZ: \\
 \sqrt{r^2+r_0^2}
 {c^{(1)}}_{\varphi,\theta_1,\theta_2}\ge (m+1)}} \sum_{y\in\lmk  T_{\varphi, r,m}\cap\bbZ^2\rmk} F\lmk
  \sqrt{r^2+r_0^2}{c^{(1)}}_{\varphi,\theta_1,\theta_2}-(m+1)
  \rmk\\
& +\sum_{\substack{
r\in\bbZ: \\
\sqrt{r^2+r_0^2}{c^{(1)}}_{\varphi,\theta_1,\theta_2}<(m+1)}} \sum_{y\in\lmk  T_{\varphi, r,m}\cap\bbZ^2\rmk}
 F\lmk 0\rmk\\
&\le
\sum_{\substack{r\in\bbZ: \\\sqrt{r^2+r_0^2}{c^{(1)}}_{\varphi,\theta_1,\theta_2}\ge (m+1)}} 
12(m+1) F\lmk
  \sqrt{r^2+r_0^2}{c^{(1)}}_{\varphi,\theta_1,\theta_2}-(m+1)
  \rmk\\
& +\sum_{\substack{r\in\bbZ:\\ \sqrt{r^2+r_0^2}{c^{(1)}}_{\varphi,\theta_1,\theta_2}<(m+1)}} 12(m+1)
 F\lmk 0\rmk\\
&\le 
\sum_{\substack{r\in{\bbZ_{\ge 0}}: \\\sqrt{r^2+r_0^2}{c^{(1)}}_{\varphi,\theta_1,\theta_2}\ge (m+1)} }
24(m+1) F\lmk
  \sqrt{r^2+r_0^2}{c^{(1)}}_{\varphi,\theta_1,\theta_2}-(m+1)
  \rmk\\
& +36\frac{(m+1)^2}{{c^{(1)}}_{\varphi,\theta_1,\theta_2}}F(0)\unit_{r_0\le \frac{m+1}{  {c^{(1)}}_{\varphi,\theta_1,\theta_2}  }}
%&=
% \lmk \sum_{r\in\bbZ: \sqrt{r^2+r_0^2}
%{c^{(1)}}_{\varphi, \theta_1,\theta_2}
%\ge (m+1)} 
%12(m+1) F\lmk
%  \sqrt{r^2+r_0^2}{c^{(1)}}_{\varphi, \theta_1,\theta_2}
%  -(m+1)
%  \rmk\rmk+24\frac{(m+1)^2}
%{{c^{(1)}}_{\varphi, \theta_1,\theta_2}}
%F(0)
%.
\end{split}
\end{align}
We then get
\begin{align}
\begin{split}
&\sum_{x\in \check C_{[\theta_1,\theta_2]}}\sum_{y\in L_\varphi(m)} F\lmk {\dist}(x,y)\rmk\\
&\le
\sum_{r_1=0}^\infty
\sum_{x\in S_{r_1,1}^{[\theta_1,\theta_2]}\cap\bbZ^2}\\
&\lmk \begin{gathered}
\sum_{\substack{r\in{\bbZ_{\ge 0}}: \\\sqrt{r^2+r_1^2}{c^{(1)}}_{\varphi,\theta_1,\theta_2}\ge (m+1)} }
24(m+1) F\lmk
  \sqrt{r^2+r_1^2}{c^{(1)}}_{\varphi,\theta_1,\theta_2}-(m+1)
  \rmk\\
 +36\frac{(m+1)^2}{{c^{(1)}}_{\varphi,\theta_1,\theta_2}}F(0)\unit_{r_1\le \frac{m+1}{  {c^{(1)}}_{\varphi,\theta_1,\theta_2}  }}
\end{gathered}\rmk\\
&\le
\sum_{r_1=0}^\infty
64(r_1+1)\\
&\lmk \begin{gathered}
\sum_{\substack{r\in{\bbZ_{\ge 0}}: \\\sqrt{r^2+r_1^2}{c^{(1)}}_{\varphi,\theta_1,\theta_2}\ge (m+1)} }
24(m+1) F\lmk
  \sqrt{r^2+r_1^2}{c^{(1)}}_{\varphi,\theta_1,\theta_2}-(m+1)
  \rmk\\
 +36\frac{(m+1)^2}{{c^{(1)}}_{\varphi,\theta_1,\theta_2}}F(0)\unit_{r_1\le \frac{m+1}{  {c^{(1)}}_{\varphi,\theta_1,\theta_2}  }}
\end{gathered}\rmk\\
&\le 
64\cdot 24\cdot 
\lmk\frac 3 {c^{(1)}}_{\varphi,\theta_1,\theta_2}\rmk^{2}\lmk 2\sqrt 2+\frac 1{c^{(1)}}_{\varphi,\theta_1,\theta_2}\rmk^{2}\pi
{(m+1)^{3}}
\kappa_{1,2,F}
+64\cdot 36\cdot \frac{(m+1)^2}{{c^{(1)}}_{\varphi,\theta_1,\theta_2}}F(0)
\lmk\frac{m+1}{{c^{(1)}}_{\varphi,\theta_1,\theta_2}}+1\rmk^{2}\\
&\le
64\cdot 144\cdot 24\cdot 
\lmk {c^{(1)}}_{\varphi,\theta_1,\theta_2}\rmk^{-4}
\lmk
\pi \kappa_{1,2,F}+F(0)
\rmk(m+1)^{4}.
\end{split}
\end{align}
We used (\ref{kikiki}).

\end{proof}

\section{Quasi-local automorphisms}\label{quasilocalsec}
In this section we collect some results from \cite{NSY}, and prove Theorem \ref{mo}.
\begin{defn}
A norm-continuous interaction on $\caA$ defined on an interval $[0,1]$
is a map
$\Phi:{\mathfrak S}_{\bbZ^2}\times [0,1]\to \caA_{\rm loc}$ such that
\begin{description}
\item[(i)]
for any $t\in[0,1]$, $\Phi(\cdot, t):{\mathfrak S}_{\bbZ^2}\to \caA_{\rm loc}$
is an interaction, and
\item[(ii)]
for any $Z\in{\mathfrak S}_{\bbZ^2}$, the map
$\Phi(Z,\cdot ):[0,1]\to \caA_{Z}$
is norm-continuous.
\end{description}
\end{defn}

To ensure that the interactions induce quasi-local automorphisms we need to impose sufficient decay properties on the interaction strength.

\begin{defn}\label{bfdef}
Let $F$ be an $F$-function on $({\bbZ^2},\dist)$.
We denote by $\caB_{F}([0,1])$ the set of all
norm continuous interactions $\Phi$ on $\caA$ defined on an interval $[0,1]$
such that the function $\lV
\Phi
\rV_{F}: [0,1]\to \bbR$ defined by
\begin{align}
\lV
\Phi
\rV_F(t):=
\sup_{x,y\in{\bbZ^2}}\frac{1}{F\lmk {\dist}(x,y)\rmk}\sum_{Z\in{\mathfrak S}_{\bbZ^2}, Z\ni x,y}
\lV\Phi(Z;t)\rV,\quad t\in[0,1],
\end{align}
is uniformly bounded, i.e., 
$\sup_{t\in[0,1]}\lV \Phi \rV(t)<\infty$.
It follows that $t \mapsto \lV \Phi \rV_F(t)$ is integrable, and we set 
\begin{align}
I_F(\Phi):=I_{1,0}(\Phi):= C_{F} \int_{0}^{1} dt\lV \Phi \rV_F(t),
\end{align}
with $C_F$ given in Definition \ref{ffuncdef}.
We also set
\begin{align}
\lV\lv \Phi\rV\rv_F:=
\sup_{x,y\in{\bbZ^2}}\frac{1}{F\lmk {\dist}(x,y)\rmk}\sum_{Z\in{\mathfrak S}_{\bbZ^2}, Z\ni x,y}
\sup_{t\in [0,1]}\lmk \lV\Phi(Z;t)\rV\rmk
\end{align}
and denote by $\hat\caB_{F}([0,1])$ the set of all $\Phi\in\caB_{F}([0,1])$
with $\lV\lv \Phi\rV\rv<\infty$.
\end{defn}
We will need some more notation.
For $\Phi\in \caB_{F}([0,1])$ and $0\le m\in\bbR$, 
we introduce a path of interactions $\Phi_{m}$ by
\begin{align}\label{pm}
\Phi_{m}\lmk X;t\rmk:=|X|^{m}\Phi\lmk X;t\rmk,\quad X\in{\mathfrak S}(\bbZ^2),\quad t\in[0,1].
\end{align}
An interaction gives rise to local (and here, time-dependent) Hamiltonians, via
\begin{align}
H_{\Lambda,\Phi}(t):=\sum_{Z\in\Lambda}\Phi(Z,t),\quad t\in[0,1],\quad \Lambda\in{\mathfrak S}_{\bbZ^2}.
\end{align}
We denote by $U_{\Lambda,\Phi}(t;s)$, the solution of 
\begin{align}
\frac{d}{dt} U_{\Lambda,\Phi}(t;s)=-iH_{\Lambda,\Phi}(t) U_{\Lambda,\Phi}(t;s),\quad s, t\in[0,1]\\
U_{\Lambda,\Phi}(s;s)=\unit.
\end{align}
We define corresponding automorphisms $\tau_{t,s}^{(\Lambda),\Phi}, \hat{\tau}_{t,s}^{(\Lambda), \Phi}$ on $\caA$ by
\begin{align}
\tau_{t,s}^{(\Lambda), \Phi}(A):=U_{\Lambda,\Phi}(t;s)^{*}AU_{\Lambda,\Phi}(t;s),\\
\hat{\tau}_{t,s}^{(\Lambda), \Phi}(A):=U_{\Lambda,\Phi}(t;s)AU_{\Lambda,\Phi}(t;s)^{*},
\end{align}
with $A \in \caA$. 
Note that
\begin{align}\label{inv}
\hat{\tau}_{t,s}^{(\Lambda), \Phi}={\tau}_{s,t}^{(\Lambda), \Phi},
\end{align}
by the uniqueness of the solution of the differential equation.

\begin{thm}[\cite{NSY}]\label{tni}
Let $F$ be an $F$-function on $({\bbZ^2}, \dist)$.
Suppose that $\Phi\in\caB_F([0,1])$.
Then the following holds:
\begin{enumerate}
\item \label{it:tdlimit} The limits
\begin{align}
\tau_{t,s}^{\Phi}(A):=\lim_{\Lambda \nearrow{\bbZ^2}}\tau_{t,s}^{(\Lambda), \Phi}(A),\quad
\hat \tau_{t,s}^{\Phi}(A):=\lim_{\Lambda \nearrow{\bbZ^2}}\hat \tau_{t,s}^{(\Lambda), \Phi}(A),\quad
A\in\caA, \quad t,s\in[0,1]
\end{align}
exist and defines strongly continuous families of automorphisms on $\caA$
such that $\hat\tau_{t,s}^{\Phi}=\tau_{s,t}^{\Phi}={\tau_{t,s}^{\Phi}}^{-1}$ and
\begin{align}
\hat \tau_{t,s}^{\Phi}\circ\hat \tau_{s,u}^{\Phi}=\hat \tau_{t,u}^{\Phi},\quad \tau_{t,t}^{\Phi}=\id_{\caA}, \quad t,s,u\in[0,1].
\end{align}
\item \label{it:lr}For any $X,Y\in {\mathfrak S}_{\bbZ^2}$ with $X\cap Y=\emptyset$
the bound 
\begin{align}
\lV
\left[
\tau_{t,s}^{\Phi}(A), B
\right]
\rV
\le \frac{2\lV A\rV\lV B\rV}{C_{F}}\lmk e^{2I_F(\Phi)}-1\rmk\lv X\rv G_{F}\lmk d(X,Y)\rmk
\end{align}
holds for all $A\in \caA_{X}$, $B\in\caA_{Y}$, and $t,s\in [0,1]$.

If $\Lambda \in {\mathfrak S}_{\bbZ^2}$ and $X \cup Y \subset \Lambda$, a similar bound holds for $\tau_{t,s}^{(\Lambda),\Phi}$.

\item \label{it:approx}
For any $X\in {\mathfrak S}_{\bbZ^2}$ 
we have
\begin{align}
&\lV
\Delta_{X(m)}\lmk
\tau_{t,s}^{\Phi}(A)\rmk
\rV
\le \frac{8\lV A\rV}{C_{F}}\lmk e^{2I_F(\Phi)}-1\rmk\lv X\rv G_{F}\lmk m\rmk,
\end{align}
for
%for all $\Lambda\in{\mathfrak S}_{\bbZ^2}$and
 $A\in \caA_{X}$.
 Here we set $\Delta_{X(0)}:=\Pi_{X}$ and $\Delta_{X(m)}:=\Pi_{X(m)}-\Pi_{X(m-1)}$
 for $m\in\bbN$.
A similar bound holds for $\tau_{t,s}^{(\Lambda),\Phi}$.
(See (\ref{gfdef}) for the definition of $G_F$.)

\item \label{it:localdyn}
For any $X,\Lambda\in {\mathfrak S}(\bbZ^2)$ with $X\subset\Lambda$, and $A \in \caA_X$
we have
\begin{align}
\lV
\tau_{t,s}^{(\Lambda), \Phi}(A)-\tau_{t,s}^{\Phi}(A)
\rV
\le\frac{2}{C_{F}} \lV A\rV e^{2I_F(\Phi)}I_F(\Phi) \lv X\rv G_{F}\lmk d\lmk X,{\bbZ^2}\setminus\Lambda\rmk
\rmk.
\end{align}
\item\label{hanahana}
If $\beta_{g}^U\lmk \Phi(X;t)\rmk=\Phi(X;t)$ for any $X\in{\mathfrak S}_{\bbZ^2}$, $t\in [0,1]$, and
$g\in G$, then
we have
$\beta_g^U\circ\tau_{t,s}^{ \Phi}=\tau_{t,s}^{\Phi}\circ \beta_g^U$
for any $t,s\in [0,1]$ and
$g\in G$,
\end{enumerate}
\end{thm}
\begin{proof}
	Item~\ref{it:tdlimit} is Theorem~3.5 of~\cite{NSY}, while~\ref{it:lr} and~\ref{it:localdyn} follow from Corollary~3.6 of the same paper by a straightforward bounding of $D(X,Y)$ and the summation in eq.~(3.80) of~\cite{NSY} respectively.
	Item~\ref{it:approx} can be obtained using~\ref{it:lr} and~\cite[Cor. 4.4]{NSY}.	

Suppose that ${\beta_g^U}\lmk \Phi(X;t)\rmk=\Phi(X;t)$ for any $X\in{\mathfrak S}_{\bbZ^2}$, $t\in [0,1]$, and
$g\in G$.
Then we have
\begin{align}
\frac{d}{dt} {\beta_g^U}\lmk U_{\Lambda,\Phi}(t;s)\rmk
=-{\beta_g^U}\lmk iH_{\Lambda,\Phi}(t)\rmk  {\beta_g^U}\lmk U_{\Lambda,\Phi}(t;s)\rmk
=- iH_{\Lambda,\Phi}(t){\beta_g^U}\lmk U_{\Lambda,\Phi}(t;s)\rmk,\quad t\in[0,1]
\end{align}
and ${\beta_g^U}\lmk U_{\Lambda,\Phi}(s;s)\rmk=\unit$.
Hence ${\beta_g^U}\lmk U_{\Lambda,\Phi}(t;s)\rmk$ and $U_{\Lambda,\Phi}(t;s)$
satisfy the same differential equation and initial condition.
Therefore we get ${\beta_g^U}\lmk U_{\Lambda,\Phi}(t;s)\rmk=U_{\Lambda,\Phi}(t;s)$.
From this, we obtain ${\beta_g^U}\tau_{t,s}^{(\Lambda), \Phi}=\tau_{t,s}^{(\Lambda), \Phi}{\beta_g^U}$,
and taking $\Lambda\uparrow \bbZ^2$, we obtain 
${\beta_g^U}\circ\tau_{t,s}^{ \Phi}=\tau_{t,s}^{\Phi}\circ {\beta_g^U}$.
\end{proof}

%Consider the same notation and assumptions as in Theorem~\ref{tni}.
%To continue we need to make additional assumptions on the function $F$.
%In particular, we assume that there is an $\alpha\in(0,1)$ such that
%\begin{align}\label{as:galp}
%	\sum_{n=0}^{\infty} (1+n)^{3}G_F(n)^{\alpha}<\infty,
%\end{align}
%where $G_F$ is as defined in~\eqref{gfdef}.
%Assume also that
%\begin{align}\label{gmoment}
%\sum_{n=0}^{\infty} (1+n)^{k}G_F(n)<\infty,
%\end{align}
%for any $k\in\bbN$.
%Furthermore, we assume that there is an $F$-function $\tilde F$ on $({\bbZ^2}, \dist)$ such that
%\begin{align}\label{as:gf}
%\max\left\{
%F\lmk \frac r 3\rmk, \sum_{n=[\frac r 3]}^{\infty} (1+n)^{2\nu+1}G_F(n)^{\alpha}
%\right\}\le \tilde F(r).
%\end{align}
The following is slightly strengthened version of  Assumption 5.15. of \cite{NSY}.
\begin{assum}\label{nsy5.15}[\cite{NSY}]
 We assume that the family of linear maps
 $\{\caK_t:\caA_{\rm loc}\to \caA\}_{t\in[0,1]}$
 is norm
continuous and satisfy the followings: There is
 a family of linear maps
$\{
\caK_t^{(n)} : \caA_{\Lambda_n} \to \caA_{\Lambda_n} \}_{t\in[0,1]} $
for each $n\ge 1$ such that:
\begin{description}
\item[(i)]
For each $n\ge 1$, the family 
$\{\caK_t^{(n)}
: \caA_{\Lambda_n}\to \caA_{\Lambda_n}\}_{t\in [0,1]}$
satisfies
\begin{description}
\item[(a)]
For each $t\in[0,1]$, $\lmk \caK_t^{(n)}(A)\rmk^*=  \caK_t^{(n)}(A^*)$
for all $\caA_{\Lambda_n}$.
\item[(b)] For each $A\in \caA_{\Lambda_n}$, the function $[0,1]\ni t\to \caK_t^{(n)}(A)$ is norm continuous.
\item[(c)] For each $t\in[0,1]$ the map $\caK_t^{(n)}: \caA_{\Lambda_n}\to \caA_{\Lambda_n}$
 is norm
 continuous and moreover, this
continuity is uniform on $[0,1]$.
\end{description}
 \item[(ii)]
  There is some $p\ge 0$ and a constant $B_1>0$ for which
  %a measurable, locally bounded function  $B:[0,1]\to [0,\infty)$ for which
given any $X\in{\mathfrak S}_{\bbZ^2}$ and $n\ge 1$ large enough so that 
$X\subset\Lambda_n$ 
\[
\lV \caK^{(n)}_t(A)\rV\le B_1  |X|^p\lV A \rV,\quad \text{for all}\quad
A\in\caA_{X}\quad \text{and}\quad t\in [0,1].
\]
\item[(iii)]
There is some $q\ge 0$, a non-negative, non-increasing function $G$ with 
$G(x)\to  0$ as $x\to\infty$,
and a constant $C_1>0$
% measurable, locally bounded function $C:[0,1]\to [0,\infty)$
 for which given any sets $X,Y\in{\mathfrak S}_{\bbZ^2}$
 and $n\ge 1$
  large enough so that $X\cup Y\subset\Lambda_n$,
  \[
  \lV\lcm  \caK^{(n)}_t(A), B\rcm \rV\le C_1|X|^q\lV A \rV\lV B\rV G\lmk\dist(X,Y)\rmk,\quad \text{for all}\quad
A\in\caA_{X},\quad B\in\caA_{Y} \quad\text{and}\quad t\in [0,1].
  \]
\item[(iv)]
There is some $r\ge 0$, a non-negative, non-increasing function $H$ 
with $H(x)\to 0$ as $x\to\infty$,
and a constant $D_1>0$
% measurable, locally bounded function $D:[0,1]\to [0,\infty)$ 
for which given any $X\in{\mathfrak S}_{\bbZ^2}$
there exists $N\ge 1$ such that for $n\ge N$
\[
\lV \caK_t^{(n)} (A)-\caK_t(A)\rV
\le D_{1} |X|^r \lV A\rV H\lmk \dist (X,\bbZ^2\setminus \Lambda_n)\rmk 
\]
 for all $A\in\caA_X$ and $t\in[0,1]$.
\end{description}
\end{assum}
The following theorem is a slight modification of Theorem 5.17 of \cite{NSY}
\begin{thm}\label{nsy5.17}
Let $F\in\caF_a$, with $\tilde F_\theta$ in (\ref{tildef}) for each $0<\theta<1$.
Assume
that $\{\caK_t\}_{t\in [0,1]}$
 is a family of linear maps satisfying Assumption \ref{nsy5.15}, 
 with $G=G_{F}$ in (iii).
(Recall Definition \ref{fadef} and (\ref{gfdef})). 
 Let 
 $\Phi\in \caB_F([0,1])$
  be an interaction satisfying $\Phi_m\in \caB_F([0,1])$ for $m=\max\{p,q,r\}$
  where 
  $p,q,r$ are numbers in Assumption \ref{nsy5.15}.
Then,
the right hand side of the following sum
\begin{align}\label{eq:psis}
\Psi\lmk
Z, t
\rmk:=
\sum_{m\ge 0} \sum_{X\subset Z,\; X(m)=Z}
\Delta_{X(m)}\lmk
\caK_t\lmk
 \Phi\lmk X; t\rmk
\rmk
\rmk,\quad Z\in {\mathfrak S}_{\bbZ^2},\quad t\in [0,1]
\end{align}
defines a path of interaction such that $\Psi\in\caB_{\tilde F_{\theta}}([0,1])$, for any $0<\theta<1$.
 Furthermore, the formula 
 \begin{align}\label{eq:psisn}
	 \Psi^{(n)}\lmk
Z, t
\rmk:=
\sum_{m\ge 0} \sum_{X\subset Z, X(m)\cap\Lambda_{n}=Z}
\Delta_{X(m)}\lmk
\caK_t^{(n)}\lmk
 \Phi\lmk X; t\rmk
\rmk
\rmk
\end{align}
defines $\Psi^{{(n)}}\in \caB_{\tilde F_{\theta}}([0,1])$ , for any $0<\theta<1$ such that
$\Psi^{(n)}\lmk
Z, t
\rmk=0$ unless $Z\subset \Lambda_{n}$,
and satisfies
\begin{align}\label{psio}
\caK_t^{(n)} \lmk H_{\Lambda_n, \Phi}(t)\rmk
=H_{\Lambda_n, \Psi^{(n)}}(t).
\end{align}
For any $t,u\in[0,1]$, we have
\begin{align}\label{convconv2}
\lim_{n\to\infty}\lV
\tau_{t,u}^{\Psi^{(n)}}\lmk A\rmk
-\tau_{t,u}^{\Psi}\lmk A\rmk
\rV=0,\quad A\in\caA.
\end{align}
Furthermore, if $\Phi_{m+k}\in\hat \caB_{F}([0,1])$ for $k\in\bbN\cup\{0\}$,  then we have
$\Psi_k^{{(n)}}, \Psi_k\in \hat \caB_{\tilde F_{\theta}}([0,1])$
for ant $0<\theta<1$.
\end{thm}

\begin{proof}
Because of $F\in\caF_{a}$, we see from 
(\ref{ugugug}) that for any $0<\alpha<1$ and $k\in\bbN$, $G_{F}^{\alpha}$ has a finite $k$-moment.
We also recall (\ref{tildef}) and (\ref{ugugug}) to see that
\begin{align}
\max\left\{ F\lmk\frac r 3\rmk, 
\sum_{n=[\frac r3]}^{\infty }(1+n)^{5}G_{F}(n)^{\alpha}
\right\}\le
\tilde C \tilde F_{\alpha(1-\theta')\varphi}(r),\quad\quad  r\ge 0,
\end{align}
for any $0<\alpha,\theta',\varphi<1$.
As this holds for any $0<\alpha,\theta',\varphi<1$, the condition in (ii) of Theorem 5.17 \cite{NSY}
holds for any $\tilde F_{\theta}$.
Therefore, from (ii) of Theorem 5.17 \cite{NSY}, we get
$\Psi,\Psi^{(n)}\in \caB_{\tilde F_{\theta}}\lmk[0,1]\rmk$
and
$\Psi^{(n)}$ converges locally in $F$-norm to $\Psi$ with respect to $\tilde F_{\theta}$,
for any $0<\theta<1$.
%If $Z$ is a finite set, we see that the right-hand side of~\eqref{eq:psis} contains only finitely many terms and hence is well-defined.
%Moreover, because of the $\Delta_{X(m)}$, it follows that $\Psi^{(s)}(Z,t) \in \caA_Z$.
%Since $\tau_{t,s}^{\tilde\Phi}$ is in automorphism we see that $\Psi^{(s)}(Z,t)$ is self-adjoint, and hence defines an interaction.
%That this interaction is in $\caB_{\tilde{F}}([0,1])$ follows then from Theorem 5.17(i) of~\cite{NSY}.
%The conditions of this theorem can be verified using Theorem~\ref{tni}, where in the notation of~\cite{NSY} we have $p=0$ and $q=r=1$.

%Similarly, equation~\eqref{eq:psisn} defines an interaction, and~\eqref{psio} can be verified by an explicit calculation, if we note that $\tau_{t,s}^{(\Lambda_n), \Phi}(\Phi(X,t))$ is in $\caA_{\Lambda_n}$.
%By part (ii) of Theorem~5.17 of~\cite{NSY} it follows that $\Psi^{(n)(s)} \in \caB_{\tilde{F}}([0,1])$, and moreover that $\Psi^{(n)(s)}$ converges to $\Psi^{(s)}$ in $F$-norm with respect to $\tilde{F}$.
Theorem 5.13 of~\cite{NSY} implies
\begin{align}
\sup_{n}\int_{0}^{1}\lV \Psi^{(n)}\rV_{\tilde F_{\theta}}(t) dt<\infty,
\end{align}
see also~\cite[eq. (5.101)]{NSY}.
Therefore, from Theorem 3.8 of \cite{NSY}, we obtain~\eqref{convconv2}.

By the proof of Theorem 5.17 and Theorem 5.13 (5.87) of \cite{NSY}, 
if $\Phi_{k+m}\in\hat \caB_{F}([0,1])$ for some $k\in\bbN$, then we have
$\Psi_k^{{(n)(s)}}, \Psi_k^{{(s)}}\in \hat \caB_{\tilde F}([0,1])$.
More precisely, instead of (5.88) of \cite{NSY}, we obtain
\begin{align}
\begin{split}
&\sum_{\substack{Z\in{\mathfrak S}_{\bbZ^2}\\Z\ni x,y
}} |Z|^{k}\sup_{t\in[0,1]}\lV\Psi(Z;t)
\rV\\
&\le 
B_1 \sum_{\substack{Z\in{\mathfrak S}_{\bbZ^2}\\Z\ni x,y
}} |Z|^{k+p}\sup_{t\in[0,1]}\lV\Phi(Z;t)
\rV
+4C_1
\sum_{n=0}^\infty G_F(n) \sum_{X: X(n+1)\ni x,y} |X|^q|X(n+1)|^k
\sup_{t\in[0,1]} \lV\Phi(X;t)\rV\\
&\le
B_1 \lV\lv \Phi_{k+p}\rv\rV_F F\lmk \dist(x,y)\rmk
+4C_1
\sum_{n=0}^\infty G_F(n)(2n+3)^{2k}
\sum_{X:X(n+1)\ni x,y} |X|^{q+k}
\sup_{t\in[0,1]} \lV\Phi(X;t)\rV\\
&\le
B_1 \lV\lv \Phi_{k+p}\rv\rV_F F\lmk \dist(x,y)\rmk
+\tilde C_{\theta}\tilde F_{\theta}\lmk \dist(x,y)\rmk
\lV\lv \Phi_{q+k}\rv\rV_F<\infty
%+4\cdot 3^{2k}C_1
%\cdot 64 \cdot \kappa_{\frac 12, 1, F}
%\kappa_{\frac 12, k, F}\lV\lv \Phi_{q+k}\rv\rV_F F\lmk \dist(x,y)\rmk
\end{split}
\end{align}
with some constant $\tilde C_{\theta}$, for each $0<\theta<1$.
In the last line we used (\ref{ugugug}) and Lemma 8.9 of \cite{NSY}.
Hence we get $\Psi_k^{{(n)}}, \Psi_k\in \hat \caB_{\tilde F_{\theta}}([0,1])$.
\end{proof}

\begin{proofof}[Theorem \ref{mo}]
Suppose $\Phi_{0}\sim \Phi_{1}$ via a path $\Phi$.
Our definition of $\Phi_0\sim\Phi_1$ means the existence of
a path of interactions satisfying Assumption 1.2 of \cite{MO}.
Therefore, Theorem 1.3 of \cite{MO}
 guarantees the existence of a path of quasi-local automorphism $\alpha_t$
 satisfying
 $\omega_{\Phi_1}=\omega_{\Phi_0}\circ\alpha_1$.
 From the proof in \cite{MO}, 
% we can see that this  path of quasi-local automorphism $\alpha_s$
% is the same one as $\tau_{s,0}^{\Psi}$ in \cite{NSY} Theorem 6.1.4., given for our path of interactions
% $\Phi$ in Definition \ref{classsym}.
 the automorphism $\alpha_t$ is given by a family of interactions
 \begin{align}\label{lalala}
\Psi\lmk
Z, t
\rmk:=
\sum_{m\ge 0} \sum_{X\subset Z,\; X(m)=Z}
\Delta_{X(m)}\lmk
\caK_t\lmk
 \dot \Phi\lmk X; t\rmk
\rmk
\rmk,\quad Z\in {\mathfrak S}_{\bbZ^2},\quad t\in [0,1],
\end{align}
with 
\begin{align}\label{kitty}
\caK_t(A):=-\int du W_\gamma(u)\tau_{u}^{\Phi(t)}(A),
\end{align}
as $\alpha_{t}=\tau_{t,0}^{\Psi}$.
(Note that by partial integral of (1.19) of \cite{MO}, we obtain (6.103) of \cite{NSY} with function
 $W_\gamma$ in (6.35) of \cite{NSY}).)
 The interaction $\Psi$ actually belongs to $\hat \caB_{F_3}([0,1])$
 for some $F_3\in\caF_a$.
 To see this, note that the path $\Phi$ in Definition \ref{classsym} 
 satisfy Assumption 6.12 of \cite{NSY}
 for any $F$-function because
 \begin{align}
 \sum_{\substack{X\in{\mathfrak S}_{\bbZ^2}\\
 X\ni x,y
 }}\lmk
\lV
\Phi\lmk X;s\rmk
\rV+|X|\lV
\dot{\Phi} \lmk X;s\rmk
\rV
\rmk
% \le 2^{(2R+1)^2}
%\sup_{s\in[0,1]}\sup_{\substack{X\in {\mathfrak S}_{\bbZ^2}}}
%\lmk
%\lV
%\Phi\lmk X;s\rmk
%\rV+|X|\lV
%\dot{\Phi} \lmk X;s\rmk
%\rV
%\rmk
\le 
\frac{ 2^{(2R+1)^2}C_b^\Phi}{F(R)} F(\dist(x,y)),
\end{align}
with $C_b^\Phi$, $R$, given in {\it 3, 4} of Definition \ref{classsym}.
In particular, it satisfies Assumption 6.12 of \cite{NSY}, with
respect to the $F$-function (see section 8 of \cite{NSY})
$
F_1(r):=\frac{e^{-r}}{(1+r)^{4}}
$. By section 8 of \cite{NSY}, $F_1$ belongs to $\caF_a$.
%Fix some $0<\theta<1$. We may choose $\tilde F_\theta$ in (\ref{tildef})
%as 
%$F_2(r):=\tilde F_\theta(r):=C_1\frac{e^{-\frac{\theta}3r}}{(1+r)^{3}}
%$
%with a suitable constant $C_1$.
Fix any $0<\alpha<1$.
Then by  Proposition 6.13 and its proof of \cite{NSY}, the family of maps given by (\ref{kitty})
((6.102) of \cite{NSY}) satisfies Assumption \ref{nsy5.15}, with $p=0$, $q=1$, $r=1$
and $G=G_{F_{2}}$, where $F_{2}(r)=(1+r)^{-4} \exp\lmk - r^{\alpha}\rmk$.
Furthermore, we have $\dot{\Phi}_m\in \hat\caB_{F_2}([0,1])$ for any $m\in\bbN$,
because
\begin{align}
\lV\lv \dot \Phi_m\rV\rv_{F_2}:=
\sup_{x,y\in{\bbZ^2}}\frac{1}{F_2\lmk {\dist}(x,y)\rmk}\sum_{Z\in{\mathfrak S}_{\bbZ^2}, Z\ni x,y}
\sup_{t\in [0,1]}\lv Z\rv^m \lmk \lV\dot\Phi(Z;t)\rV\rmk
\le \frac{2^{(2R+1)^2} (2R+1)^{2m}C_b^\Phi}{F_2(R)}<\infty.
\end{align}
We have $F_{2}\in \caF_{a}$, and fixing any $0<\alpha'<\alpha$,
 $\tilde F_{2}(r):= (1+r)^{-4} \exp\lmk - r^{\alpha'}\rmk$
 satisfy 
\begin{align}
\max\left\{ F_{2}\lmk\frac r 3\rmk, \lmk  F_{2}\lmk \lcm \frac r 3 \rcm \rmk\rmk^{\theta}\right\}\le
C_{2,\theta,\alpha'}\tilde F_{2}(r),\quad\quad  r\ge 0,
\end{align}
for a suitable constant $C_{2,\theta,\alpha'}$.

Therefore, by Theorem \ref{nsy5.17}, 
$\Psi$ given by (\ref{lalala}) for this $\caK_t$ and $\dot\Phi$
satisfy $\Psi_1,\Psi\in\hat \caB_{\tilde F_2}([0,1])$
 for  $\tilde F_2\in\caF_a$ above.
 
 If $\Phi$ is $\beta_{g}$-invariant, then $\tau^{\Phi(t)}$ commutes with
 $\beta_{g}$, hence $\caK_{t}$ commutes with $\beta_{g}$.
 As $\Pi_{X}$ commutes with $\beta_{g}$ and $\dot \Phi$ is $\beta_{g}$-invariant, we see that 
 $\Psi$ is $\beta_{g}$-invariant.

\end{proofof}

\begin{prop}\label{tokyo1877}
Let  $F,\tilde F\in \caF_a$ be $F$-functions
of the form 
$F(r)=(1+r)^{-4} \exp\lmk - r^{\theta}\rmk$,
 $\tilde F(r):=(1+r)^{-4} \exp\lmk - r^{\theta'}\rmk$
 with some constants $0<\theta'<\theta<1$.
Let $\Psi,\tilde\Psi\in\caB_{F}([0,1])$ be a path of interactions such that
$\Psi_{1}\in \caB_{F}([0,1])$.
Finally,  let $\tau_{t,s}^{\tilde \Psi}$ and $\tau_{t,s}^{(\Lambda_n),{\tilde \Psi}}$ be 
automorphisms given  by $\Psi,\tilde\Psi$ from  Theorem~\ref{tni}.

Then, with $s \in [0,1]$,
the right hand side of the following sum
\begin{align}\label{eq:psis}
\Xi^{(s)}\lmk
Z, t
\rmk:=
\sum_{m\ge 0} \sum_{X\subset Z,\; X(m)=Z}
\Delta_{X(m)}\lmk
\tau_{t,s}^{\tilde \Psi}\lmk
 \Psi\lmk X; t\rmk
\rmk
\rmk,\quad Z\in{\mathfrak S}_{\bbZ^2},\quad t\in [0,1]
\end{align}
defines a path of interaction such that $\Xi^{(s)}\in\caB_{\tilde F}([0,1])$.
 Furthermore, the formula 
 \begin{align}\label{eq:psisn}
	 \Xi^{(n)(s)}\lmk
Z, t
\rmk:=
\sum_{m\ge 0} \sum_{X\subset Z, X(m)\cap\Lambda_{n}=Z}
\Delta_{X(m)}\lmk
\tau_{t,s}^{(\Lambda_n), \tilde \Psi}\lmk
 \Psi\lmk X; t\rmk
\rmk
\rmk
\end{align}
defines $\Xi^{{(n)(s)}}\in \caB_{\tilde F}([0,1])$ such that
$\Xi^{(n)}\lmk
Z, t
\rmk=0$ unless $Z\subset \Lambda_{n}$,
and satisfies
\begin{align}\label{psio}
\tau_{t,s}^{(\Lambda_n), \tilde \Psi} \lmk H_{\Lambda_n, \Psi}(t)\rmk
=H_{\Lambda_n, \Xi^{(n)(s)}}(t).
\end{align}
For any $t,u\in[0,1]$, we have
\begin{align}\label{convconv}
\lim_{n\to\infty}\lV
\tau_{t,u}^{\Xi^{(n)(s)}}\lmk A\rmk
-\tau_{t,u}^{\Xi^{(s)}}\lmk A\rmk
\rV=0,\quad A\in\caA.
\end{align}
Furthermore, if $\Psi_1\in\hat \caB_{F}([0,1])$, then we have
$\Xi^{{(n)(s)}}, \Xi^{{(s)}}\in \hat \caB_{\tilde F}([0,1])$.
\end{prop}
\begin{proof}
From Theorem \ref{nsy5.17}, it suffices to show that
the family $\{\caK_t:=\tau_{t,u}^{\tilde\Psi}\}$
satisfies the Assumption \ref{nsy5.15}.
This follows from Theorem \ref{tni}.
\end{proof}


\begin{thebibliography}{}




\bibitem[BMNS]{bmns}
S.~ Bachmann, S.~Michalakis, B.~Nachtergaele, and R.~Sims.
\newblock{Automorphic Equivalence within Gapped Phases of Quantum Lattice Systems.}\newblock{Communications in Mathematical Physics}
{\bf 309}, 835--871, 2012. 




\bibitem[BR1]{BR1}
 O.~Bratteli and D.~W.~Robinson.
\newblock {\em Operator Algebras and Quantum Statistical 
 Mechanics 1.} Springer-Verlag. (1986).
 \bibitem[BR2]{BR2}
 O.~Bratteli and  D.~W.~Robinson.
 \newblock {\em Operator Algebras and Quantum Statistical 
 Mechanics 2.} Springer-Verlag. (1996).
 
 
 \bibitem[CGLW]{cglw}
 X.~Chen, Z.~C. Gu, Z.~X. Liu, and X.~G. Wen,
 \newblock{Symmetry protected topological orders and the group cohomology of their
symmetry group.}
 Phys. Rev. B {\bf 87}, 155114 (2013)

 \bibitem[C]{Connes}
A.~Connes,
\newblock{Periodic automorphisms of the hyperfinite factor of type $II_1$.}
\newblock{Acta Sci. Math. (Szeged)} {\bf 39} (1977), no. 1--2, 39--66.

 
 \bibitem[DW]{DW}
R.~Dijkgraaf, E.~Witten, 
\newblock{Topological gauge theories and group cohomology.}
Commun.Math. Phys.{\bf 129}, 393--429 (1990)

\bibitem[J]{jones}
V.~Jones.
\newblock{Actions of finite groups on the hyperfinite type $II_{1}$ factor.}
Mem. Amer. Math. Soc. 28 (1980), no. 237.

\bibitem[KOS]{kos} A.~Kishimoto, N.~Ozawa, and S.~Sakai.
\newblock{Homogeneity of the pure state space of a separable C*-algebra}.
\newblock{Canad. Math. Bull.} {\bf 46} 365--37 (2003).


\bibitem[MM]{MillerMIyake2016}
J.~Miller and A.~Miyake, \newblock{ Hierarchy of universal entanglement in 2D
  measurement-based quantum computation}, 
  \newblock{Quantum Information }{\bf 2},
  16036 (2016). 





\bibitem[MGSC]{molnar}
A.~Molnar, Y.~Ge, N.~Schuch, and J.~I.~Cirac
\newblock{A generalization of the injectivity
condition for projected entangled pair
states}
J. Math. Phys. {\bf 59}, 021902 (2018).
\bibitem[MO]{MO}A.~Moon and Y.~Ogata
\newblock{Automorphic equivalence within gapped phases in the bulk}
\newblock{Journal of Functional Analysis : {\bf 278}, Issue 8, 1 108422 (2020)}

\bibitem[NO]{NO}
P.~Naaijkens, Y.~Ogata
\newblock{Stability of the distal split property and absence of superselection sectors in 2D systems}
In preparation.

\bibitem[NSY]{NSY} B. Nachtergaele, R. Sims, and A. Young. 
\newblock{ Quasi-locality bounds for quantum lattice systems. I. Lieb-Robinson bounds, quasi-local maps, and spectral flow automorphisms}. J. Math. Phys. \textbf{60}, 061101 (2019)

\bibitem[O1]{ogata}
\newblock{A $\bbZ_{2}$-Index of Symmetry Protected Topological Phases with Time Reversal Symmetry for Quantum Spin Chains.}
 Commun. Math. Phys. {\bf 374}, 705--734 (2020). 
 \bibitem[O2]{IAMP}
\newblock{One World Mathematical Physics Seminar 15. Dec. 2020}
 \url{https://youtu.be/cXk6Fk5wD_4}
 \newblock {Theoretical studies of topological phases of matter 17. Dec 2020}
 \url{https://www.ms.u-tokyo.ac.jp/%7Eyasuyuki/yitp2020x.htm}
\newblock{ Current Developments in Mathematics 4th January 2021} 
\url{https://www.math.harvard.edu/event/current-developments-in-mathematics-2020/}
\bibitem[P]{powers}
R.~T.~Powers.
\newblock{ Representations of uniformly hyperfinite algebras and their associated von Neumann rings.}
\newblock{ Ann. of Math.}
{\bf  86} 138--171 (1967).


\bibitem[T]{takesaki}
M.~Takesaki.
\newblock{Theory of operator algebras. I.}
\newblock{\em{Encyclopaedia of Mathematical Sciences}.} Springer-Verlag. (2002).

\bibitem[Y]{Beni2016}
B.~Yoshida, \newblock{Topological phases with generalized global symmetries}\/,
  Phys. Rev. B {\bf 93}, 155131 (2016).

\end{thebibliography}
\end{document}